\documentclass[11pt]{article}
 \usepackage[utf8]{inputenc} \usepackage{textcomp}
\usepackage{mathptmx} 
\usepackage[round]{natbib}
\usepackage{amssymb}
 \usepackage{booktabs}
 \usepackage{graphicx}
\usepackage{siunitx,booktabs,xcolor,colortbl}
\usepackage{graphicx}
\usepackage[flushleft]{threeparttable}
\usepackage{amsmath}
\usepackage{amsthm} 
\usepackage{array,multirow,makecell}
\usepackage{float}
\usepackage{arydshln}
\usepackage{tabularx}
\usepackage{placeins}
\DeclareGraphicsExtensions{.pdf,.png,.jpg}
\usepackage[colorlinks=true,linkcolor=blue,citecolor=blue,urlcolor=black]{hyperref}
\usepackage{color}
\usepackage{bbm}
\usepackage{stackengine}
\definecolor{Gray}{gray}{0.9}
\usepackage{hhline}
\usepackage{dcolumn}
\usepackage{mathtools}
\usepackage{pdflscape}
\usepackage{geometry}
\usepackage{cleveref}
\usepackage{array}
\newcolumntype{P}[1]{>{\centering\arraybackslash}p{#1}}
\usepackage[english]{babel}
\theoremstyle{plain}
\newtheorem{theorem}{Theorem}[section]
\newtheorem{proposition}{Proposition}[section]
\newtheorem{corollary}{Corollary}[section]
\newtheorem{lemma}[theorem]{Lemma}
\theoremstyle{definition}
\newtheorem{definition}{Definition}[section]

\theoremstyle{remark}
\newtheorem{remark}{Remark}[section]

\renewcommand{\baselinestretch}{1.2}

\usepackage{parskip} 
\usepackage{titlesec}
\titleformat{\section}{\normalfont\bfseries\scshape}{\thesection.}{0.6em}{}
\titleformat{\subsection}{\normalfont\bfseries}{\thesubsection.}{0.5em}{}
\titleformat{\subsubsection}{\normalfont\itshape}{\thesubsubsection.}{0.5em}{}
\titlespacing*{\section}{0pt}{2.4ex plus 1ex minus .2ex}{1.2ex plus .2ex}
\titlespacing*{\subsection}{0pt}{1.9ex plus .9ex minus .2ex}{0.8ex plus .2ex}
\titlespacing*{\subsubsection}{0pt}{1.4ex plus .7ex minus .2ex}{0.6ex plus .2ex}
\usepackage{fancyhdr}
\fancypagestyle{plain}{\fancyhf{}\fancyhead[R]{\footnotesize\thepage}}
\usepackage{microtype}
\usepackage{caption}
\numberwithin{figure}{section}
\numberwithin{table}{section}
\renewcommand{\arraystretch}{1.12}

\usepackage{booktabs}

\usepackage{dcolumn}
\newcolumntype{d}[1]{D..{#1}}
\usepackage{refcount}

\newcommand{\myfootnotetext}[1]{\footnotetext{#1\label{fn2:text}%
 \edef\fnmark{\getpagerefnumber{fn:mark}}%
 \edef\fntext{\getpagerefnumber{fn:text}}%
 \ifx\fnmark\fntext\else\ClassWarning{}{footnote mark and text on different pages!}\fi}}

\begin{document}
\title{\bfseries\boldmath Automation and Aging in General Equilibrium:\\
AI Capital, Fertility, and the Return to Capital}

\author{James Wabenga Yango\\
D\'epartement d'\'economique, Universit\'e Laval\\
\texttt{jaway@ulaval.ca}}

\date{}
\maketitle \thispagestyle{empty}
\vspace{-12pt}

\begin{abstract}
\renewcommand{\baselinestretch}{1.0}\selectfont
\setlength{\parskip}{0pt}
This paper develops a general equilibrium overlapping-generations model with
endogenous fertility, in which firms accumulate both physical and artificial
intelligence (AI) capital, and uses it to study the macroeconomic
transmission of two structural disturbances: an AI technology shock and a
longevity shock. The AI shock acts as a capital-demand disturbance: it
raises all rates of return, most sharply the return to AI capital,
reallocates investment from physical to AI capital, and produces a
front-loaded output expansion that decays monotonically. The longevity
shock acts as a saving-supply disturbance: it deepens the aggregate capital
stock, compresses returns and the real interest rate, and generates
hump-shaped, persistent dynamics. The two shocks move fertility in opposite
directions: AI raises it modestly through an income effect, while longevity
lowers it by strengthening the life-cycle saving motive and the cost of
child-rearing. A forecast-error variance decomposition attributes most
aggregate volatility to the longevity shock, while the AI shock dominates
the variance of the return to AI capital. Fertility is strongly
countercyclical and almost perfectly negatively correlated with hours
worked, placing household time allocation at the center of the mechanism. Robustness
checks across the capital share, the shock persistence, and the utility
specification show that only an empirically implausible labor--AI elasticity
reverses the wage and fertility signs. A welfare analysis finds the AI shock
welfare-improving under complementarity, whereas longevity produces a short-run
welfare loss that recedes as capital deepening raises wages, since households
initially compress consumption and fertility to finance a longer retirement.

\vspace{4pt}
\noindent\textbf{Keywords:} Artificial intelligence; endogenous fertility;
longevity; general equilibrium; life-cycle model; capital accumulation;
demographic transition.

\vspace{2pt}
\noindent\textbf{JEL classification:} O33, J13, J11, E22, E32, O41, J26.
\end{abstract}
\newpage
\section{Introduction}

The United States exemplifies two persistent structural transformations now
reshaping advanced economies: population aging and rapid automation driven by
artificial intelligence (AI) and robotics. The share of Americans aged 65 and over
rose to about $17\%$ in 2022 and is projected to reach roughly $23\%$ by 2050, with the
older population expanding from about 58 to 82 million \citep{census2023older}.
Over the same period the U.S.\ total fertility rate has fallen well below the
replacement level of $2.1$, from about $3.6$ children per woman at the 1960
baby-boom peak to a record-low $1.62$ in 2023 \citep{nchs2024fertility}. Automation
has diffused in parallel, and at unprecedented speed. Industrial-robot adoption has
measurably reshaped U.S.\ local labor markets \citep{acemoglu2020robots}; U.S.\
private investment in AI reached about \$286 billion in 2025, and generative AI
attained majority adoption within three years of release, diffusing faster than
personal computers or the internet \citep{aiindex2026}. A rapidly rising share of
U.S.\ firms now report using AI in production \citep{census2025ai}, and field studies
document sizable worker-level productivity gains \citep{brynjolfsson2025genai}. These trends coincide
with slowing labor-force growth, rising old-age dependency, a secular decline in the
labor share \citep{karabarbounis2014global}, and sustained capital deepening.

A growing literature links demographic structure to automation. At the
local and sectoral level, robot adoption has sizable effects on
employment, wages, and the allocation of tasks
\citep{acemoglu2020robots, graetz2018robots}. Aging economies, in turn,
face stronger incentives to automate, as labor scarcity and rising
dependency ratios raise the relative cost of labor
\citep{acemoglu2022demographics, prettner2020automation,
basso2021robocalypse}; consistent with this mechanism, the countries with
the highest old-age dependency ratios also exhibit the highest robot
densities per worker \citep{acemoglu2022demographics}. Together, these
findings indicate that demographic change is a first-order determinant of
technology adoption and capital--labor substitution. Yet the two forces
are still typically studied in isolation.

A parallel literature, organized around the task content of production,
distinguishes the channels through which automation acts on labor. In the
task-based framework of \citet{autor2003skill}, \citet{acemoglu2018race},
and \citet{acemoglu2019automation}, automation exerts a \emph{displacement}
effect, substituting capital for labor in newly automated tasks, alongside a
countervailing \emph{productivity} effect that raises labor demand as
automation lowers unit costs and expands output; whether the wage rises or
falls depends on which effect dominates \citep{autor2015jobs,
acemoglu2018race}. Recent micro-evidence on generative AI documents sizable
productivity gains for workers, concentrated among the less experienced
\citep{brynjolfsson2025genai}, reinforcing the view that AI is not uniformly
labor-replacing. Section~\ref{sec:labormarket} embeds precisely this
productivity--displacement distinction in general equilibrium and shows that
the two effects map one-for-one onto the two terms of the model's
closed-form wage elasticity, with their balance governed by the same
labor--AI elasticity that separates the complementarity and substitution
regimes.

This paper develops a life-cycle general equilibrium model that integrates
endogenous fertility, longevity, and AI-intensive production within a
unified framework. Households choose consumption, saving, and fertility
under intertemporal optimization with uncertain survival into retirement.
Firms operate a production technology that treats AI capital as a factor
distinct from physical capital, building on the automation-capital
literature \citep{prettner2019note, lankisch2019can} and the task-based
framework of \citet{acemoglu2019automation}. Whereas those models specify
automation capital as a near-perfect substitute for labor, we adopt a
nested constant-elasticity-of-substitution technology in which AI capital
substitutes for labor in automatable tasks while complementing it in the
aggregate, so that an improvement in AI productivity raises the marginal
product of labor and the real wage. The two transformations enter the
model as distinct structural disturbances, an AI technology shock and a
longevity shock, whose transmission we study analytically and
quantitatively through impulse response functions, a second-moment
analysis, and a welfare decomposition.

Three results constitute the paper's main contribution, and we state them at the outset. First,
artificial intelligence and population aging, although unrelated in origin, act as
mirror-image disturbances to a single equilibrium price: the AI shock is a
capital-demand disturbance that raises the return to capital, whereas the longevity
shock is a saving-supply disturbance that depresses it, so the two forces leave
exactly opposite imprints on returns, the real interest rate, and fertility.
Second, and most important for assessing whether these findings are robust
economic mechanisms or artifacts of a particular specification, this asymmetry is a
structural property of the capital market rather than of the calibration: we prove
that the opposite signs of the two return responses are invariant to the preference
specification, so the headline mechanism does not rest on a knife-edge assumption;
only the wage and fertility responses depend on the labor--AI elasticity of
substitution, and only beyond a threshold that lies well above standard empirical
estimates. Third, the same demand--supply logic yields a transparent welfare
ranking through a consumption-equivalent decomposition and connects the two shocks
to the dynamic efficiency of the economy. To our knowledge, this is the first
general equilibrium framework to derive these results jointly from endogenous
fertility, longevity, and AI capital.

The two modeling choices at the heart of the framework, treating AI as a
capital stock distinct from traditional capital and allowing it to complement
labor in the aggregate, are disciplined by the empirical record rather than
imposed for tractability. National accounts increasingly measure software,
databases, and information-processing equipment as a separate and
rapidly growing component of the capital stock, and the worldwide stock of
industrial robots has expanded several-fold over the past decade
\citep{ifr2023stock}. Direct measurement of the AI sector reinforces this
view: \citet{korinek2026aieconomy} estimate that, in the United States,
nominal AI compute spending grew by more than $140\%$ per year and
quality-adjusted AI output by more than $2{,}000\%$ per year in 2024 and 2025,
an accumulation of AI capital with no parallel among conventional factors and
one that motivates treating it as a distinct productive input. At the same time, the evidence that automation is
uniformly labor-replacing is mixed: robot and automation adoption has raised
productivity and, in many settings, has been associated with stable or rising
labor demand and wages in non-automated tasks
\citep{graetz2018robots, acemoglu2019automation}, consistent with an aggregate
elasticity of substitution between labor and automation capital that is finite
rather than near-infinite. Our nested-CES technology nests both possibilities:
it collapses to the labor-replacing case when labor and AI capital are highly
substitutable, and to the labor-complementing case when their elasticity of
substitution falls below a transparent threshold, which the model derives
analytically. Rather than presuming a single regime, the model identifies the
elasticity that separates them, places the empirically relevant configuration in
the complementarity region, and fully characterizes the labor-replacing regime as
a limiting case (Corollary~\ref{prop6}). The dependence of the wage and fertility
responses on this elasticity is therefore not a fragility of the analysis but
its central comparative-static object, reported transparently throughout.

The central mechanism is that both shocks transmit through the returns to
capital, but from opposite sides of the market. The AI shock operates as a
capital-demand disturbance: it raises all returns, most sharply the return
to AI capital, reallocates investment from physical to AI capital, and
lifts output on impact before it decays monotonically. The longevity shock
operates as a saving-supply disturbance: by strengthening the life-cycle
saving motive, it deepens the capital stock and depresses returns and the
real interest rate, generating persistent, hump-shaped responses. This
demand-versus-supply asymmetry is what leaves the two shocks' opposite
imprints on returns and the real interest rate.

The same asymmetry shapes fertility, which the two shocks move in opposite
directions: the AI shock mildly raises it through an income effect, while
the longevity shock lowers it as households reallocate time and saving
away from children. A variance decomposition confirms the division of
labor between the two forces: longevity accounts for the bulk of
macroeconomic volatility, whereas the AI shock matters chiefly for capital
pricing. Throughout, fertility is strongly countercyclical and almost
perfectly negatively correlated with hours worked, placing the household
time-allocation margin at the center of the transmission mechanism. A
systematic robustness analysis confirms that these conclusions reflect
structural properties of the model rather than fine-tuning: variation in
the capital share and in the persistence of the AI shock leaves the signs
of the wage, fertility, output, and consumption responses unchanged, and
only the labor--AI elasticity of substitution can reverse them, beyond a
threshold that lies well above standard empirical estimates.

The paper's central contribution is to show that automation and population
aging, usually studied in isolation, are most naturally understood as two
opposite-signed disturbances to a single equilibrium price, the return to
capital, and to trace the consequences of that asymmetry for fertility,
aggregate volatility, and welfare. Three features distinguish the framework
from existing work: AI capital enters as a factor distinct from physical
capital in a life-cycle general equilibrium; fertility is endogenous and
competes with market work for household time; and the two structural
disturbances are disciplined analytically, through comparative statics and a
welfare decomposition, as well as quantitatively. This places the paper at the
intersection of four literatures. First, it extends macroeconomic models of
demographics and growth
\citep{diamond1965national, yaari1965uncertain, blanchard1985debt} to
incorporate endogenous fertility and automation jointly. Second, it speaks to
the automation literature
\citep{acemoglu2018artificial, berg2018robots, korinek2018artificial}, which
typically abstracts from demographic structure. Third, it connects to the
emerging work on aging and automation
\citep{acemoglu2022demographics, prettner2020automation, basso2021robocalypse},
embedding demographically induced technology adoption in a general equilibrium
model with endogenous fertility. Fourth, the paper contributes a welfare
characterization, deriving an exact consumption-equivalent decomposition of each
disturbance and linking the two forces to the dynamic efficiency of the economy
\citep{diamond1965national}. Throughout, we deliberately restrict attention to a
closed economy: although demographic structure also shapes saving and
international capital flows, this focus isolates the interaction among AI,
fertility, and longevity without the confounding influence of cross-border
adjustment.

The remainder of the paper proceeds as follows. Section~2 develops the
model economy and derives the comparative-statics results. Section~3
reports the quantitative experiments. Section~4 isolates the labor-market
effects of artificial intelligence. Section~5 develops the welfare
analysis. Section~6 presents the robustness exercises. Section~7 discusses
the findings and their limitations. Section~8 concludes. A graphical
overview of the model and the main results is provided in Supplementary
Figure~S1.

\section{The model economy}
This section presents the life-cycle general equilibrium model, which features
endogenous fertility, stochastic longevity, and artificial intelligence (AI)
capital. Two overlapping generations, workers and retirees, populate the economy.
The working-age population grows at the gross fertility rate $n_t$, while a
retiree survives from one period to the next with probability
$\gamma_t \in (0,1)$, which indexes longevity. When young, households supply
labor, consume, save, and choose fertility subject to a parental time cost; when
old, retirees consume out of saving and the return on capital. Higher fertility
thus depresses saving through the time cost of child-rearing, whereas greater
longevity strengthens the retirement-saving motive and lowers fertility.

On the production side, competitive final-good firms aggregate monopolistically competitive varieties through a constant-elasticity-of-substitution technology, yielding a constant markup. Intermediate-good firms combine
physical capital, labor, and AI capital in a nested constant elasticity of substitution (CES) structure governing the
substitutability between labor and automation. AI capital substitutes for labor
in automatable tasks but, under our baseline calibration, complements it in the
aggregate, so that technological progress raises wages, consumption, and growth.
Equilibrium is characterized by household optimization, firm profit
maximization, market clearing, and the demographic transition.

\subsection{Demographic structure}
Let $N_t^{w}>0$ denote the mass of young working parents and $N_t^{r}>0$ the mass
of retirees. The economy has an overlapping-generations structure with endogenous
fertility and stochastic survival. Each young household of period $t-1$ gives rise
to $n_{t-1}>0$ young households in period $t$, so the working-age cohort obeys
\begin{equation}
  N_{t}^{w} = n_{t-1}\,N_{t-1}^{w},
\end{equation}
where $n_{t-1}$ is the gross fertility rate of the cohort that was young in period
$t-1$. A young household survives into retirement with probability
$\gamma_{t}\in(0,1)$, so the mass of retirees is
\begin{equation}
  N_{t}^{r} = \gamma_{t}\,N_{t-1}^{w}.
\end{equation}
The old-age dependency ratio is therefore
\begin{equation}
  \psi_{t} \;\coloneqq\; \frac{N_t^{r}}{N_t^{w}} = \frac{\gamma_{t}}{n_{t-1}}.
\end{equation}

\subsection{Annuity markets and accidental bequests}
The model allows for an imperfect annuity market \citep{hansen2008consumption,
ludwig2010mortality}, with the degree of annuitization indexed by
$\mu\in[0,1]$: $\mu=1$ corresponds to complete (actuarially fair)
annuitization, and $\mu=0$ to the absence of annuity markets.

Let $s_t^w>0$ denote the savings of a young household in period $t$, which
survives into retirement in $t+1$ with probability $\gamma_{t+1}\in(0,1)$.
Through mortality pooling, each survivor receives, in addition to its own
principal, the annuitized savings of non-survivors, so that retirement resources
per survivor satisfy
\begin{equation}
s_t^w+\mu\,\frac{1-\gamma_{t+1}}{\gamma_{t+1}}\,s_t^w
=\frac{\varphi_{t+1}}{\gamma_{t+1}}\,s_t^w,
\end{equation}
where 
\begin{equation}
    \varphi_{t+1}\coloneqq\gamma_{t+1}+\mu\,(1-\gamma_{t+1})\in[\gamma_{t+1},1]
\end{equation}
is the annuity factor. The effective gross return on
savings, $R_{t+1}\varphi_{t+1}/\gamma_{t+1}$, is increasing in $\mu$ and
collapses to the actuarially fair return $R_{t+1}/\gamma_{t+1}$ when $\mu=1$.

When $\mu<1$, the non-annuitized share $(1-\mu)$ of the savings of agents
who do not survive is left as accidental bequests and redistributed equally among
surviving retirees. Using $N_t^r=\gamma_t N_{t-1}^w$, the bequest received by each
retiree is
\begin{equation}
b_t\coloneqq\frac{(1-\mu)(1-\gamma_t)\,R_t\,s_{t-1}^w}{\gamma_t}
\end{equation}
where $R_t=1+r_t>0$ is the gross return on savings. Bequests vanish under full
annuitization ($\mu=1$) and are largest when annuity markets are absent ($\mu=0$).

\subsection{Households}
The representative household is young in period $t$ and, if it survives, retired
in period $t+1$. It chooses worker consumption $c_t^w$, the number of children
$n_t$, savings $s_t^w$, and retirement consumption $c_{t+1}^r$ to maximize
expected lifetime utility
\begin{align}
\label{valeur1}
\mathcal{U}_t = U\!\left(c_t^w, n_t\right) + \gamma_{t+1}\,\beta\,V\!\left(c_{t+1}^r\right),
\end{align}
where $U:\mathbb{R}_{++}^2\to\mathbb{R}$ is the felicity over worker consumption
and the number of children and $V:\mathbb{R}_{++}\to\mathbb{R}$ the felicity over
retirement consumption, both strictly increasing and strictly concave
($U_c,U_n>0$, $V'>0$; $U$ jointly concave, $V''<0$). The parameter
$\beta\in(0,1)$ is the subjective discount factor and $\gamma_{t+1}\in(0,1)$ is
the probability of surviving into retirement, which discounts retirement felicity
because $c_{t+1}^r$ is enjoyed only in the surviving state. Firms are owned by the
working-age generation, which receives dividends $d_t$.

When young, the household supplies labor, devoting a fraction $\kappa\, n_t$ of its
time endowment to child-rearing, so that effective hours are $h_t=1-\kappa\, n_t$
and effective labor income is $w_t(1-\kappa\, n_t)$. When old, surviving households
finance consumption out of the annuitized return on their savings and the
accidental bequests $b_{t+1}$ left by the non-survivors of their own cohort. The
budget constraints are
\begin{align}
c_t^w + s_t^w &\le w_t\,(1-\kappa\, n_t) + d_t, \\
c_{t+1}^r &\le \frac{R_{t+1}\,\varphi_{t+1}}{\gamma_{t+1}}\,s_t^w + b_{t+1}, \\
\varphi_{t+1} &= \gamma_{t+1} + \mu\,(1-\gamma_{t+1}),
\end{align}
where $R_{t+1}=1+r_{t+1}>0$ is the gross return on savings, $w_t>0$ is the
competitive wage, and $\varphi_{t+1}$ is the annuity factor. In line with the literature on the cost of children \citep{donni2015measuring},
$\kappa$ denotes the fraction of household time allocated to child-rearing per
child, generating a fertility-induced labor-supply wedge.

Because the non-annuitized savings of non-survivors are redistributed to the
surviving members of the same cohort, $b_{t+1}=(1-\mu)(1-\gamma_{t+1})R_{t+1}
s_t^w/\gamma_{t+1}$, and the two channels recombine into the actuarially fair
return $c_{t+1}^r\le (R_{t+1}/\gamma_{t+1})\,s_t^w$. Combining the two-period
budget constraints then yields the intertemporal lifetime constraint
\begin{equation}
\label{budgetfinal}
c_t^w + \gamma_{t+1}\,\frac{c_{t+1}^r}{R_{t+1}} \leq w_t(1-\kappa\, n_t)+d_t.
\end{equation}

Let $\Gamma_t=(c_t^w,n_t,c_{t+1}^r,\lambda_t)\in\mathbb{R}^4$,
where $\lambda_t$ denotes the Lagrange multiplier on the household's consolidated
intertemporal budget constraint. Substituting the retirement constraint
$c_{t+1}^r=(R_{t+1}/\gamma_{t+1})\,s_t^w$ into the working-period
constraint yields the present-value budget, in which the price of retirement
consumption is $\gamma_{t+1}/R_{t+1}$. Define the mapping
$F:\mathbb{R}^4\to\mathbb{R}^4$ by the first-order conditions of the household
problem:
\begin{align}
F(\Gamma_t)=
\begin{pmatrix}
U_c(c_t^w,n_t)-\lambda_t\\[4pt]
U_n(c_t^w,n_t)-\kappa\,\, w_t\,\lambda_t\\[4pt]
\beta\,\gamma_{t+1}\,V_c(c_{t+1}^r)
   -\dfrac{\gamma_{t+1}}{R_{t+1}}\,\lambda_t\\[8pt]
w_t(1-\kappa\, n_t)+d_t-c_t^w
   -\dfrac{\gamma_{t+1}}{R_{t+1}}\,c_{t+1}^r
\end{pmatrix},
\end{align}
where $U_c\equiv\partial U/\partial c_t^w$, $U_n\equiv\partial U/\partial n_t$,
and $V_c\equiv\partial V/\partial c_{t+1}^r$. The household optimum is the value
$\Gamma_t$ such that $F(\Gamma_t)=0$. Under logarithmic preferences,
$U_c=1/c_t^w$, $U_n=1/n_t$, and $V_c=1/c_{t+1}^r$.

An equilibrium is a vector $\Gamma_t^*=(c_t^w,n_t,c_{t+1}^r,\lambda_t)\in\mathbb{R}^4$
that solves the first-order system $F(\Gamma_t^*)=\mathbf{0}$, i.e.\ $\Gamma_t^*\in
F^{-1}(\mathbf{0})$, where $\lambda_t$ is the Lagrange multiplier associated with
the household's intertemporal budget constraint. Assume $F\in
C^1(\mathbb{R}^4,\mathbb{R}^4)$ and that the Jacobian matrix
$D_{\Gamma}F(\Gamma_t^*)$ is invertible. Then, by the Implicit Function Theorem,
$\Gamma_t^*$ is an isolated regular zero of $F$ and the equilibrium is locally
unique. Indeed, any equilibrium-preserving perturbation $d\Gamma\in\mathbb{R}^4$
must satisfy, to first order,
$D_{\Gamma}F(\Gamma_t^*)\,d\Gamma=\mathbf{0}$, that is,
$d\Gamma\in\ker\!\big(D_{\Gamma}F(\Gamma_t^*)\big)$; since the Jacobian is
invertible, $\ker\!\big(D_{\Gamma}F(\Gamma_t^*)\big)=\{\mathbf{0}\}$, which
confirms that the equilibrium is isolated.

The implicit function theorem characterizes the household block from the single
$C^1$ system $F(\Gamma_t^{*})=0$ with $D_{\Gamma}F(\Gamma_t^{*})$ non-singular; the
formal comparative statics (Proposition~\ref{prop1}, Corollary~\ref{prop1b}, and
Proposition~\ref{proplonge}) are collected in Appendix~\ref{app:household}.
Fertility acts as a time wedge on labor income, $w_t(1-\kappa\, n_t)$:
a higher $n_t$ contracts resources and lowers saving,
$\partial s_t^{w}/\partial n_t<0$, whereas greater longevity strengthens the
life-cycle saving motive, $\partial s_t^{w}/\partial\gamma_{t+1}>0$. Both shocks
depress consumption in both phases of life ($dc_t^{w},dc_{t+1}^{r}<0$), and
fertility falls with longevity, $dn_t/d\gamma_{t+1}<0$, as higher life expectancy
raises the opportunity cost of child-rearing. Together, these results map
demography and longevity into saving, consumption, and fertility through one
coherent equilibrium mechanism.

\subsection{Firms}
The economy comprises perfectly competitive final-good producers and a continuum
of monopolistically competitive intermediate-good firms indexed by
$i\in[0,M_t]$, where $M_t>0$ denotes the endogenous mass of active firms.
In the spirit of macroeconomic models that tie productive activity to
demographic structure \citep{bielecki2018demographics}, we link the extensive
margin of production to demographics by letting the mass of firms scale with
total population $N_t=N_t^{w}+N_t^{r}$. Because the level of $M_t$ is not separately
identified from the scale of the economy, we normalize the firm-to-population
ratio to one, so that $M_t=N_t$ and the number of active firms co-moves
one-for-one with population through the entry--exit mechanism described below.
Product variety is thus endogenous to demographics: population dynamics shape not
only factor supplies but also the number of operating firms, and hence the range
of varieties available in the economy.

\subsubsection{Final-good producers}

A representative final-good producer operates under perfect competition and
assembles the continuum of intermediate varieties $j\in[0,N_t]$ into the final
good through the CES aggregator
\begin{equation}
y_t=
\left[
\frac{1}{N_t}
\int_{0}^{N_t}
y_t(j)^{\frac{\xi-1}{\xi}}\,dj
\right]^{\frac{\xi}{\xi-1}},
\qquad \xi>1,
\end{equation}
where $y_t>0$ is aggregate final output, $y_t(j)>0$ the quantity of variety $j$,
and $\xi$ the elasticity of substitution across varieties. The $1/N_t$
normalization neutralizes the mechanical love-of-variety effect, so that variety
influences the economy through market structure rather than through a direct
productivity term.

Taking the variety prices $\{P_t(j)\}_j$ and the aggregate price index $P_t>0$ as
given, the producer chooses inputs $\{y_t(j)\}_j$ to minimize the cost of
delivering $y_t$. Cost minimization yields the demand system
\begin{equation}
y_t(j)=\left(\frac{P_t(j)}{P_t}\right)^{-\xi} y_t,
\end{equation}
and the associated aggregate price index
\begin{equation}
P_t=
\left[
\frac{1}{N_t}
\int_{0}^{N_t}
P_t(j)^{\,1-\xi}\,dj
\right]^{\frac{1}{1-\xi}},
\end{equation}
where $P_t(j)/P_t>0$ is the relative price of variety $j$.

\subsubsection{Intermediate-good producers}
There is a continuum of intermediate-good firms indexed by $j\in[0,N_t]$. In each period $t$, firm $j$ produces output $y_t(j)>0$ using physical capital $k_t(j)>0$, labor $h_t(j)>0$, and AI capital $a_t(j)>0$, according to the nested CES-Cobb-Douglas technology
\citep{prettner2019note,acemoglu2019automation,lankisch2019can}:
\begin{equation}
y_t(j)=k_t(j)^{\alpha}
\Big[h_t(j)^{\rho}+\big(\phi_t\,a_t(j)\big)^{\rho}\Big]^{\frac{1-\alpha}{\rho}},
\end{equation}
where $\alpha\in(0,1)$ is the share of physical capital, $\phi_t>0$ is the
productivity of AI capital, and $\rho=\frac{\sigma-1}{\sigma}$ with $\sigma>0$,
$\sigma\neq1$, the elasticity of substitution between labor and AI capital. The inner
CES aggregator $\big[h_t(j)^{\rho}+(\phi_t a_t(j))^{\rho}\big]^{1/\rho}$ is
homogeneous of degree one, so the technology exhibits constant returns to scale.
For $\sigma>1$ (equivalently $\rho\in(0,1)$), as in our calibration, labor and AI
capital are gross substitutes within the task nest.

Firm $j$ chooses inputs $(h_t(j),k_t(j),a_t(j))\in\mathbb{R}_{++}^{3}$ to minimize
cost,
\begin{equation}
\min_{h_t(j),\,k_t(j),\,a_t(j)}\;
w_t\,h_t(j)+R_t^{k}\,k_t(j)+R_t^{a}\,a_t(j),
\end{equation}
subject to the technology constraint
\begin{equation}
y_t(j)=k_t(j)^{\alpha}
\Big[h_t(j)^{\rho}+\big(\phi_t\,a_t(j)\big)^{\rho}\Big]^{\frac{1-\alpha}{\rho}},
\end{equation}
where $w_t$ is the wage and $R_t^{k}$, $R_t^{a}$ are the rental rates of physical
and AI capital, respectively. Because the technology has constant returns to
scale, cost minimization yields a marginal cost that is common to all firms and
independent of the scale $y_t(j)$. Let $\theta_t>0$ denote the Lagrange multiplier on the production constraint, that
is, the firm's real marginal cost. Cost minimization yields the first-order
conditions
\begin{align}
R_t^{k}
&=\alpha\,k_t(j)^{\alpha-1}
\Big[h_t(j)^{\rho}+(\phi_t a_t(j))^{\rho}\Big]^{\frac{1-\alpha}{\rho}}\theta_t,\\[4pt]
R_t^{a}
&=(1-\alpha)\,\phi_t^{\rho}\,k_t(j)^{\alpha}
\Big[h_t(j)^{\rho}+(\phi_t a_t(j))^{\rho}\Big]^{\frac{1-\alpha-\rho}{\rho}}
a_t(j)^{\rho-1}\,\theta_t,\\[4pt]
w_t
&=(1-\alpha)\,k_t(j)^{\alpha}
\Big[h_t(j)^{\rho}+(\phi_t a_t(j))^{\rho}\Big]^{\frac{1-\alpha-\rho}{\rho}}
h_t(j)^{\rho-1}\,\theta_t,
\end{align}
so each factor price equals its marginal product valued at $\theta_t$. By Euler's
theorem, constant returns imply the unit-cost representation $w_t\,h_t(j)+R_t^{k}\,k_t(j)+R_t^{a}\,a_t(j)=\theta_t\,y_t(j)$,
so that $\theta_t$ is scale-invariant and common to all firms. Facing the
constant-elasticity demand $y_t(j)=(P_t(j)/P_t)^{-\xi}y_t$, each firm prices at the
constant gross markup $\xi/(\xi-1)$. Under symmetry ($P_t(j)=P_t$), marginal cost
is therefore given by  $\theta_t=(\xi-1)/\xi$.

\subsection{Investment funds and financial-market clearing}
Saving is intermediated by perfectly competitive, risk-neutral investment funds.
Each period, the representative fund collects current saving $s_t^{w}$, repays the
gross return $R_{t-1}s_{t-1}^{w}$ on past saving, and owns the two capital stocks,
traditional capital $k_t$ and AI capital $a_t$, which it rents to intermediate
firms at rates $R_t^{k}$ and $R_t^{a}$ while financing investment $i_t^{k}$ and
$i_t^{a}$. Its net cash flow, rebated to households, is
\begin{equation}
d_t^{f}=s_t^{w}-R_{t-1}s_{t-1}^{w}+R_t^{k}k_t+R_t^{a}a_t-i_t^{k}-i_t^{a}.
\end{equation}
The fund chooses $\{i_t^{k},i_t^{a},k_{t+1},a_{t+1}\}_{t\ge0}$ to maximize the
expected present value of net cash flows,
\begin{equation}
\max\; \mathbb{E}_t\sum_{s\ge0}\Lambda_{t,t+s}\,d_{t+s}^{f},
\qquad
\Lambda_{t,t+s}=\prod_{j=1}^{s}\frac{1}{R_{t+j-1}},
\end{equation}
where, being competitive and risk neutral, the fund discounts at the gross
risk-free rate. Following \citet{christiano2005nominal}, capital accumulation is
subject to investment adjustment costs,
\begin{align}
k_{t+1}&=(1-\delta_k)\,k_{t}
+\left[1-\frac{\phi^{k}}{2}\left(\frac{i_t^{k}}{i_{t-1}^{k}}-\vartheta_{t-1}\right)^{2}\right]i_t^{k},\\[4pt]
a_{t+1}&=(1-\delta_a)\,a_{t}
+\left[1-\frac{\phi^{a}}{2}\left(\frac{i_t^{a}}{i_{t-1}^{a}}-\vartheta_{t-1}\right)^{2}\right]i_t^{a},
\end{align}
where $\delta_k,\delta_a\in(0,1)$ are depreciation rates and $\phi^{k},\phi^{a}>0$
govern the adjustment costs. The reference $\vartheta_{t-1}$ is the gross growth
rate of investment, which on the balanced growth path equals population growth,
$\vartheta_t=n_t$, so that the bracketed terms equal one and adjustment costs
vanish in steady state. The multipliers $q_t^{k}$ and $q_t^{a}$ on the two laws of
motion are the shadow prices of installed physical and AI capital (Tobin's $q$),
and coincide with the asset prices $q_k,q_a$ reported in the impulse responses.

Writing gross investment growth as $g^{x}_t\equiv i^{x}_t/i^{x}_{t-1}$,
$x\in\{k,a\}$, the first-order conditions deliver a no-arbitrage condition that
equates the expected gross returns on the two assets, together with the
investment Euler equations:
\begin{align}
R_t&=\frac{\mathbb{E}_t\!\big[R^k_{t+1}+(1-\delta_k)q^k_{t+1}\big]}{q^k_t}
   =\frac{\mathbb{E}_t\!\big[R^a_{t+1}+(1-\delta_a)q^a_{t+1}\big]}{q^a_t},\\[4pt]
&q^k_t\Big[1-\tfrac{\phi^k}{2}\big(g^k_t-\vartheta_{t-1}\big)^2
   -\phi^k g^k_t\big(g^k_t-\vartheta_{t-1}\big)\Big]
=1-\phi^k\,\mathbb{E}_t\!\Big[\tfrac{q^k_{t+1}}{R_t}\,(g^k_{t+1})^{2}
   \big(g^k_{t+1}-\vartheta_t\big)\Big],\\[4pt]
&q^a_t\Big[1-\tfrac{\phi^a}{2}\big(g^a_t-\vartheta_{t-1}\big)^2
   -\phi^a g^a_t\big(g^a_t-\vartheta_{t-1}\big)\Big]
=1-\phi^a\,\mathbb{E}_t\!\Big[\tfrac{q^a_{t+1}}{R_t}\,(g^a_{t+1})^{2}
   \big(g^a_{t+1}-\vartheta_t\big)\Big].
\end{align}
Financial-market clearing requires that household saving fully finance the value
of the capital stock carried into the following period,
\begin{equation}
s_t^{w}=q_t^{k}k_{t+1}+q_t^{a}a_{t+1}.
\end{equation}
The zero-profit condition for the competitive funds then determines the gross return on saving $R_t$, which the no-arbitrage condition equalizes across the two productive assets.
\subsection{Market-clearing conditions}

In equilibrium, all markets clear, firms maximize profits, and households maximize
utility. Household saving is channeled entirely into investment in the two capital
stocks, physical capital and AI capital, so that aggregate investment $z_t$ is the
only vehicle for saving. By symmetry across intermediate firms,
$k_t(j)=k_t$, $a_t(j)=a_t$, and $h_t(j)=h_t$ for all $j\in[0,N_t]$, where the
common per-firm input levels satisfy
\[
k_t\coloneqq \frac{1}{N_t}\int_0^{N_t} k_t(j)\,dj,
\quad
a_t\coloneqq \frac{1}{N_t}\int_0^{N_t} a_t(j)\,dj,
\quad
h_t\coloneqq \frac{1}{N_t}\int_0^{N_t} h_t(j)\,dj.
\]
Labor-market clearing equates aggregate labor demand $N_t h_t$ to aggregate
effective labor supply: each of the $N_t^{w}$ workers supplies $(1-\kappa\, n_t)$
units of time, so that
\begin{align}
N_t\,h_t = (1-\kappa\, n_t)\,N_t^{w}.
\end{align}
Goods-market clearing allocates output between consumption and investment. In
per-worker terms, with the old-age dependency ratio $\psi_t=N_t^{r}/N_t^{w}$, the
aggregate resource constraint is
\begin{align}
y_t = c_t + z_t,
\end{align}
where aggregate consumption per worker weights retirees by their relative mass,
\begin{align}
c_t = c_t^{w} + \psi_t\,c_t^{r},
\end{align}
and total investment combines physical and AI investment,
\begin{align}
z_t = i_t^{k} + i_t^{a}.
\end{align}

\subsection{Exogenous processes}
Formally, artificial intelligence productivity $\phi_t$ and the survival probability of retirees $\gamma_t$ follow stationary AR(1) processes:
\begin{align}
\phi_t &= (1-\rho_{\phi})\bar{\phi} + \rho_{\phi}\phi_{t-1} + \varepsilon_t^{\phi}, 
\qquad \rho_{\phi}\in(0,1), \\
\gamma_t &= (1-\rho_{\gamma})\bar{\gamma} + \rho_{\gamma}\gamma_{t-1} + \varepsilon_t^{\gamma},
\qquad \rho_{\gamma}\in(0,1),
\end{align}
where $\bar{\phi}$ and $\bar{\gamma}$ denote steady-state levels, and $\varepsilon_t^{\phi}$ and $\varepsilon_t^{\gamma}$ are i.i.d. shocks capturing innovations in technology and longevity risk. Shocks to $\phi_t$ represent exogenous improvements in AI efficiency that raise the marginal productivity of AI capital, while shocks to $\gamma_t$ capture changes in expected longevity that reshape intertemporal saving and consumption decisions. The model is used to compute impulse responses and transitional dynamics, highlighting the joint propagation of technological progress and demographic risk through general equilibrium channels.
\subsection{Comparative statics}
\label{sec:comparative}
This section examines how artificial intelligence reshapes the production structure and, through general equilibrium adjustments, household decisions and
long-run dynamics. AI capital enters production in a CES nest with labor that accommodates both substitutability and complementarity, so that the prevailing
regime governs whether technological progress complements or displaces labor. We first characterize the firm-level substitution induced by AI capital, and then
trace its general equilibrium effects on wages, fertility, saving, consumption,
and growth.
\begin{proposition}[AI and factor substitution]
\label{prop3}
Fix the firm's output at $\bar y>0$ and regard AI capital $a$ as a parameter.
For given factor prices $(w,R^{k})\in\mathbb{R}_{++}^2$, the restricted
cost-minimization problem
\[
\min_{(h,k)\in\mathbb{R}_{++}^2}\; w\,h+R^{k}k
\qquad\text{s.t.}\qquad
F(k,h,a)=k^{\alpha}\big[h^{\rho}+(\phi a)^{\rho}\big]^{\frac{1-\alpha}{\rho}}=\bar y,
\]
admits a unique interior solution $(h(a),k(a))$, the isoquant being strictly
convex, and the conditional input demands are strictly decreasing in AI capital:
\[
h'(a)<0,\qquad k'(a)<0.
\]
At constant output, AI capital therefore substitutes for both labor and physical
capital, in the sense that a larger AI stock economizes on the two other inputs.
\end{proposition}
\begin{proof}
See Appendix~\ref{appendprop3}.
\end{proof}

\begin{remark}
The output-constant substitutability of Proposition~\ref{prop3} should not be
confused with the effect of an AI expansion on factor prices in general
equilibrium. Because $\partial^2 F/\partial a\,\partial k>0$ for all admissible
parameters and $\partial^2 F/\partial a\,\partial h>0$ whenever $\rho<1-\alpha$,
a larger AI stock raises the marginal products of both physical capital and
labor. Thus, although AI capital substitutes for physical capital and labor at
constant output (Hicks substitute), it raises their marginal products (Edgeworth
complement); in general equilibrium, a positive AI shock therefore increases the
rental rate of physical capital and the real wage. Our baseline calibration,
$\rho=0.5$ and $\alpha=0.33$, satisfies $\rho<1-\alpha=0.67$ (equivalently,
$\sigma=2<1/\alpha\approx3.03$), so the model operates in the complementarity
regime; the labor-replacing case is reserved for the near-perfect-substitute
limit studied in Corollary~\ref{prop6}.
\end{remark}

By substituting for traditional physical capital, AI capital triggers broader equilibrium effects that propagate through labor markets and household decisions. Proposition~\ref{prop3} and the global comparative statics of an AI expansion (Proposition~\ref{prop4}, formalized in Appendix~\ref{appendprop4}) together show that artificial intelligence operates through both factor reallocation and general equilibrium adjustment. At the firm level, AI substitutes for labor and traditional capital
by automating production tasks. At the aggregate level, the resulting gains in productivity raise wages, saving, and lifetime consumption; because the income effect dominates, higher wages modestly raise fertility despite the increased
opportunity cost of child-rearing. Artificial intelligence thus reshapes not only the structure of production but also demographic dynamics and intertemporal allocation.

The next proposition studies the case in which artificial intelligence complements labor rather than replaces it.
\begin{proposition}[General equilibrium effects of AI under complementarity]
\label{prop5}
Assume $\rho<1-\alpha$, so that AI capital is Edgeworth-complementary to both
labor and physical capital,
$\partial^2 Y_t/\partial a_t\,\partial H_t>0$ and
$\partial^2 Y_t/\partial a_t\,\partial K_t>0$. Then, at a regular equilibrium, an
increase in AI capital $a_t$ satisfies
\[
\frac{\partial H_t}{\partial a_t}>0,
\qquad
\frac{\partial w_t}{\partial a_t}>0,
\qquad
\frac{\partial C_t}{\partial a_t}>0,
\]
the last inequality holding provided the investment share is interior (the
induced rise in investment does not exceed the rise in output). The rental of AI
capital, by contrast, is not unambiguously signed: at fixed labor it falls
through diminishing returns, $\partial R_t^{a}/\partial a_t\big|_{H_t}<0$, while
its general equilibrium response
$dR_t^{a}/da_t=\partial^2Y_t/\partial a_t^{2}
+(\partial^2Y_t/\partial a_t\partial H_t)\,(dH_t/da_t)$ is of ambiguous sign.
\end{proposition}
\begin{proof}
See Appendix~\ref{appendprop5}.
\end{proof}

Under complementarity ($\rho<1-\alpha$), AI raises the marginal product of labor,
so a larger AI stock stimulates labor demand, the real wage, and, unless the
investment response is too strong, aggregate consumption, yielding positive
spillovers on welfare and macroeconomic performance. These implications reverse
sharply when labor and AI become highly substitutable ($\rho>1-\alpha$), the case
to which we now turn.
\begin{corollary}[AI effects under (near-)perfect substitutability]
\label{prop6}
Let $Y_t=K_t^{\alpha}\big[H_t^{\rho}+(\phi a_t)^{\rho}\big]^{\frac{1-\alpha}{\rho}}$
and take $\rho\to1^{-}$ ($\sigma\to\infty$), so that labor and AI capital become
perfect substitutes and the inner aggregate tends to $H_t+\phi a_t$. Holding the
effective-labor aggregate constant (equivalently, $K_t$ and $w_t$ fixed), one
efficiency unit of AI capital displaces $\phi$ units of labor,
$\partial H_t/\partial a_t=-\phi$, and the equilibrium responses to AI
accumulation are
\[
\frac{\partial H_t}{\partial a_t}<0,
\qquad
\frac{\partial w_t}{\partial a_t}\le 0,
\qquad
\frac{\partial R_t^{a}}{\partial a_t}\le 0,
\qquad
\frac{\partial C_t}{\partial a_t}\gtrless 0,
\]
reversing the labor-demand and wage responses of Proposition~\ref{prop5}. Since
$\mathrm{MPA}=\phi\,\mathrm{MPL}$, the rental of AI capital is tied to the wage,
$R_t^{a}=\phi\,w_t$, so $\partial R_t^{a}/\partial a_t$ shares the sign of
$\partial w_t/\partial a_t$ and both fall as AI accumulates. Artificial
intelligence is thus labor-replacing.
\end{corollary}
\begin{proof}
See Appendix~\ref{appendprop6}.
\end{proof}
\begin{remark}
At $\rho=1$ exactly, the labor--AI nest is linear, so interior factor demands are
determinate only along the price locus $R_t^{a}=\phi\,w_t$; off this ray the firm
specializes in a single factor (a corner solution). We therefore state the
comparative statics as the limit $\rho\to1^{-}$, along which the displacement
$\partial H_t/\partial a_t=-\phi$ and the price identity $R_t^{a}=\phi\,w_t$ hold.
The limit is necessarily one-sided: since $\rho=1-1/\sigma$ with $\sigma>0$, one
has $\rho<1$ throughout, so $\rho=1$ is the upper bound of the admissible range,
approached only from below as $\sigma\to\infty$ (a right limit would require
$\sigma<0$, for which the technology is not well defined).
\end{remark}
The corollary underscores the central role of the elasticity of substitution
$\sigma$. The two regimes are separated by the threshold $\rho=1-\alpha$
(equivalently $\sigma=1/\alpha$), at which $\partial^2 Y_t/\partial H_t\partial
a_t$ changes sign: AI is labor-augmenting for $\rho<1-\alpha$
(Proposition~\ref{prop5}) and labor-replacing for $\rho>1-\alpha$. The baseline
calibration ($\sigma=2<1/\alpha\approx3.03$) lies in the complementarity regime,
so the labor-replacing case merely delimits the parameter space.\footnote{In that
regime an AI expansion lowers the marginal product of labor, contracts labor
demand, and depresses the wage; in the perfect-substitute limit the rental of AI
capital, tied to the wage through $R_t^{a}=\phi\,w_t$, falls alongside it. The
effect on consumption is ambiguous: output still (weakly) rises, but the fall in
labor income works against it, so the sign of $\partial C_t/\partial a_t$ depends
on the strength of the induced investment response. Throughout, technological
change interacts with demographic decisions through the labor-supply adjustments
induced by the time cost of children.}

The time-cost channel linking fertility to effective labor supply is formalized in
Lemma~\ref{prop7}, and a two-capital decomposition of growth in
Proposition~\ref{prop8}; both are collected in Appendix~\ref{app:proofs}. The latter
decomposes the growth effect of AI into a positive productivity channel and two
offsetting channels, consumption crowding-out and capital dilution; the net effect is
positive precisely when next-period capital is more AI-elastic than current capital. Under complementarity and a sufficiently
strong productivity response this condition holds and AI raises long-run growth,
whereas otherwise the offsetting channels dominate. More broadly, the section
shows that artificial intelligence jointly reshapes production, labor-market
outcomes, demographic behavior, and growth, underscoring the tight interaction
between technological change, household decisions, and macroeconomic development.

\subsection{Calibration and functional forms}
We adopt logarithmic functions,
\begin{align}
    u(c_t^{w},n_t)=\log(c_t^{w})+\log(n_t),
\end{align} 
\begin{align}v(c_{t+1}^{r})=\log (c_{t+1}^{r})
\end{align} 
which deliver interior solutions, unit intertemporal elasticity, and a
multiplicative valuation of fertility, as is standard in macro-demographic models.
\section{Quantitative experiments}
\label{sec:quant}
This section characterizes the macroeconomic transmission of the model’s two structural forces, artificial intelligence productivity and longevity risk, tracing their equilibrium effects on the transitional dynamics of output,
factor prices, saving, consumption, and demographic outcomes. The analysis
proceeds in three steps. Section~\ref{sec:calibration} presents the baseline
calibration. Section~\ref{sec:irf} studies each shock through its impulse
response functions. Section~\ref{sec:moments} summarizes the model's second-order
properties, isolating the relative contribution of each shock to aggregate
volatility.

\subsection{Calibration}
\label{sec:calibration}
The calibration (Table~\ref{tab:parameters}) combines conventional macroeconomic
values with one steady-state target. On the technology side, the physical-capital
share is $\alpha=0.33$ and the labor--AI substitution parameter is $\rho=0.50$,
i.e.,\ $\sigma=2$; since $\rho<1-\alpha=0.67$ (equivalently
$\sigma<1/\alpha\approx3.03$), the baseline lies in the complementarity regime in
which AI raises the marginal product of labor. Although direct estimates of the
labor--AI elasticity are scarce, the broader literature on capital--labor
substitution provides a guide, and it places the relevant elasticity well below
the reversal threshold $\sigma^{*}=1/\alpha\approx3.03$: survey evidence centers on
$0.4$--$0.6$ \citep{chirinko2008sigma}, micro-aggregated estimates on $0.5$--$0.7$
\citep{oberfield2021micro}, and even the gross-substitute estimates invoked to
explain the falling labor share reach only about $1.3$
\citep{karabarbounis2014global}. The baseline $\sigma=2$ is therefore a
deliberately conservative choice, set above the bulk of these estimates yet still
safely within the complementarity regime; at the empirically more likely lower
values, complementarity is only reinforced. AI productivity is normalized to
$\phi_{ss}=1.20$, depreciation rates are $\delta_k=0.08$ and $\delta_a=0.12$
(faster obsolescence of AI capital), and investment adjustment costs are
$\phi_k=2.50$ and $\phi_a=3.00$. The elasticity of substitution across varieties
is $\xi=10$, implying a gross markup $\xi/(\xi-1)\approx1.11$. On preferences, the
discount factor is $\beta=0.96$. On demographics, the retiree survival probability
is $\gamma_{ss}=0.93$, and the time cost per child is set to $\kappa=0.1949$ so that steady-state
fertility equals $n_{ss}=2$; together these imply an old-age dependency ratio
$\psi_{ss}=\gamma_{ss}/n_{ss}=0.465$. The AI and longevity shocks are persistent,
with $\rho_\phi=0.85$ and $\rho_\gamma=0.90$.
\begin{table}[!ht]
\centering
\caption{Baseline Calibration of the Model}
\label{tab:parameters}
\begin{tabular}{l c p{8.2cm}}
\toprule
\textbf{Parameter} & \textbf{Value} & \textbf{Economic interpretation} \\
\midrule
$\alpha$        & 0.33   & Share of physical capital in production \\
$\beta$         & 0.96   & Household discount factor \\
$\rho$          & 0.50   & CES substitution parameter between labor and AI capital \\
$\phi_{ss}$     & 1.20   & Steady-state productivity level of AI capital \\
$\rho_{\phi}$   & 0.85   & Persistence of the AI productivity shock \\
$\phi_a$        & 3.00   & AI-capital adjustment-cost parameter \\
$\phi_k$        & 2.50   & Physical-capital adjustment-cost parameter \\
$\gamma_{ss}$   & 0.93   & Survival probability of retirees (long-run) \\
$\rho_{\gamma}$ & 0.90   & Persistence of the longevity shock \\
$\delta_k$      & 0.08   & Depreciation rate of physical capital \\
$\delta_a$      & 0.12   & Depreciation rate of AI capital \\
$\kappa$        & 0.1949 & Time cost per child (calibrated so that $n_{ss}=2$) \\
$n_{ss}$        & 2.00   & Steady-state fertility (children per household) \\
$\psi_{ss}$     & 0.465  & Steady-state old-age dependency ratio \\
$\xi$           & 10     & Substitution elasticity, intermediate goods \\ 
\bottomrule
\end{tabular}
\end{table}
Table~\ref{tab:parameters} presents the baseline calibration. The parameters are divided into two categories.
The first category comprises conventional values from the macroeconomic literature: the capital share $\alpha=0.33$, the household
discount factor $\beta=0.96$, the physical-capital depreciation rate
$\delta_k=0.08$, and the elasticity of substitution across intermediate goods
$\xi=10$, implying a steady-state markup of approximately eleven percent.

AI capital exhibits a higher depreciation rate than physical capital
($\delta_a=0.12$), reflecting the accelerated obsolescence of the equipment and
software in which it is embodied. The persistence parameters for the two driving
processes, $\rho_\phi=0.85$ and $\rho_\gamma=0.90$, and the adjustment-cost
parameters $\phi_a=3.00$ and $\phi_k=2.50$, govern the speed at which the economy
absorbs each disturbance.

The second category is calibrated to ensure that the deterministic steady state reproduces a limited
number of demographic and technological targets. The time cost per child,
$\kappa=0.1949$, is calibrated to deliver a steady-state fertility rate of
$n_{ss}=2$ children per household; given the long-run survival probability of
retirees $\gamma_{ss}=0.93$, this yields an old-age dependency ratio
$\psi_{ss}=\gamma_{ss}/n_{ss}=0.465$. The steady-state productivity of AI capital
is normalized to $\phi_{ss}=1.20$, and the elasticity of substitution between
labor and AI capital is governed by the CES parameter $\rho=0.50$.
\subsection{Transitional dynamics}
\label{sec:irf}
We analyze each shock through its impulse response functions, reported in
Figures~\ref{fig:irf1} and~\ref{fig:irf2}. To keep responses of very
different magnitude visible, the two disturbances are shown on a dual vertical
axis: the left-hand scale (solid blue line, circular markers) corresponds to the
AI technology shock $\phi_t$, and the right-hand scale (dashed vermillion line,
square markers) to the longevity shock $\gamma_t$. All variables are expressed as
percentage deviations from the steady state over a horizon of twenty periods
following the shock. Figure~\ref{fig:irf1} collects the responses of the main
aggregates, while Figure~\ref{fig:irf2} reports factor prices and demographic
variables. Table~\ref{tab:summary} previews the qualitative predictions that the
impulse responses confirm: the two shocks share the same sign on quantities but
carry opposite signs on the returns to capital and on fertility, which is the
exact imprint of the demand-versus-supply asymmetry.

\begin{table}[!ht]
\centering
\caption{Qualitative transmission of the two shocks. ``$+$'' (``$-$'') denotes an
increase (decrease); ``$+\,+$'' marks the largest response. The AI shock is
evaluated in the baseline complementarity regime.}
\label{tab:summary}
{\small\setlength{\tabcolsep}{5pt}%
\begin{tabular}{@{}p{3.8cm} P{1.8cm} P{2.2cm} p{4.7cm}@{}}
\toprule
Variable & AI shock $\varepsilon_{\phi}$ & Longevity shock $\varepsilon_{\gamma}$ & Dominant channel \\
\midrule
Output $y$                      & $+$    & $+$ & front-loaded (demand) vs.\ hump-shaped (supply) \\
Real wage $w$                   & $+$    & $+$ & AI raises the marginal product of labor; longevity deepens capital \\
AI capital $a$                  & $+$    & $+$ & investment reallocation vs.\ capital deepening \\
Return to AI capital $R^{a}$    & $+\,+$ & $-$ & capital demand vs.\ saving supply \\
Return to physical capital $R^{k}$ & $+$ & $-$ & capital demand vs.\ saving supply \\
Real interest rate $R$          & $+$    & $-$ & capital demand vs.\ saving supply \\
Fertility $n$                   & $+$    & $-$ & income effect vs.\ life-cycle saving and child-rearing cost \\
Hours worked $h$                & $-$    & $+$ & mirror of fertility under a fixed time endowment \\
Old-age dependency $\psi$       & $-$    & $+$ & higher survival and lower fertility compound \\
\bottomrule
\end{tabular}}
\end{table}

\subsubsection{The AI productivity shock}
A positive AI productivity shock increases the marginal productivity of AI
capital. As shown in Figure~\ref{fig:irf2}, the shock immediately increases the
return on AI capital (approximately $+0.43\%$), which exceeds the return on
physical capital. Households therefore reallocate their savings toward AI capital,
consistent with the portfolio no-arbitrage condition between the two assets. This reallocation
raises AI investment and the AI capital stock, while physical investment and
physical capital decline. Because aggregate saving is fixed in the short run, the
expansion of AI capital \emph{crowds out} physical capital. The resulting response is a
reallocation across asset classes rather than a uniform expansion of the capital
stock. As physical capital contracts, diminishing returns raise the marginal
product of the remaining stock; the complementarity between AI capital and
physical capital in production reinforces this effect, since a larger AI capital
stock and more productive labor both raise the marginal product of physical
capital. Relative prices move accordingly. The price of AI capital rises
(approximately $+0.19\%$) while the price of physical capital declines slightly
(approximately $-0.02\%$), a manifestation of Tobin's $q$, whereby investment
flows toward the capital good valued above its replacement cost.

In the labor market, the higher productivity of AI capital shifts labor demand
outward and raises the real wage (approximately $+0.027\%$). By raising
productivity and expanding the stock of AI capital that complements labor, the
shock increases both labor demand and wages. At the same time, the associated
wealth effect, as households become wealthier, reduces labor supply, so that wages
rise even as hours worked fall. The responses of the main aggregates to both
shocks are reported in Figure~\ref{fig:irf1}.
\begin{figure}[!ht]
\centering
\includegraphics[width=\textwidth,height=0.9\textheight,keepaspectratio]{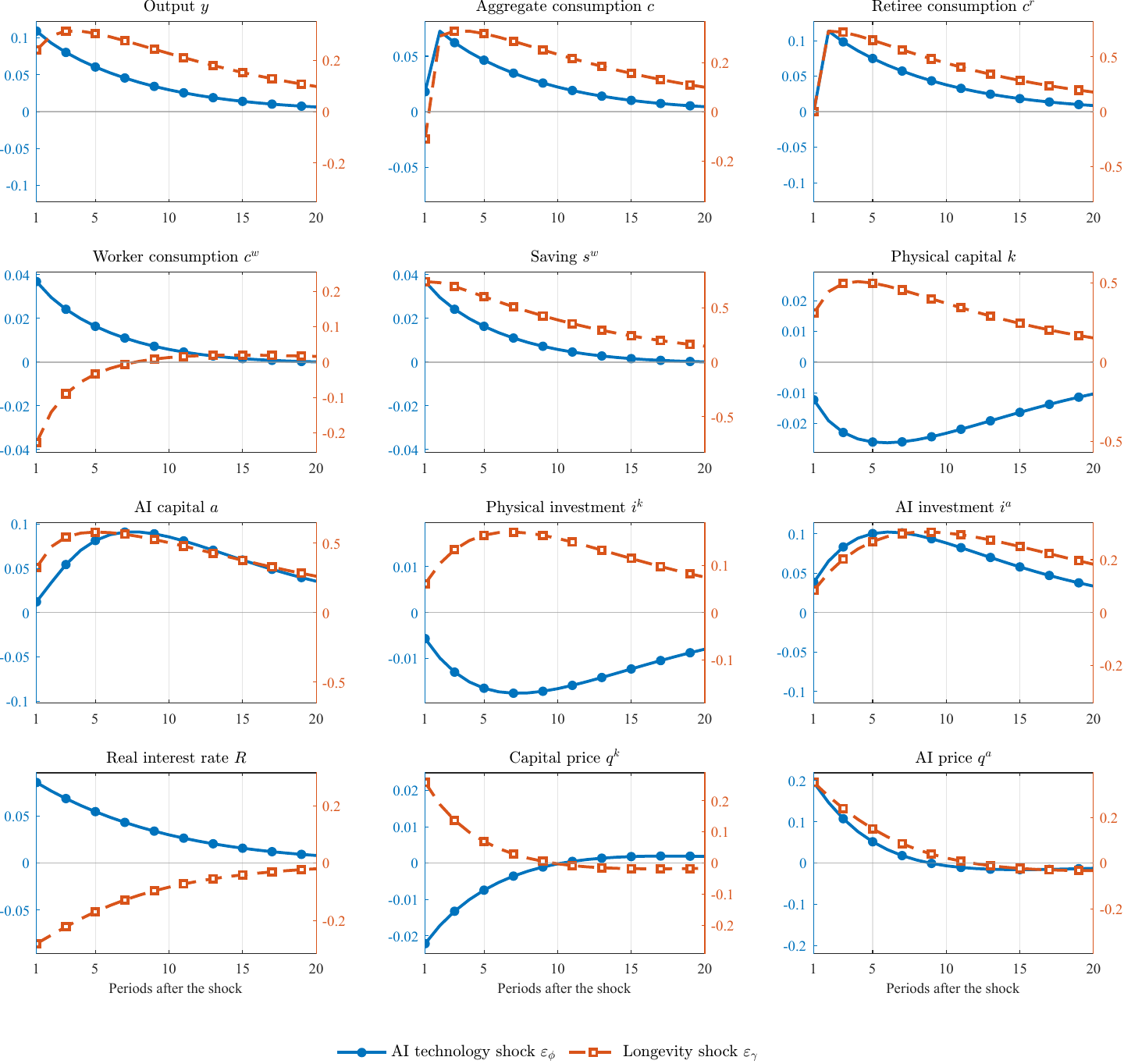}
\caption{Impulse Responses of the Main Aggregates}
\label{fig:irf1}
\end{figure}
The combination of higher wages and lower hours reflects a contraction in labor
\emph{supply} along an outward-shifting demand curve, rather than a decline in labor demand.
Internal consistency requires that AI capital complement labor; were AI a strong
substitute, the marginal product of labor and the wage could instead decline.

Labor supply responds to two opposing forces. The substitution effect, in which a higher
wage increases the reward to working, pushes hours up, while the income effect,
greater wealth increases the demand for leisure, pushes them down. The slight net
reduction in hours worked indicates that the income effect dominates in the short
run. The higher return on capital also raises the real interest rate
(approximately $+0.08\%$). Through the Euler equation, this tilts the consumption
profile toward the future; in levels, however, the wealth effect dominates and
aggregate consumption rises, driven primarily by retirees, who hold a larger share
of wealth.
\begin{figure}[!ht]
\centering
\includegraphics[width=\textwidth,height=0.9\textheight,keepaspectratio]{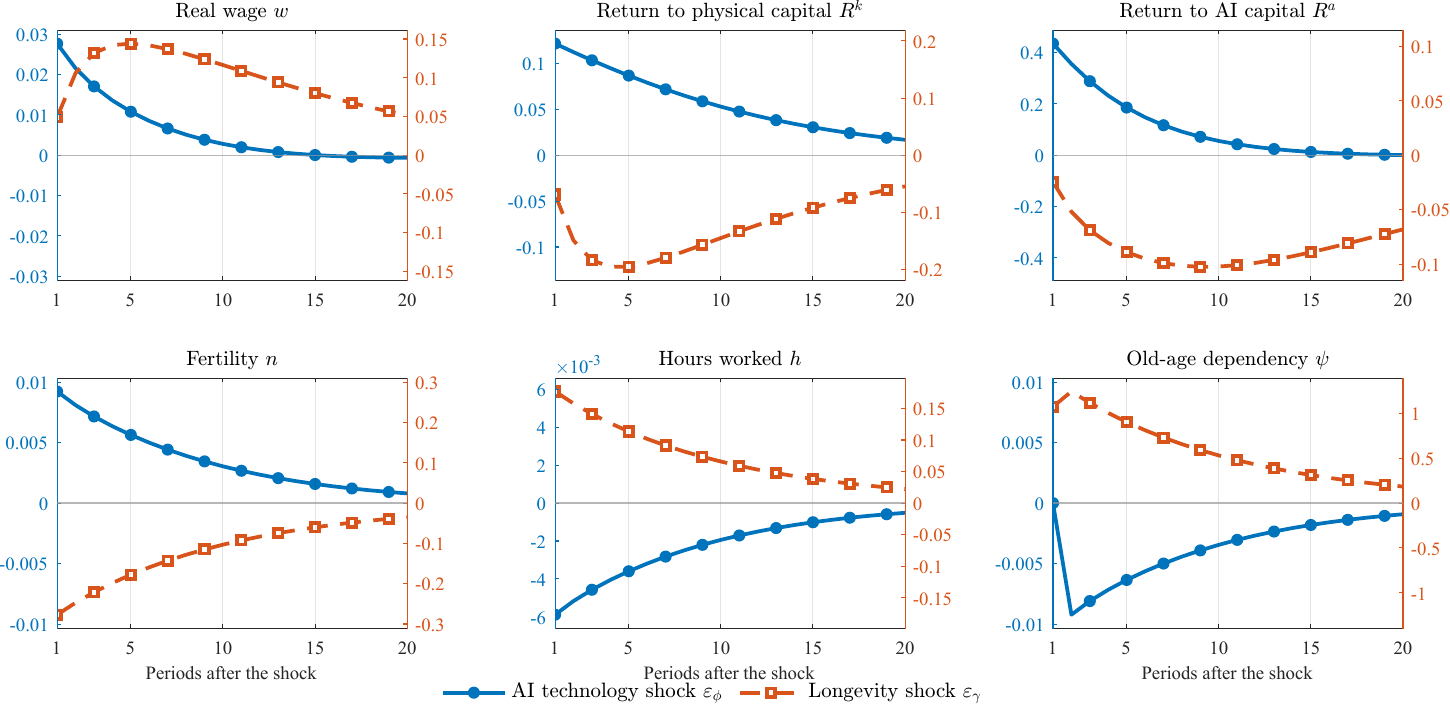}
\caption{Impulse Responses of Factor Prices and Demographic Variables}
\label{fig:irf2}
\end{figure}
The higher wage increases the opportunity cost of parental time, a substitution
effect that would, in isolation, reduce fertility. The offsetting income effect
leaves fertility only marginally higher. The AI shock thus exerts large effects in
the productive and financial blocks but only minimal short-run demographic
effects, as the factor-price and demographic responses in Figure~\ref{fig:irf2}
make clear. Finally, as AI capital accumulates, its marginal productivity and its
return premium over physical capital erode, the physical capital stock recovers,
and output, which rises on impact by approximately $+0.11\%$, declines steadily
as the productivity impulse dissipates.

\subsubsection{The longevity shock}
The longevity shock generates the opposite pattern across
Figures~\ref{fig:irf1} and~\ref{fig:irf2}. A rise in life expectancy strengthens
the motive for retirement saving: workers cut current consumption by approximately
$0.2\%$ to save more. The resulting increase in the supply of loanable
funds works in exactly the opposite direction from the AI shock, which operated on
the demand side.

Higher savings raise the accumulation of both physical and AI capital, and hence
investment in each. Because capital accumulates only gradually, output rises slowly
and peaks at approximately $+0.3\%$ in the fourth or fifth period
(Figure~\ref{fig:irf1}), a more sluggish profile than under the AI shock. As the
capital stock expands while technology is held fixed, the return on capital falls:
the real interest rate declines by about $-0.2\%$, and the returns on physical and
AI capital turn negative (approximately $-0.18\%$ and $-0.11\%$, respectively).
This outcome is indicative of a savings glut. The price of capital $q^{k}$ rises on
impact, and the larger capital stock raises labor productivity, generating a
hump-shaped increase in the real wage $w$ of about $+0.14\%$.

The old-age dependency ratio $\psi$ exhibits the largest response in the system
(Figure~\ref{fig:irf2}). Since $\psi=\gamma/n$, it is governed by the survival rate
of retirees ($\gamma$) and the fertility rate ($n$), and both margins move in the
same direction here. First, higher survival rates enlarge the elderly population.
Second, fertility declines by approximately $-0.2\%$ as households, anticipating
longer lifespans, prioritize their own future consumption and reduce investment in
child-rearing; lower interest rates and a longer retirement horizon further tilt the allocation away from children and toward saving. Households accordingly work
more hours to accumulate retirement wealth, while retirees, more numerous and
longer-lived, raise their consumption.

\subsection{Business cycle properties}
\label{sec:moments}
The model is driven by two structural shocks, each calibrated to a standard
deviation of $1\%$ and uncorrelated with the other: an AI technology shock
($\varepsilon_\phi$) and a longevity shock ($\varepsilon_\gamma$). We summarize its
second-order properties through theoretical moments and persistence, the
unconditional variance decomposition, and contemporaneous cross-correlations; the supporting tables are provided in the Appendix
(Tables~\ref{tab:moments}--\ref{tab:corr}).

Volatility is concentrated in prices and demographic ratios rather than in real
quantities: the real interest rate ($\sigma=0.0151$), the dependency ratio $\psi$
($0.0141$), the capital returns $R^a$ ($0.0133$) and $R^k$ ($0.0104$), and fertility
$n$ ($0.0125$) are the most volatile, whereas output, consumption, capital, and
investment fall between $0.0005$ and $0.0045$. The accumulation variables are highly
persistent ($\rho(1)\!\ge\!0.97$ for $k$, $a$, $i^k$, $i^a$, and $w$), which accounts
for the slow, hump-shaped impulse responses documented above, while asset prices and
workers' consumption revert more quickly ($\rho(1)\in[0.65,0.80]$). The variance
decomposition delivers the central finding: longevity dominates, explaining ninety
to one hundred percent of the variance in nearly every variable: output ($95.8\%$),
consumption ($98\%$), capital ($99.7\%$), fertility ($99.9\%$), and the dependency
ratio ($100\%$); whereas the AI shock matters only for capital pricing, accounting
for $75.5\%$ of the variance in $R^a$, $18.8\%$ of $R^k$, $18.5\%$ of $q^a$, and
$9.3\%$ of the real interest rate. Since the two shocks share a common standard
deviation, these shares reflect the model's transmission mechanism rather than
differences in shock size. The correlation structure indicates a single dominant
factor in the real block: output, consumption, capital, AI capital, investment, the
wage, and saving co-move strongly ($0.8$ to $0.99$), the real interest rate moves
countercyclically as higher saving depresses its equilibrium value, and $R^a$ is
only weakly tied to the real block ($-0.27$), being governed by the
longevity-independent AI shock. Fertility is the mirror image, negatively correlated
with every real aggregate and \emph{perfectly} with hours worked ($-1.00$), a
direct expression of the household's time-allocation trade-off between market work
and children. Taken together, these second-order statistics reinforce the
general equilibrium intuition obtained from the impulse response analysis.

A noteworthy feature of the correlation matrix is that worker consumption, $c^{w}$, is the only consumption variable that moves negatively with output ($\mathrm{corr}(c^{w}, y) = -0.32$). This pattern is a direct signature of the saving-supply mechanism, not an anomaly. The longevity shock accounts for most of the volatility in both $c^{w}$ and $y$ (Table~\ref{tab:vardec}), prompting workers to compress $c^{w}$ to finance capital deepening, thereby increasing $y$. Thus, $c^{w}$ and $y$ move in opposite directions under the dominant shock. The same logic explains the negative comovement of $c^{w}$ with $s^{w}$, $k$, $a$, $i^{k}$, $i^{a}$, $h$, and prices of installed capital, $q^{k}$ and $q^{a}$; all move with output in response to the saving-supply channel. In addition, $c^{w}$ moves positively with fertility, $n$, since the longevity shock depresses both. In a standard one-agent RBC model with only TFP shocks, consumption would be uniformly procyclical. Thus, the negative correlation here highlights the source of fluctuations in an OLG economy dominated by a saving-supply force. Figure~\ref{fig:vardec} summarizes this division of labor between the two disturbances: the longevity shock accounts for nearly the entire forecast-error variance of the quantity block and of the real interest rate, while the AI shock is concentrated in the pricing of capital, where it explains the bulk of the variance of the return to AI capital.

\begin{figure}[!htbp]
\centering
\includegraphics[width=0.82\textwidth]{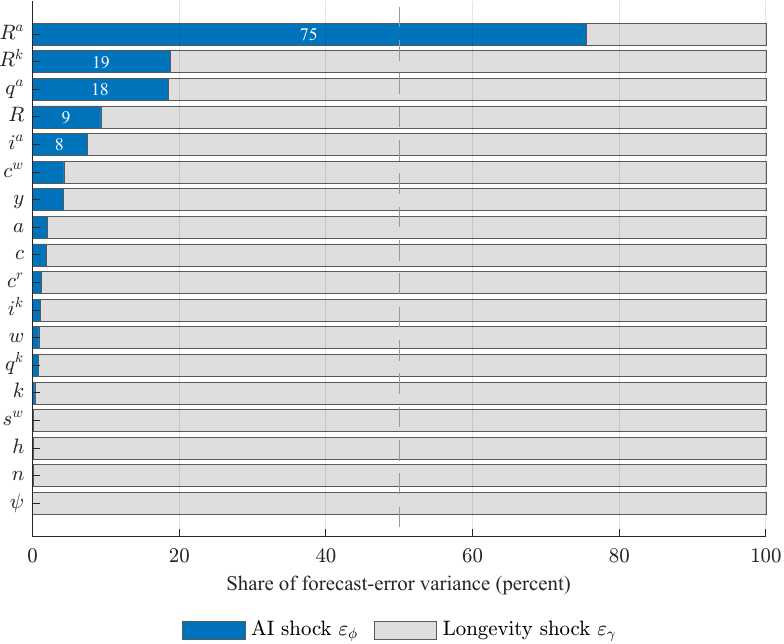}
\caption{Unconditional forecast-error variance decomposition: share attributable
to the AI shock $\varepsilon_{\phi}$ (dark) and the longevity shock
$\varepsilon_{\gamma}$ (light), summing to one (Table~\ref{tab:vardec}).
Variables are ordered by the AI share.}
\label{fig:vardec}
\end{figure}

\section{Artificial intelligence and the labor market}
\label{sec:labormarket}
The transmission analysis established that, in the complementarity regime, an AI
expansion raises the real wage. We now isolate the labor-market footprint of
artificial intelligence and show that it is governed by \emph{two} distinct
elasticity thresholds rather than one: a threshold for the \emph{level} of the
wage and a separate threshold for the \emph{labor share}. Because the model's
factor demands are available in closed form, the relevant elasticities can be
written explicitly and evaluated at the calibration without further simulation.
Throughout this section the comparative statics are partial-equilibrium objects:
they hold physical capital $K_t$ and hours $H_t$ at their equilibrium levels and
trace the factor-demand response to a marginal expansion of the AI stock, thereby
isolating the technological channel from the household adjustments already
analyzed in Section~\ref{sec:irf}.

\subsection{Wages and the labor share}
\label{sec:lm_wages}

\begin{lemma}[Factor shares and labor-market elasticities]
\label{lem:labor}
Let $Y_t=K_t^{\alpha}\big[H_t^{\rho}+(\phi_t a_t)^{\rho}\big]^{(1-\alpha)/\rho}$ with
$\rho=1-1/\sigma$, and write
$s_{a}\equiv(\phi_t a_t)^{\rho}\big/\!\big[H_t^{\rho}+(\phi_t a_t)^{\rho}\big]\in(0,1)$
and $s_{L}\equiv(1-\alpha)(1-s_{a})$ for the AI share of the labor--AI composite and
the aggregate labor share. At fixed $(K_t,H_t)$, the elasticities of output, the
wage, the AI rental, and the labor share with respect to AI capital are
$\partial\ln Y_t/\partial\ln a_t=(1-\alpha)\,s_{a}>0$,
$\partial\ln w_t/\partial\ln a_t=(1-\alpha-\rho)\,s_{a}$,
$\partial\ln R^{a}_t/\partial\ln a_t=(1-\alpha-\rho)\,s_{a}-(1-\rho)$, and
$\partial\ln s_{L}/\partial\ln a_t=-\rho\,s_{a}$.
\end{lemma}
\begin{proof}
See Appendix~\ref{app:labor}.
\end{proof}

\begin{proposition}[The labor-share threshold]
\label{prop_share}
An AI expansion lowers the aggregate labor share if and only if labor and AI
capital are gross substitutes in the task nest, $\sigma>1$; it raises the labor
share if $\sigma<1$, and leaves it unchanged in the Cobb--Douglas nest $\sigma=1$.
The size of the decline, $|\partial\ln s_{L}/\partial\ln a_t|=\rho\,s_{a}$, is
increasing in both the degree of substitutability $\rho$ and the AI income share
$s_{a}$.
\end{proposition}
\begin{proof}
Immediate from $\partial\ln s_{L}/\partial\ln a_t=-\rho\,s_{a}$ in
Lemma~\ref{lem:labor}, since $s_a>0$ and $\operatorname{sign}\rho
=\operatorname{sign}(\sigma-1)$.
\end{proof}

\begin{proposition}[Two thresholds: the wage level versus the labor share]
\label{prop_two}
The wage and the labor share respond to AI through two distinct knife-edges. The
wage rises if and only if $\sigma<1/\alpha$ (Edgeworth complementarity,
Proposition~\ref{prop5}), whereas the labor share falls if and only if
$\sigma>1$. Since $\alpha\in(0,1)$ the two thresholds are ordered,
$1<1/\alpha$, and for every $\sigma$ in the intermediate band
\[
1<\sigma<\frac{1}{\alpha}
\]
an AI expansion \emph{simultaneously} raises the real wage and lowers the labor
share. The baseline calibration $\sigma=2$ lies inside this band
($1<2<3.03$).
\end{proposition}
\begin{proof}
The wage sign follows from $\partial\ln w_t/\partial\ln a_t=(1-\alpha-\rho)s_a$,
which is positive iff $\rho<1-\alpha$, i.e.\ $\sigma<1/\alpha$; the labor-share
sign is Proposition~\ref{prop_share}. Ordering and the band are then immediate.
\end{proof}

\begin{corollary}[Wage--share decoupling at the baseline]
\label{cor_decouple}
At the baseline calibration an AI expansion raises the real wage yet lowers
labor's share of income: the wage gain is outpaced by the rise in output, so labor
receives a smaller slice of a larger pie. The model therefore reconciles a
positive wage response to automation with the secular decline of the labor share
documented for capital- and automation-intensive economies
\citep{karabarbounis2014global, acemoglu2018race}, without invoking any
fall in the wage itself.
\end{corollary}
\begin{proof}
By Lemma~\ref{lem:labor}, $\partial\ln w_t/\partial\ln a_t>0$ while
$\partial\ln(w_tH_t/Y_t)/\partial\ln a_t
=\partial\ln w_t/\partial\ln a_t-\partial\ln Y_t/\partial\ln a_t
=(1-\alpha-\rho)s_a-(1-\alpha)s_a=-\rho s_a<0$ for $\sigma>1$.
\end{proof}

\begin{remark}[A Hicksian reading, and the labor share beyond Cobb--Douglas]
\label{rem:hicks}
Proposition~\ref{prop_share} is the classical Hicksian statement that a factor's
income share rises with its own quantity if and only if the elasticity of
substitution against it is below unity: AI's share within the task nest rises with
$a_t$ precisely when $\sigma>1$, at the expense of labor. The result extends to the
general CES nest of Section~\ref{sec:robust_nest}. With outer elasticity
$\sigma_o=1/(1-\eta)$ and physical-capital value share $s_{K}$, the aggregate labor
share is $s_{L}=(1-s_{K})(1-s_{a})$ and responds to AI according to
\[
\frac{\partial\ln s_{L}}{\partial\ln a_t}=-\,s_{a}\,(\rho-\eta\,s_{K}),
\]
which collapses to $-\rho\,s_{a}$ in the Cobb--Douglas outer nest
($\eta=0$, $s_{K}=\alpha$). The labor share thus falls whenever $\rho>\eta\,s_{K}$;
since $\eta\,s_{K}<s_{K}<\rho$ at the baseline, this holds for \emph{every} outer
elasticity $\sigma_o$, so the share-eroding effect of AI is more robust than the
wage threshold of Proposition~\ref{prop_outer}, which it nests as the special case
$\sigma_o\to1$.
\end{remark}

Table~\ref{tab:labor} evaluates these elasticities at the baseline calibration,
where the AI income share is $s_{a}=0.206$ (so that labor receives
$s_{L}=0.53$ of output). The signs realize the band of
Proposition~\ref{prop_two}: a ten-percent rise in the AI stock raises the wage by
about $0.35\%$ and output by about $1.38\%$, and in consequence lowers the labor
share by about $1.0\%$; the rental of AI capital falls by about $4.7\%$ through
diminishing returns, while AI's own income share rises by about $4.0\%$.
Figure~\ref{fig:labor} traces the same elasticities across the elasticity of
substitution $\sigma$, making the two knife-edges visible: the wage elasticity
crosses zero at $\sigma=1/\alpha\approx3.03$ and the labor-share elasticity at
$\sigma=1$, with the shaded band $1<\sigma<1/\alpha$ marking the empirically
relevant region in which AI is wage-enhancing but share-eroding.

\begin{table}[!ht]
\centering
\caption{Labor-market elasticities of an AI expansion at the baseline calibration
($\sigma=2$, $\alpha=0.33$, $s_{a}=0.206$): partial elasticity with respect to the
AI stock (holding $K_t$, $H_t$ fixed) and the response to a ten-percent rise in AI
capital.}
\label{tab:labor}
\begin{tabular}{l c r r}
\toprule
Object & Elasticity $\partial\ln(\cdot)/\partial\ln a_t$ & Value & Response to $+10\%$ AI \\
\midrule
Output $Y$            & $(1-\alpha)\,s_a$              & $+0.138$ & $+1.38\%$ \\
Real wage $w$         & $(1-\alpha-\rho)\,s_a$         & $+0.035$ & $+0.35\%$ \\
Labor share $s_L$     & $-\rho\,s_a$                   & $-0.103$ & $-1.03\%$ \\
AI income share $s_a$ & $\rho\,(1-s_a)$               & $+0.397$ & $+3.97\%$ \\
AI rental $R^{a}$     & $(1-\alpha-\rho)\,s_a-(1-\rho)$ & $-0.465$ & $-4.65\%$ \\
\bottomrule
\end{tabular}
\end{table}

\begin{figure}[!ht]
\centering
\includegraphics[width=0.82\textwidth]{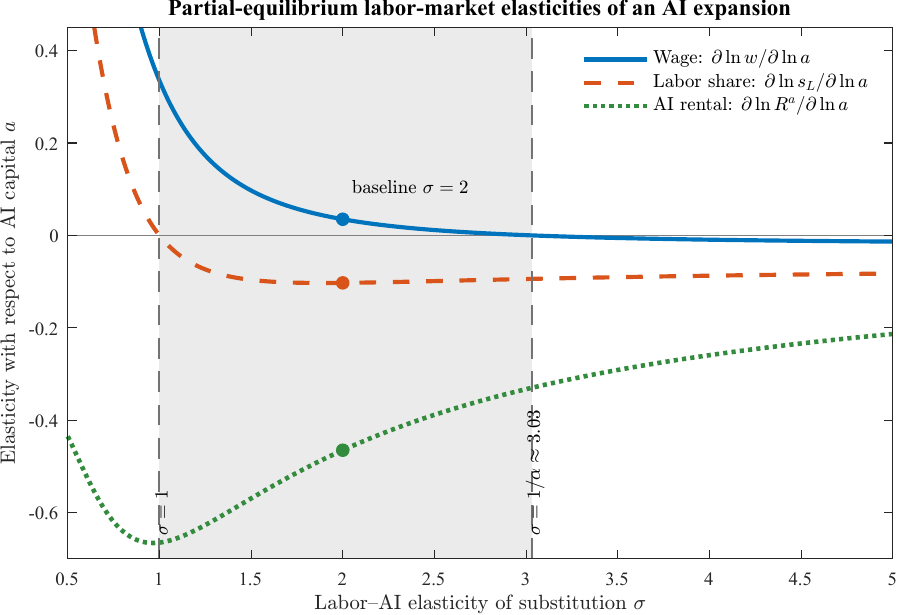}
\caption{Partial-equilibrium labor-market elasticities of an AI expansion against
the labor--AI elasticity $\sigma$. The wage elasticity (solid) changes sign at
$\sigma=1/\alpha\approx3.03$ and the labor-share elasticity (dashed) at $\sigma=1$;
the AI rental (dotted) is negative throughout. The shaded band $1<\sigma<1/\alpha$
contains the baseline $\sigma=2$.}
\label{fig:labor}
\end{figure}

The two-threshold structure clarifies what the model does and does not say about
automation and labor. It does not predict that AI depresses wages at empirically
plausible elasticities: the wage falls only beyond $\sigma=1/\alpha$, outside the
range of existing estimates. It does, however, predict that AI erodes the labor
share for any $\sigma>1$, a far weaker and empirically well-supported condition.
The apparent tension between ``AI is good for workers'' (wages rise) and ``AI
shifts income to capital'' (the labor share falls) is thus resolved within a
single technology: both hold simultaneously throughout the band that contains the
calibration, because output rises faster than the wage. This decoupling is the
labor-market counterpart of the capital-demand interpretation of the AI shock
developed in Section~\ref{sec:irf}.

\subsection{Productivity and displacement: a structural decomposition}
\label{sec:prod_disp}
The wage elasticity admits an interpretation in terms of the two canonical channels
of the task-based literature \citep{autor2003skill, acemoglu2018race,
acemoglu2019automation}: a \emph{productivity} effect, by which automation lowers
costs, expands output, and raises labor demand, and a \emph{displacement} effect, by
which capital substitutes for labor in automatable tasks. In the present model these
two forces are not assumed; they emerge as the two additive terms of the closed-form
wage elasticity and are individually measurable.

\begin{proposition}[Productivity effect]
\label{prop_prod}
An AI expansion raises average labor productivity for \emph{every} value of the
labor--AI elasticity:
\[
\frac{\partial\ln(Y_t/H_t)}{\partial\ln a_t}=(1-\alpha)\,s_{a}>0 .
\]
Output per worker therefore rises with automation regardless of whether AI
complements or substitutes for labor.
\end{proposition}
\begin{proof}
Holding $H_t$ fixed, $\partial\ln(Y_t/H_t)/\partial\ln a_t=\partial\ln Y_t/\partial\ln a_t
=(1-\alpha)s_a$ by Lemma~\ref{lem:labor}; positivity follows from $\alpha<1$, $s_a>0$.
\end{proof}

\begin{proposition}[Productivity--displacement decomposition of the wage]
\label{prop_proddisp}
The wage response to AI is the difference between the productivity effect and a
displacement effect,
\begin{equation}
\frac{\partial\ln w_t}{\partial\ln a_t}
=\underbrace{(1-\alpha)\,s_{a}}_{\text{productivity effect}}
-\underbrace{\rho\,s_{a}}_{\text{displacement effect}},
\label{eq:proddisp}
\end{equation}
where the productivity effect equals the rise in output per worker
(Proposition~\ref{prop_prod}) and the displacement effect equals the erosion of
labor's income share, $\rho\,s_a=-\partial\ln s_{L}/\partial\ln a_t$. The wage rises
if and only if the productivity effect dominates, $(1-\alpha)>\rho$, that is
$\sigma<1/\alpha$.
\end{proposition}
\begin{proof}
Immediate from Lemma~\ref{lem:labor}, writing
$(1-\alpha-\rho)s_a=(1-\alpha)s_a-\rho s_a$; the displacement term coincides with
$-\partial\ln s_{L}/\partial\ln a_t$ by the same lemma.
\end{proof}

\begin{corollary}[The productivity--pay gap]
\label{cor_gap}
Automation opens a gap between labor productivity, which always rises, and the real
wage, which rises only under complementarity. The gap equals the displacement effect
and the decline in the labor share,
\[
\frac{\partial\ln(Y_t/H_t)}{\partial\ln a_t}-\frac{\partial\ln w_t}{\partial\ln a_t}
=\rho\,s_{a}=-\frac{\partial\ln s_{L}}{\partial\ln a_t} .
\]
\end{corollary}
\begin{proof}
Subtract Proposition~\ref{prop_proddisp} from Proposition~\ref{prop_prod} and use
$-\partial\ln s_L/\partial\ln a_t=\rho s_a$.
\end{proof}

At the extensive margin, the strength of displacement is governed, as everywhere in
the model, by the labor--AI elasticity. In the perfect-substitute limit of
Corollary~\ref{prop6} one efficiency unit of AI capital replaces $\phi$ units of
labor and the displacement effect overwhelms the productivity effect, so the wage
falls; under the baseline complementarity the displacement effect is present but
dominated, so automation raises both output per worker and the wage while still
shifting income toward capital. Figure~\ref{fig:proddisp} plots the two effects and
their net across $\sigma$: they are equal at the knife-edge $\sigma=1/\alpha$, where
the wage response vanishes, and the productivity effect dominates throughout the
empirically relevant band that contains the baseline.

\begin{table}[!ht]
\centering
\caption{Productivity--displacement decomposition of the wage response to a
ten-percent AI expansion at the baseline calibration ($\sigma=2$, $s_a=0.206$).}
\label{tab:proddisp}
\begin{tabular}{l c r r}
\toprule
Channel & Elasticity & Value & Response to $+10\%$ AI \\
\midrule
Productivity effect $(1-\alpha)s_a$  & $\partial\ln(Y/H)/\partial\ln a$ & $+0.138$ & $+1.38\%$ \\
Displacement effect $\rho\,s_a$      & $-\partial\ln s_L/\partial\ln a$ & $-0.103$ & $-1.03\%$ \\
\midrule
Net wage effect $(1-\alpha-\rho)s_a$ & $\partial\ln w/\partial\ln a$    & $+0.035$ & $+0.35\%$ \\
Productivity--pay gap $\rho\,s_a$     & productivity $-$ wage            & $+0.103$ & $+1.03\%$ \\
\bottomrule
\end{tabular}
\end{table}

\begin{figure}[!ht]
\centering
\includegraphics[width=0.82\textwidth]{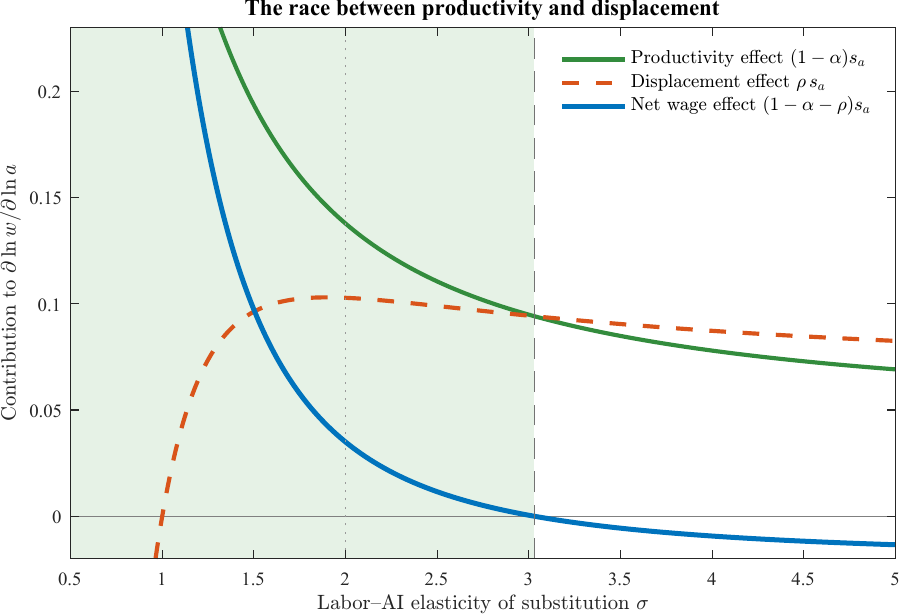}
\caption{The race between productivity and displacement. The wage response (solid)
is the productivity effect $(1-\alpha)s_a$ minus the displacement effect
$\rho\,s_a$; the two are equal at $\sigma=1/\alpha\approx3.03$. The shaded region
$\sigma<1/\alpha$, where the wage rises, contains the baseline $\sigma=2$.}
\label{fig:proddisp}
\end{figure}

This decomposition gives the model a direct empirical interpretation. The
productivity effect is the structural counterpart of the documented
output-per-worker gains from automation \citep{graetz2018robots} and, more
recently, from generative AI, where field evidence reports sizable productivity
gains concentrated among less-experienced workers \citep{brynjolfsson2025genai};
the displacement effect is the counterpart of the task-substitution channel
emphasized by \citet{acemoglu2019automation} and \citet{autor2015jobs}. The model's contribution
is to show that, within a single calibrated technology, the same elasticity that
governs the sign of the wage response also fixes the relative magnitude of the two
effects, so that the productivity--pay gap, the wedge between rising output per
worker and a more slowly rising wage, is exactly the erosion of the labor share.

\section{Welfare analysis}
\label{sec:welfare}
The positive analysis establishes that the two disturbances move returns,
fertility, and output in systematically opposite ways. We now ask how they
map into \emph{welfare}. We adopt the household's expected lifetime utility as
the welfare criterion, derive an exact decomposition of the welfare effect of a
structural shock into income, return, and longevity channels, and relate the two
shocks to the dynamic efficiency of the economy.

\begin{definition}[Cohort welfare and the consumption-equivalent measure]
\label{def:welfare}
The ex ante expected lifetime welfare of the cohort that is young in period $t$,
evaluated at the equilibrium allocation, is
\[
W_t \;\equiv\; \log c_t^{w}+\log n_t+\beta\,\gamma_{t+1}\,\log c_{t+1}^{r}.
\]
A utilitarian planner ranks equilibria by
$\mathcal{W}=\sum_{t\ge0}\Delta^{t}\,W_t$, with cohort weights
$\Delta\in(0,1)$. For a perturbation that changes welfare by $\mathrm{d}W_t$, the
\emph{consumption-equivalent variation} $\omega_t$ is the permanent proportional
change in lifetime consumption, scaling both $c_t^{w}$ and $c_{t+1}^{r}$ by
$1+\omega_t$, that reproduces the same welfare change. Under logarithmic
preferences,
\[
\omega_t=\exp\!\left(\frac{\mathrm{d}W_t}{1+\beta\gamma_{t+1}}\right)-1 ,
\]
since scaling both consumptions by $1+\omega_t$ raises $W_t$ by
$(1+\beta\gamma_{t+1})\log(1+\omega_t)$. The derivation is given in
Appendix~\ref{app:cev}.
\end{definition}

The welfare effect of a structural shock decomposes exactly into income, return,
and longevity channels, obtained from an envelope representation of marginal welfare
(Lemma~\ref{lem:envelope}, stated and proved in Appendix~\ref{app:proofs}).

\begin{proposition}[Welfare decomposition of a structural shock]
\label{prop_welfdecomp}
For a structural disturbance $\xi\in\{\phi_t,\gamma_{t+1}\}$, the total effect on
cohort welfare is
\[
\frac{\mathrm{d}W_t}{\mathrm{d}\xi}
=\underbrace{\lambda_t h_t\,\frac{\mathrm{d}w_t}{\mathrm{d}\xi}
+\lambda_t\,\frac{\mathrm{d}d_t}{\mathrm{d}\xi}}_{\text{income channel}}
+\underbrace{\lambda_t\,\frac{\gamma_{t+1}c_{t+1}^{r}}{R_{t+1}^{2}}\,
\frac{\mathrm{d}R_{t+1}}{\mathrm{d}\xi}}_{\text{return channel}}
+\underbrace{\Big(\beta\log c_{t+1}^{r}-\lambda_t\tfrac{c_{t+1}^{r}}{R_{t+1}}\Big)
\frac{\mathrm{d}\gamma_{t+1}}{\mathrm{d}\xi}}_{\text{longevity channel}} .
\]
\end{proposition}
\begin{proof}
See Appendix~\ref{app:welfdecomp}.
\end{proof}

\begin{proposition}[Signs of the welfare effects]
\label{prop_welfsign}
Under the baseline complementarity regime $(\rho<1-\alpha)$:
\emph{(a)} the AI shock is welfare-improving, $\mathrm{d}W_t/\mathrm{d}\phi_t>0$,
because $\mathrm{d}\gamma_{t+1}/\mathrm{d}\phi_t=0$ while
$\mathrm{d}w_t/\mathrm{d}\phi_t,\ \mathrm{d}d_t/\mathrm{d}\phi_t,\
\mathrm{d}R_{t+1}/\mathrm{d}\phi_t>0$ (Propositions~\ref{prop4}--\ref{prop5});
\emph{(b)} the welfare effect of the longevity shock is of ambiguous sign: the
value of a longer life (the term $\beta\log c_{t+1}^{r}$, positive provided
retirement consumption exceeds the felicity's unit reference, $c_{t+1}^{r}>1$)
and the capital-deepening wage gain
($\mathrm{d}w_t/\mathrm{d}\gamma_{t+1}>0$) pull welfare up, whereas the cost of
financing retirement and the compression of returns
($\mathrm{d}R_{t+1}/\mathrm{d}\gamma_{t+1}<0$) pull it down. Longevity raises
cohort welfare if and only if
\[
\beta\log c_{t+1}^{r}+\lambda_t h_t\,\frac{\mathrm{d}w_t}{\mathrm{d}\gamma_{t+1}}
\;>\;
\lambda_t\frac{c_{t+1}^{r}}{R_{t+1}}
+\lambda_t\frac{\gamma_{t+1}c_{t+1}^{r}}{R_{t+1}^{2}}\Big|\tfrac{\mathrm{d}R_{t+1}}{\mathrm{d}\gamma_{t+1}}\Big|
-\lambda_t\frac{\mathrm{d}d_t}{\mathrm{d}\gamma_{t+1}} .
\]
\end{proposition}
\begin{proof}
See Appendix~\ref{app:welfsign}.
\end{proof}

\begin{proposition}[The two shocks and dynamic efficiency]
\label{prop_dyneff}
Along a balanced path on which aggregate capital and output grow at the gross
population growth rate $n_t$, the competitive equilibrium is dynamically efficient
if and only if the gross return weakly exceeds that growth rate, $R_t\ge n_t$
\citep{diamond1965national}. The marginal effect of a structural shock $\xi$ on the
efficiency margin decomposes as
\[
\frac{\mathrm{d}(R_t-n_t)}{\mathrm{d}\xi}
=\frac{\mathrm{d}R_t}{\mathrm{d}\xi}-\frac{\mathrm{d}n_t}{\mathrm{d}\xi}.
\]
For the AI shock both terms are nonnegative
($\mathrm{d}R_t/\mathrm{d}\phi_t>0$, $\mathrm{d}n_t/\mathrm{d}\phi_t\ge0$), so the
margin widens if and only if the return response dominates the fertility response,
$\mathrm{d}R_t/\mathrm{d}\phi_t>\mathrm{d}n_t/\mathrm{d}\phi_t$, which holds in the
baseline calibration where the return to capital is the most responsive variable.
For the longevity shock both terms are nonpositive
($\mathrm{d}R_t/\mathrm{d}\gamma_{t+1}<0$, $\mathrm{d}n_t/\mathrm{d}\gamma_{t+1}<0$),
so the net effect on the margin is ambiguous in general; the saving-supply force
nonetheless compresses $R_t$ toward the growth rate, the channel through which
population aging threatens dynamic efficiency.
\end{proposition}
\begin{proof}
See Appendix~\ref{app:dyneff}.
\end{proof}

\begin{remark}
Because parents bear the full time cost of child-rearing, $\kappa w_t$ per child,
the fertility margin carries no static externality in this economy: the private
and social marginal costs of a child coincide. The welfare-relevant friction is
therefore intertemporal and operates through Proposition~\ref{prop_dyneff}: it is
the saving-supply force unleashed by longevity, not the fertility decision per se,
that can push the economy toward dynamic inefficiency.
\end{remark}

The decomposition makes precise the sense in which the two structural forces are
not symmetric in welfare terms. The AI shock is unambiguously welfare-improving in
the complementarity regime and pushes the economy deeper into the dynamically
efficient region; the longevity shock confers the first-order benefit of a longer
life but, by compressing returns and deepening capital, simultaneously raises the
cost of financing retirement and compresses the return toward the economy's growth
rate, so that its net welfare effect must be settled quantitatively.
Figure~\ref{fig:welfare} reports the calibrated counterpart of
Definition~\ref{def:welfare}: the consumption-equivalent welfare paths
$\omega_t$ implied by each shock along its transition. In line with
Proposition~\ref{prop_welfsign}, the AI shock delivers a small but uniformly
positive welfare gain that decays monotonically, small because AI operates chiefly
through asset prices rather than through the consumption and fertility quantities
that enter welfare. The longevity shock, by contrast, produces a short-run welfare
loss roughly an order of magnitude larger, as households compress worker
consumption and fertility to finance the longer retirement; this loss recovers
gradually as the induced capital deepening raises wages. The calibrated outcome
thus resolves the ambiguity of Proposition~\ref{prop_welfsign}(b) on impact in
favor of the cost-of-financing and return-compression channels, and the two paths
are plotted on separate vertical scales because the AI effect is far smaller.

\begin{figure}[!ht]
\centering
\IfFileExists{welfare.pdf}{%
  \includegraphics[width=0.72\textwidth]{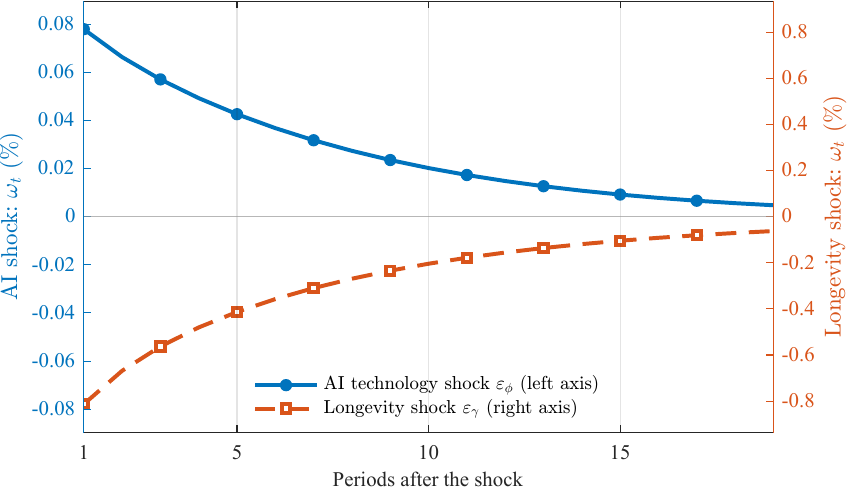}}{%
  \fbox{\parbox[c][3.2cm][c]{0.72\textwidth}{\centering\itshape
  welfare.pdf: run \texttt{make\_welfare\_figure.m} (or
  \texttt{make\_all\_figures.m}) to generate this figure.}}}
\caption{Consumption-equivalent welfare paths $\omega_t$ (percent,
Definition~\ref{def:welfare}) for the AI shock $\varepsilon_{\phi}$ (solid, left
axis) and the longevity shock $\varepsilon_{\gamma}$ (dashed, right axis). The
scales differ because the AI welfare effect is an order of magnitude smaller.}
\label{fig:welfare}
\end{figure}

\section{Robustness analysis}
\label{sec:robustness}

We assess the sensitivity of the impulse responses to a positive AI
productivity shock along three structural margins: the labor--AI elasticity
of substitution $\sigma$, the physical-capital share $\alpha$, and the
persistence of the AI shock $\rho_{\phi}$. The first governs the
substitution--complementarity distinction that determines the sign of the wage
and fertility responses; the second alters the relative weight of physical
capital in production; the third varies the half-life of the technological
impulse. For each exercise, we re-solve the model over a grid of values,
hold all remaining parameters at their baseline, and report the impulse
responses of six headline variables: the real wage, fertility, output,
aggregate consumption, and the rentals of AI and physical capital.
Throughout, series are identified by three redundant attributes: color
(from a colorblind-safe palette), line style, and marker shape; the figures
therefore remain readable under black-and-white reproduction. In the
bottom-row bar charts, fills additionally carry a light-to-dark grayscale
gradient and their edges retain the line colors, so that each parameter
value is identifiable by both shade and edge. Because the $\sigma$ exercise
governs whether AI complements or substitutes for labor, we discuss it in
detail here; the supporting exercises on $\alpha$ and $\rho_{\phi}$ are
reported in Appendix~\ref{app:robust}.
\begin{figure}[!htbp]
\centering
\includegraphics[width=\textwidth,height=0.85\textheight,keepaspectratio]%
{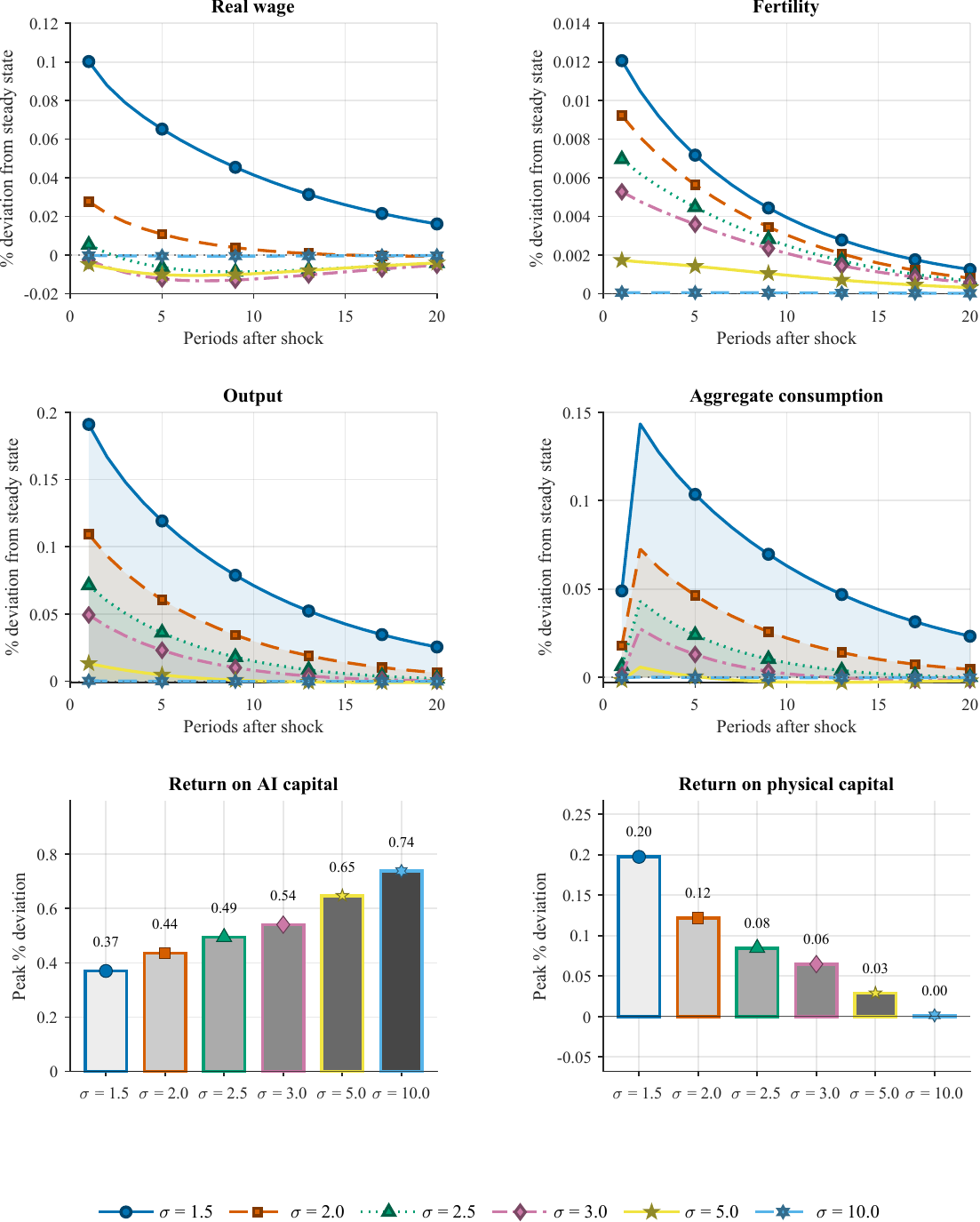}
\caption{Impulse responses to a positive AI shock for alternative values of
the labor--AI elasticity $\sigma$. Top: real wage and fertility. Middle: output
and aggregate consumption. Bottom: peak responses of the two capital rentals.}
\label{fig:robust_sigma}
\end{figure}
\subsection{The labor--AI elasticity of substitution}
The qualitative response to an AI shock depends on whether labor and AI
capital are Edgeworth complements or substitutes in production. The
threshold separating the two regimes is $\sigma=1/\alpha\approx 3.03$; the baseline value $\sigma=2$ places
the economy in the complementarity regime. We sweep
$\sigma \in \{1.5, 2, \ldots, 10\}$, spanning the
complementarity regime ($\sigma<1/\alpha$), the knife-edge
($\sigma\approx 1/\alpha$), and the substitution regime
($\sigma>1/\alpha$). This grid provides a direct test of the analytical
claim of Corollary~\ref{prop6}: as $\sigma$ crosses $1/\alpha$, the wage
and fertility responses to AI reverse sign.

Three patterns emerge from Figure~\ref{fig:robust_sigma}. First, the wage
response is positive and large for low $\sigma$, declines monotonically as
$\sigma$ rises, vanishes near the knife-edge $\sigma\approx 3.03$, and
turns negative once $\sigma>1/\alpha$; the decline is most pronounced at
intermediate elasticities and then attenuates toward zero in the
perfect-substitute limit, where the wage becomes insensitive to the AI
shock because labor's marginal product is pinned down by the near-linear
technology. This sign reversal is precisely the prediction of
Proposition~\ref{prop5} and Corollary~\ref{prop6}: AI raises the marginal
product of labor under complementarity and depresses it under substitution.
Second, the fertility response tracks the same wage channel. Because the
income effect that sustains fertility operates through the AI-induced wage
gain, the modest baseline rise in fertility weakens monotonically as
$\sigma$ increases and becomes negligible in the perfect-substitute limit,
where the AI shock no longer raises lifetime wealth on net and households
cease to reallocate time and saving toward children. Fertility thus serves
as a sensitive diagnostic of the prevailing regime. Third, output is positive
at impact for every $\sigma$, since the AI shock raises the marginal
product of AI capital regardless of the nest; the response path flattens
as $\sigma$ rises because the labor-side amplification through wages and
labor demand weakens. The bottom row reports the peak responses of the two
rental rates. The rental of AI capital jumps sharply on impact in every
case, and its peak response rises in $\sigma$, from about $0.37\%$ at
$\sigma=1.5$ to $0.74\%$ at $\sigma=10$, as production leans more heavily on
the AI factor; the peak response of the physical-capital rental moves in the
opposite direction, falling steadily from about $0.20\%$ to nearly zero as
the marginal product is reallocated toward AI capital when the two factors
become closer substitutes for labor. The exercise confirms that the substantive
results of the paper derive from the complementarity regime that the
baseline calibration occupies; the labor-replacing regime is internally
consistent and exhibits the sign patterns predicted by
Corollary~\ref{prop6}, but it lies outside the range of empirical
estimates of $\sigma$. We stress that this elasticity dependence is a
deliberate and transparent feature of the framework rather than a hidden
fragility: the model nests both regimes within a single technology, locates
the knife-edge analytically at $\sigma=1/\alpha$, and characterizes the
behavior on either side, so that a reader can map any preferred empirical value
of the labor--AI elasticity directly into the model's predictions.

The complementary exercises reported in Appendix~\ref{app:robust} confirm
that the signs of the wage, fertility, output, and consumption responses
to a positive AI shock are robust to plausible variation in the capital
share $\alpha$ and in the persistence $\rho_{\phi}$ of the shock; only the
labor--AI elasticity of substitution can reverse them, and only beyond the
empirically supported range.

\subsection{Robustness to the production structure}
\label{sec:robust_nest}
A natural concern is that the substitution--complementarity distinction is an
artifact of the Cobb--Douglas outer nest, which fixes the physical-capital share
at $\alpha$ and imposes a unit elasticity between physical capital and the
labor--AI composite. We address this directly by embedding the technology in a
strictly more general two-level CES and showing that the complementarity
threshold survives, merely relocating to a transparent function of two separately
disciplinable elasticities.

\begin{proposition}[Generality of the complementarity threshold]
\label{prop_outer}
Replace the Cobb--Douglas outer aggregator with a CES nest of physical capital and
the labor--AI composite,
\[
Y_t=\Big[\alpha\,K_t^{\eta}+(1-\alpha)\,G_t^{\eta}\Big]^{1/\eta},
\qquad
G_t=\big[H_t^{\rho}+(\phi_t a_t)^{\rho}\big]^{1/\rho},
\]
with outer elasticity $\sigma_o=1/(1-\eta)$ between physical capital and the
composite $G_t$, and inner elasticity $\sigma=1/(1-\rho)$ between labor and AI
capital. Let
$s_{K}\equiv \alpha K_t^{\eta}\big/\!\big[\alpha K_t^{\eta}+(1-\alpha)G_t^{\eta}\big]$
be the equilibrium value share of physical capital. Then AI capital is an
Edgeworth complement to labor, $\partial^2 Y_t/\partial H_t\,\partial a_t>0$, so
that a positive AI shock raises the marginal product of labor and the wage, if and
only if
\[
\frac{1}{\sigma}>\frac{s_{K}}{\sigma_o}
\qquad\Longleftrightarrow\qquad
\sigma<\sigma^{*}\equiv\frac{\sigma_o}{s_{K}} .
\]
The baseline is the Cobb--Douglas limit $\sigma_o\to1$, in which $s_K\to\alpha$ is
constant and the threshold collapses to $\sigma^{*}=1/\alpha$.
\end{proposition}
\begin{proof}
See Appendix~\ref{app:outer}.
\end{proof}

Proposition~\ref{prop_outer} shows that the knife-edge is governed not by the
Cobb--Douglas assumption but by the ratio of the outer elasticity to the
physical-capital share. Evaluated at the baseline factor share
($s_K\approx\alpha=0.33$), the boundary is $\sigma^{*}=\sigma_o/s_K\approx
3.03\,\sigma_o$: it equals $2.42$ for $\sigma_o=0.8$, $3.03$ for the Cobb--Douglas
case $\sigma_o=1$, and $3.64$ for $\sigma_o=1.2$. The baseline economy
($\sigma=2$) therefore remains in the complementarity regime for every outer
elasticity $\sigma_o>\sigma\,s_K\approx0.67$, that is, for the Cobb--Douglas nest
and for the upper part of the empirical range of the capital--labor substitution
elasticity. Only when physical capital and the labor--AI composite are strong
gross complements, $\sigma_o$ below roughly two-thirds, does AI become
labor-replacing at $\sigma=2$. The qualitative conclusions of the positive
analysis thus rest on two transparent and independently estimable elasticities,
$\sigma$ and $\sigma_o$, rather than on the functional form of the outer nest.

\subsection{Relation to functional forms in the AI macroeconomics literature}
\label{sec:robust_litforms}
The nested CES nests, as special cases, the production technologies used elsewhere in
the macroeconomics of automation, so the robustness of the paper's conclusions to
those alternatives can be read directly off the elasticity grid above. Three
benchmark specifications recur in that literature. First, \emph{automation capital as
a perfect substitute for labor}: \citet{prettner2019note} and \citet{lankisch2019can}
write output as $Y_t=K_t^{\alpha}\big(H_t+\phi_t a_t\big)^{1-\alpha}$, with AI capital
and labor perfectly substitutable. This is exactly the limit $\sigma\to\infty$ of our
inner nest (Corollary~\ref{prop6}); it places the economy in the labor-replacing
regime and reverses the wage and fertility responses. Our framework reproduces this
case but shows that it requires an elasticity far above existing estimates of the
labor--automation elasticity. Second, \emph{task-based automation}: the framework of
\citet{acemoglu2018race} and \citet{acemoglu2019automation} aggregates a continuum of
tasks performed by capital or labor and, at the aggregate level, reduces to a CES
between labor and automation capital whose elasticity is governed by the share of
automated tasks, which $\sigma$ parameterizes directly. Third, \emph{labor-augmenting
AI}, $Y_t=K_t^{\alpha}\big(\phi_t a_t H_t\big)^{1-\alpha}$, in which AI raises labor
productivity multiplicatively; here $\partial^2 Y_t/\partial H_t\partial a_t>0$ for
every parameter value, so an AI shock raises the marginal product of labor
unconditionally, strengthening rather than challenging the complementarity result.
Confronting the model with these alternative forms therefore leaves the headline
mechanism, the opposite-signed transmission of the AI and longevity shocks through the
return to capital, intact: only the perfect-substitute specification overturns the
wage and fertility signs, and it does so outside the empirically relevant range of the
labor--AI elasticity.

\subsection{Sensitivity to the utility functional form}
\label{sec:robust_utility}
The exercises above vary technology and shock parameters while holding
preferences at the logarithmic benchmark $u(c^{w},n)=\log c^{w}+\log n$,
$v(c^{r})=\log c^{r}$. Because the logarithmic case is a knife-edge along
several margins, implying a unit intertemporal elasticity of substitution
(IES) and a unit-elastic demand for children, we now separate the conclusions
that are properties of preferences from those that are structural. We consider three
departures from the benchmark: (i) constant-relative-risk-aversion (CRRA)
felicity over consumption, $v(c^{r})=(c^{r})^{1-\theta}/(1-\theta)$ with
$\theta>0$, $\theta\neq1$, which detaches the IES $1/\theta$ from unity;
(ii) iso-elastic felicity over fertility,
$u(c^{w},n)=\log c^{w}+\chi\,n^{1-\eta}/(1-\eta)$, which detaches the own-price
elasticity of fertility demand from $-1$; and (iii) Greenwood--Hercowitz--Huffman
(GHH) non-separability, which removes the wealth effect on the time-allocation
margin.

Two conclusions are invariant across these specifications, and one is genuinely
knife-edge. First, the demand-versus-supply asymmetry that organizes the
paper, namely the AI shock raising the return to capital and the longevity shock
compressing it, is a property of the capital market and the production block,
not of preferences. We record this formally.
\begin{proposition}[Preference-robustness of the demand--supply asymmetry]
\label{prop_robpref}
Consider any preferences $U(c^{w},n)+\beta\gamma_{t+1}V(c^{r})$ with $U,V\in C^{1}$,
strictly increasing and strictly quasi-concave, for which consumption and the
number of children are normal goods and the retirement-saving motive is increasing
in survival, $\partial s_t^{w}/\partial\gamma_{t+1}>0$. Then, at an interior
equilibrium with capital-market clearing $s_t^{w}=q_t^{k}k_{t+1}+q_t^{a}a_{t+1}$,
the gross return on saving satisfies
\[
\frac{\partial R_t}{\partial \phi_t}>0,
\qquad
\frac{\partial R_t}{\partial \gamma_{t+1}}<0 .
\]
The preference parameters $(\theta,\eta,\chi)$ affect the magnitude of these
responses but not their sign.
\end{proposition}
\begin{proof}
See Appendix~\ref{app:robpref}.
\end{proof}

Second, the \emph{opposite-sign} fertility responses to the two shocks survive
for the whole class of preferences in which children are a normal good carrying a
time cost: the AI shock raises fertility through the income effect of a higher
wage (Proposition~\ref{prop4}), while the longevity shock lowers it by
strengthening the life-cycle saving motive and raising the opportunity cost of
child-rearing (Proposition~\ref{proplonge}). Neither mechanism relies on the
logarithmic form. Third, the one result that is specific to the benchmark is the
\emph{near-perfect} negative correlation between fertility and hours. By
Lemma~\ref{prop7}, the time-cost channel vanishes exactly when fertility demand is
unit-elastic ($\varepsilon=-1$), which is the logarithmic case $\eta=0$; for
$\eta>0$ (inelastic demand) the fertility--hours link weakens, and for $\eta<0$
(elastic demand) it strengthens. The GHH specification likewise attenuates the
AI-to-fertility channel by muting the wealth effect, but leaves the
longevity channel and the return asymmetry intact. We therefore read the unit
fertility--hours correlation as the sharp benchmark of a more general,
same-signed pattern, exactly as anticipated in the Discussion. Taken together,
these exercises, spanning technology, shock persistence, and the preference
specification, confirm that the paper's main conclusions are structural
properties of the model rather than artifacts of fine-tuning.

\section{Discussion}
\label{sec:discussion}
The analysis brings together three strands of macroeconomic research that are
typically pursued in isolation. The first is the macroeconomics of automation and
artificial intelligence, concerned with how capital-augmenting and task-replacing
technologies reshape factor returns and the functional distribution of income
\citep{acemoglu2018race, aghion_ai_growth}. The second is the economics of
endogenous fertility, organized around the trade-off between the number of children
and the parental time they require \citep{becker_fertility, becker_murphy_tamura},
which the unified growth literature places at the center of the demographic
transition \citep{galor_weil_2000}. The third is the life-cycle and
overlapping-generations literature on longevity, in which a longer expected lifespan
strengthens the saving motive and deepens the capital stock
\citep{longevity_saving, delacroix_licandro_1999, cervellati_sunde_2005}. By
embedding all three within a single general equilibrium model with two types of
capital, our framework lets technological and demographic forces interact rather
than operate along separate margins; that interaction is the source of the paper's
central result.

That result turns on the return to capital. Both shocks are transmitted through the
same equilibrium price, but from opposite sides of the market: the AI shock shifts
the demand for capital outward, raising returns and the real interest rate, whereas
the longevity shock shifts the supply of capital outward, pushing them down. This
demand-versus-supply logic explains why two apparently unrelated disturbances leave
opposite imprints not only on returns but also on fertility, which responds jointly
to the wage, the opportunity cost of children \citep{becker_fertility}, and the
saving motive. Seen in this light, the model delivers a unified account of two
phenomena usually discussed separately: the prospect that AI sustains the return to
capital \citep{aghion_ai_growth}, and the secular decline in real interest rates
associated with population aging \citep{eggertsson_secular_2019,
gagnon_demographics_2021}. Under our calibration the second force prevails, so the
equilibrium interest rate falls on balance, consistent with the low-$r$ environment
documented for aging economies \citep{gagnon_demographics_2021}.

Several features of the model qualify these conclusions and indicate where they
are most and least robust. The sign of the fertility and wage responses to the AI
shock hinges on whether AI capital complements labor in the aggregate; were the two
strongly substitutable, the marginal product of labor and the wage could fall
\citep{acemoglu2018race}, and the income effect that supports fertility would
weaken or reverse. As stressed above, this dependence is not a hidden fragility but
the model's central comparative-static margin, reported transparently for both
regimes. The near-perfect negative
correlation between fertility and hours worked is, in turn, a direct expression of
the household's time-allocation structure, in which market work and child-rearing
compete for a fixed time endowment \citep{becker_fertility}; richer models with home
production or market childcare would attenuate this link. Finally, the variance
decomposition holds the two shocks to a common standard deviation. The finding that
longevity dominates aggregate volatility is therefore a statement about the model's
transmission mechanism, not evidence that longevity disturbances are empirically
larger than technological ones. Disciplining the relative magnitudes and frequencies
of the two shocks with data is a necessary step before drawing firm conclusions
about their historical contributions.

The two shocks are also asymmetric in welfare terms. In
the complementarity regime the AI shock is unambiguously welfare-improving and pushes
the economy deeper into the dynamically efficient region, whereas longevity confers
the first-order benefit of a longer life but, by compressing returns and deepening
capital, raises the cost of financing retirement and compresses the return $R_t$
toward the economy's growth rate $n_t$ \citep{diamond1965national}. In the calibrated
model these opposing forces resolve, on impact, in favor of the cost-of-financing and
return-compression channels: longevity produces a pronounced short-run welfare loss,
an order of magnitude larger than the small, uniformly positive gain from the AI shock,
as households cut worker consumption and fertility to finance the longer retirement;
the loss then recedes as the induced capital deepening raises wages. This links the model to the secular-stagnation reading of
population aging \citep{eggertsson_secular_2019, gagnon_demographics_2021} from the
supply side of capital, and supplies a transparent consumption-equivalent metric
(Definition~\ref{def:welfare}) that ranks the two disturbances along the transition.

These qualifications notwithstanding, the analysis carries a clear, policy-relevant
message. Because the fertility decline operates through the reallocation of
household time and saving rather than through technology per se, policies that lower
the time cost of children are likely to be more effective at sustaining fertility
than interventions aimed at the pace of automation. More broadly, the framework
implies that the macroeconomic consequences of AI and population aging cannot be
assessed in isolation: their effects on capital returns partly offset one another,
whereas their effects on fertility compound in the same direction, with first-order
implications for the long-run trajectory of the old-age dependency ratio
\citep{galor_weil_2000}.
\section{Conclusion}

This paper has analyzed the joint macroeconomic effects of artificial
intelligence and rising longevity in a unified general equilibrium model with a
life-cycle structure, endogenous fertility, and two forms of capital, physical
and AI. The central finding is that the two forces transmit through the same
equilibrium price from opposite sides of the market: the AI shock is a
capital-demand disturbance that raises returns, most sharply on AI capital,
reallocates investment toward that asset, and lifts output on impact, whereas the
longevity shock is a saving-supply disturbance that deepens the capital stock,
compresses returns and the real interest rate, and propagates through hump-shaped
dynamics. The same asymmetry moves fertility in opposite directions, mildly upward
under AI through an income effect and downward under longevity, so that the old-age
dependency ratio is the most responsive variable in the system; fertility is
robustly countercyclical and almost perfectly negatively correlated with hours,
placing household time allocation at the center of the mechanism.

Quantitatively, longevity is the dominant source of fluctuations, accounting for
the bulk of the variance of most aggregates, while AI operates chiefly as an
asset-return shock concentrated on the price of capital, particularly the return
to AI capital. These conclusions are structural rather than calibrated: across the
capital share, the shock persistence, the utility specification, and a
generalized CES production nest, only an empirically implausible labor--AI
elasticity, beyond the threshold $\sigma^{*}=\sigma_o/s_K$ (equal to
$1/\alpha\approx 3.03$ in the baseline), can reverse the wage and fertility
signs. A welfare analysis sharpens the
asymmetry: under complementarity the AI shock yields a small, uniformly positive
welfare gain and moves the economy deeper into the dynamically efficient region,
whereas the longevity shock trades the value of a longer life against costlier
retirement financing and compressed returns, pushing $R_t$ toward the growth rate
$n_t$. In the calibrated model the latter forces dominate on impact, so longevity
generates a short-run welfare loss, an order of magnitude larger than the AI gain,
that recovers only gradually as capital deepening raises wages.

The policy implication is direct: because the fertility decline operates through
the reallocation of household time and saving rather than through technology
itself, policies that lower the time cost of children are likely to be more
effective at sustaining fertility than interventions aimed at the pace of
automation. Introducing skill heterogeneity, labor-market frictions, and a richer
treatment of AI--labor complementarity, and disciplining the relative size and
frequency of the two shocks with data, are natural next steps before assessing
their historical contributions.

\section*{Declaration of competing interest}
The author declares no known competing financial interests or personal
relationships that could have appeared to influence the work reported in
this paper.

\section*{Acknowledgments}
The author is grateful to seminar and conference participants for helpful comments
and suggestions. All remaining errors are the author's own.

\section*{Data availability}
This study is theoretical and computational and does not use external
empirical datasets. The replication code that generates the model solution,
impulse responses, variance decompositions, and robustness exercises is
available from the author upon reasonable request and will be deposited in a
public repository upon acceptance.

\section*{Declaration of generative AI and AI-assisted technologies in the
manuscript preparation process}
During the preparation of this work, the author used Claude (Anthropic), an
AI-assisted writing tool, for language editing, reference formatting, and
improving the clarity and structure of the exposition. The author reviewed and
edited the output as needed and takes full responsibility for the content of
the published article.

\bigskip
\addcontentsline{toc}{section}{References}
\bibliographystyle{plainnat}
\bibliography{reference}

\clearpage
\appendix
\numberwithin{equation}{section}
\numberwithin{table}{section}
\numberwithin{figure}{section}
\section{Proofs and additional comparative-statics results}
\label{app:proofs}
\subsection{Saving, consumption, and fertility in the household block}
\label{app:household}
\begin{proposition}
\label{prop1}
Let $s_t^w = w_t(1-\kappa\,n_t) + d_t - c_t^w$, where the equilibrium vector
$\Gamma_t^*=(c_t^w,n_t,c_{t+1}^r,\lambda_t)$ is implicitly defined by
$F(\Gamma_t,n_t,\gamma_{t+1})=0$, with $F\in C^1$ and
$\det\!\big(D_{\Gamma}F(\Gamma_t^*)\big)\neq 0$. Assume the Jacobian-based
monotonicity conditions
\[
e_1^{\top}\,D_{\Gamma}F(\Gamma_t^*)^{-1}\,F_n(\Gamma_t^*)>0,
\qquad
e_1^{\top}\,D_{\Gamma}F(\Gamma_t^*)^{-1}\,F_\gamma(\Gamma_t^*)<0,
\]
where $e_1=(1,0,0,0)^{\top}$. Then 
$\frac{\partial s_t^w}{\partial n_t}<0$
and $\frac{\partial s_t^w}{\partial \gamma_{t+1}}>0$.
\end{proposition}
\begin{proof}
Let $\Gamma_t^*=(c_t^w,n_t,c_{t+1}^r,\lambda_t)\in\mathbb{R}^4$ be implicitly defined by $F(\Gamma_t,n_t,\gamma_{t+1})=0$, where $F\in C^1$ and $D_{\Gamma}F(\Gamma_t^*)$ is invertible. By the Implicit Function Theorem, there exists a $C^1$ mapping $\Gamma_t=\Gamma_t(n_t,\gamma_{t+1})$ in a neighborhood of $(n_t,\gamma_{t+1})$ such that
\[
\frac{\partial \Gamma_t}{\partial n_t} = -D_{\Gamma}F(\Gamma_t^*)^{-1}F_n(\Gamma_t^*),
\qquad
\frac{\partial \Gamma_t}{\partial \gamma_{t+1}} = -D_{\Gamma}F(\Gamma_t^*)^{-1}F_\gamma(\Gamma_t^*).
\]

Define $e_1=(1,0,0,0)^\top$. Since $s_t^w=w_t(1-\kappa\, n_t)-c_t^w$, we have
\[
\frac{\partial s_t^w}{\partial n_t}
= -\kappa\, w_t - e_1^\top \frac{\partial \Gamma_t}{\partial n_t},
\qquad
\frac{\partial s_t^w}{\partial \gamma_{t+1}}
= - e_1^\top \frac{\partial \Gamma_t}{\partial \gamma_{t+1}}.
\]
Substituting the expressions for $\partial \Gamma_t/\partial n_t$ and $\partial \Gamma_t/\partial \gamma_{t+1}$ yields
\[
\frac{\partial s_t^w}{\partial n_t}
= -\kappa\, w_t + e_1^\top D_{\Gamma}F(\Gamma_t^*)^{-1}F_n(\Gamma_t^*),
\qquad
\frac{\partial s_t^w}{\partial \gamma_{t+1}}
= e_1^\top D_{\Gamma}F(\Gamma_t^*)^{-1}F_\gamma(\Gamma_t^*).
\]

The assumed sign restrictions
\[
e_1^\top D_{\Gamma}F(\Gamma_t^*)^{-1}F_n(\Gamma_t^*)>0,
\qquad
e_1^\top D_{\Gamma}F(\Gamma_t^*)^{-1}F_\gamma(\Gamma_t^*)<0
\]
imply the stated comparative statics.
\end{proof}

\begin{corollary}
\label{prop1b}
Under the assumptions of Proposition~\ref{prop1}, let
$\Gamma_t=(c_t^w,n_t,c_{t+1}^r,\lambda_t)$ be implicitly defined by
$F(\Gamma_t,n_t,\gamma_{t+1})=0$ with $D_{\Gamma}F(\Gamma_t^*)$ invertible, so
that $\Gamma_t=\Gamma_t(n_t,\gamma_{t+1})$ is a $C^1$ equilibrium manifold whose
total differential is
\[
d\Gamma_t=-\,D_{\Gamma}F(\Gamma_t^*)^{-1}
\big(F_{n}(\Gamma_t^*)\,dn_t+F_\gamma(\Gamma_t^*)\,d\gamma_{t+1}\big).
\]
Projecting this differential onto the consumption components and applying
Proposition~\ref{prop1} yields $dc_t^w<0$ and $dc_{t+1}^r<0$ whenever $dn_t>0$
and $d\gamma_{t+1}>0$.
\end{corollary}
\begin{proof}
Let $\Gamma_t=(c_t^w,n_t,c_{t+1}^r,\lambda_t)\in\mathbb{R}^4$ with equilibrium defined by $F(\Gamma_t,n_t,\gamma_{t+1})=0$, where $F\in C^1$ and $D_{\Gamma}F(\Gamma_t^*)$ is invertible. By the Implicit Function Theorem,
\[
d\Gamma_t=-D_{\Gamma}F(\Gamma_t^*)^{-1}\big(F_{n}(\Gamma_t^*)\,dn_t+F_\gamma(\Gamma_t^*)\,d\gamma_{t+1}\big).
\]
Let $\mathbf{c}_t=(c_t^w,c_{t+1}^r)^\top$ and $\Pi=\begin{pmatrix}1&0&0&0\\0&0&1&0\end{pmatrix}$. Then
\[
d\mathbf{c}_t=-\Pi D_{\Gamma}F(\Gamma_t^*)^{-1}\big(F_{n}(\Gamma_t^*)\,dn_t+F_\gamma(\Gamma_t^*)\,d\gamma_{t+1}\big).
\]
By Proposition \ref{prop1}, both projected directions are strictly positive, hence $d\mathbf{c}_t\ll 0$ for $dn_t>0$ and $d\gamma_{t+1}>0$, which implies the result.
\end{proof}

\begin{proposition}
\label{proplonge}
Let $F:\mathbb{R}^{4}\times\mathbb{R}\to\mathbb{R}^{4}$ be of class $C^{1}$, and
let $(\Gamma_t^{*},\gamma_{t+1}^{*})$ be a point satisfying
$F(\Gamma_t^{*},\gamma_{t+1}^{*})=\mathbf{0}$, where
$\Gamma_t=(c_t^{w},n_t,c_{t+1}^{r},\lambda_t)\in\mathbb{R}^{4}$ and
$\gamma_{t+1}\in\mathbb{R}$ denotes the exogenous longevity parameter. Suppose
the partial Jacobian $D_{\Gamma}F(\Gamma_t^{*},\gamma_{t+1}^{*})\in\mathbb{R}^{4\times4}$
is invertible. Then there exist open neighborhoods $U\ni\gamma_{t+1}^{*}$ and
$V\ni\Gamma_t^{*}$ and a unique map $\gamma_{t+1}\mapsto\Gamma_t(\gamma_{t+1})$
of class $C^{1}(U,V)$ such that
$F\!\big(\Gamma_t(\gamma_{t+1}),\gamma_{t+1}\big)=\mathbf{0}$ for all
$\gamma_{t+1}\in U$. In particular, the equilibrium fertility
$n_t(\gamma_{t+1})\coloneqq e_2^{\top}\Gamma_t(\gamma_{t+1})$, with
$e_2=(0,1,0,0)^{\top}$, is a well-defined $C^{1}$ function with
\[
\frac{dn_t}{d\gamma_{t+1}}
=-\,e_2^{\top}\,\big[D_{\Gamma}F(\Gamma_t^{*},\gamma_{t+1}^{*})\big]^{-1}\,
F_{\gamma}(\Gamma_t^{*},\gamma_{t+1}^{*}).
\]
If, moreover,
$e_2^{\top}\big[D_{\Gamma}F(\Gamma_t^{*},\gamma_{t+1}^{*})\big]^{-1}F_{\gamma}(\Gamma_t^{*},\gamma_{t+1}^{*})>0$,
then $dn_t/d\gamma_{t+1}<0$, i.e.\ higher longevity lowers equilibrium fertility.
\end{proposition}
\begin{proof}
By hypothesis, $F\in C^{1}$ and the partial Jacobian
$D_{\Gamma}F(\Gamma_t^{*},\gamma_{t+1}^{*})$ is invertible at the point
$(\Gamma_t^{*},\gamma_{t+1}^{*})$ satisfying
$F(\Gamma_t^{*},\gamma_{t+1}^{*})=0$. By the Implicit Function Theorem, there
exists a neighborhood $U\ni\gamma_{t+1}^{*}$ and a unique map
$\gamma_{t+1}\mapsto\Gamma_t(\gamma_{t+1})$ of class $C^{1}(U,\mathbb{R}^{4})$
such that
\[
F\big(\Gamma_t(\gamma_{t+1}),\gamma_{t+1}\big)=0,
\qquad \gamma_{t+1}\in U .
\]
Differentiating this identity with respect to $\gamma_{t+1}$ and evaluating at
$(\Gamma_t^{*},\gamma_{t+1}^{*})$ gives
\[
D_{\Gamma}F(\Gamma_t^{*})\,\Gamma_t'(\gamma_{t+1})
+F_{\gamma}(\Gamma_t^{*})=0,
\]
and, since $D_{\Gamma}F(\Gamma_t^{*})$ is invertible,
\[
\Gamma_t'(\gamma_{t+1})
=-\,D_{\Gamma}F(\Gamma_t^{*})^{-1}\,F_{\gamma}(\Gamma_t^{*}).
\]
Let $e_2=(0,1,0,0)^{\top}$, so that the equilibrium fertility component is
$n_t(\gamma_{t+1})=e_2^{\top}\Gamma_t(\gamma_{t+1})$. Projecting the previous
identity onto $e_2$ yields
\[
\frac{dn_t}{d\gamma_{t+1}}
=e_2^{\top}\Gamma_t'(\gamma_{t+1})
=-\,e_2^{\top}\,D_{\Gamma}F(\Gamma_t^{*})^{-1}\,F_{\gamma}(\Gamma_t^{*}).
\]
By the maintained sign condition
$e_2^{\top}D_{\Gamma}F(\Gamma_t^{*})^{-1}F_{\gamma}(\Gamma_t^{*})>0$, the
right-hand side is strictly negative; hence
\[
\frac{dn_t}{d\gamma_{t+1}}<0,
\]
that is, greater longevity lowers equilibrium fertility.
\end{proof}

\begin{remark}[Closed-form verification of the longevity sign conditions]
\label{rem:closedform}
Under the baseline logarithmic preferences the household block admits the
closed forms
\[
c_t^{w}=\frac{\Omega_t}{2+\beta\gamma_{t+1}},\qquad
s_t^{w}=\frac{\beta\gamma_{t+1}\,\Omega_t}{2+\beta\gamma_{t+1}},\qquad
n_t=\frac{\Omega_t}{\kappa\,w_t\,(2+\beta\gamma_{t+1})},
\]
where $\Omega_t=w_t+d_t$ is full income and the time-cost identity
$\kappa\,n_t\,w_t=c_t^{w}$ reconciles $s_t^{w}=\beta\gamma_{t+1}c_t^{w}$ with the
budget constraint $s_t^{w}=w_t(1-\kappa n_t)+d_t-c_t^{w}$. Differentiating at
fixed prices $(w_t,d_t)$ verifies the sign conditions invoked above directly,
without recourse to the Jacobian restriction:
\[
\frac{\partial s_t^{w}}{\partial\gamma_{t+1}}
=\frac{2\beta\,\Omega_t}{(2+\beta\gamma_{t+1})^{2}}>0,\qquad
\frac{\partial c_t^{w}}{\partial\gamma_{t+1}}
=-\frac{\beta\,\Omega_t}{(2+\beta\gamma_{t+1})^{2}}<0,\qquad
\frac{dn_t}{d\gamma_{t+1}}
=-\frac{\beta\,\Omega_t}{\kappa\,w_t\,(2+\beta\gamma_{t+1})^{2}}<0 .
\]
Greater longevity therefore raises retirement saving while lowering both current
consumption and equilibrium fertility, confirming
Proposition~\ref{proplonge} and the $\gamma_{t+1}$ components of
Proposition~\ref{prop1} and Corollary~\ref{prop1b} in closed form.
\end{remark}

\subsection{Proof of Proposition~\ref{prop3}}
\label{appendprop3}
\begin{proof}
Fix output at $\bar y>0$, treat $a$ as a parameter, and minimize $wh+R^{k}k$
subject to $F(k,h,a)=k^{\alpha}G^{1-\alpha}=\bar y$, where
$G=\big[h^{\rho}+(\phi a)^{\rho}\big]^{1/\rho}$ and $G_h=G^{1-\rho}h^{\rho-1}$.
Since $F$ is strictly quasi-concave in $(h,k)$ for $\alpha\in(0,1)$,
$\rho\in(0,1)$, the isoquant is strictly convex and the solution
$(h(a),k(a))\in\mathbb{R}_{++}^2$ is unique. The tangency condition
$w/R^{k}=F_h/F_k=\tfrac{1-\alpha}{\alpha}\,\tfrac{k}{G}G_h$ together with the
constraint yields, after eliminating $k$,
\begin{equation}
(1-\alpha+\alpha\rho)\,\ln G+\alpha(1-\rho)\,\ln h=\text{const}.
\label{eq:tangency}
\end{equation}
Let $s_h=h^{\rho}/G^{\rho}$ and $s_a=(\phi a)^{\rho}/G^{\rho}$ ($s_h+s_a=1$) be
the labor and AI shares in the composite, so that
$d\ln G=s_h\,d\ln h+s_a\,d\ln a$. Differentiating~\eqref{eq:tangency} gives
\[
\frac{d\ln h}{d\ln a}
=-\frac{(1-\alpha+\alpha\rho)\,s_a}{(1-\alpha+\alpha\rho)\,s_h+\alpha(1-\rho)}<0,
\qquad
\frac{d\ln G}{d\ln a}
=\frac{\alpha(1-\rho)\,s_a}{(1-\alpha+\alpha\rho)\,s_h+\alpha(1-\rho)}>0,
\]
where both signs use $\alpha\in(0,1)$, $\rho\in(0,1)$ and $s_a\in(0,1)$. Hence
$h'(a)<0$. Finally, the constraint gives $k=\bar y^{1/\alpha}G^{-(1-\alpha)/\alpha}$,
so $d\ln k/d\ln a=-\tfrac{1-\alpha}{\alpha}\,d\ln G/d\ln a<0$, i.e.,\ $k'(a)<0$.
\end{proof}

\subsection{Global comparative statics of AI}
\label{appendprop4}
\begin{proposition}[Global comparative statics of AI]
\label{prop4}
Let the competitive equilibrium be implicitly defined by $F(z_t,a_t)=0$, where
$z_t=(w_t,n_t,c_t^{w},s_t^{w},c_{t+1}^{r})\in\mathbb{R}_{++}^{5}$, $F\in C^{1}$,
and $D_zF(z_t^{*})$ is nonsingular. Assume $0<\rho<1-\alpha$, so that AI capital
is Edgeworth-complementary to both labor and physical capital
($\partial^2F/\partial a\,\partial h>0$ and $\partial^2F/\partial a\,\partial k>0$).
Then the equilibrium mapping $z_t=z_t(a_t)$ is $C^{1}$, with
\[
\frac{\partial z_t}{\partial a_t}
=-\big(D_zF(z_t^{*})\big)^{-1}F_{a_t}(z_t^{*}),
\]
and, provided full income $\Omega_t=w_t+d_t+b_t$ rises with AI ($\partial \Omega_t/\partial a_t>0$),
\[
\frac{\partial w_t}{\partial a_t}>0,\quad
\frac{\partial c_t^{w}}{\partial a_t}>0,\quad
\frac{\partial s_t^{w}}{\partial a_t}>0,\quad
\frac{\partial c_{t+1}^{r}}{\partial a_t}>0,
\]
while the fertility response is determined by the elasticity of full income
relative to the wage,
\[
\operatorname{sign}\frac{\partial n_t}{\partial a_t}
=\operatorname{sign}\!\left(\frac{d\ln \Omega_t}{d\ln a_t}-\frac{d\ln w_t}{d\ln a_t}\right),
\]
so that $\partial n_t/\partial a_t<0$ if and only if full income rises less than
proportionally to the wage. The sign pattern of $\partial z_t/\partial a_t$ is
therefore $(+,\,s_n,\,+,\,+,\,+)$ for $(w_t,n_t,c_t^{w},s_t^{w},c_{t+1}^{r})$,
with $s_n=\operatorname{sign}\big(d\ln \Omega_t/d\ln a_t-d\ln w_t/d\ln a_t\big)$.
\end{proposition}

\begin{proof}
\emph{Step 1: existence of the $C^{1}$ equilibrium mapping.}
Since $F\in C^{1}$ and the Jacobian $D_zF(z_t^{*})$ is nonsingular at the
equilibrium $(z_t^{*},a_t)$ solving $F(z_t^{*},a_t)=0$, the implicit function
theorem yields a unique $C^{1}$ map $z_t=z_t(a_t)$ on a neighborhood of $a_t$,
with
\[
\frac{\partial z_t}{\partial a_t}
=-\big(D_zF(z_t^{*})\big)^{-1}F_{a_t}(z_t^{*}).
\]

\emph{Step 2: factor-price responses (technological channel).}
Differentiating the marginal products of the production function
$F(k,h,a)=k^{\alpha}\big[h^{\rho}+(\phi a)^{\rho}\big]^{(1-\alpha)/\rho}$,
\[
\frac{\partial^2F}{\partial a\,\partial h}
=(1-\alpha)(1-\alpha-\rho)\phi^{\rho}k^{\alpha}h^{\rho-1}a^{\rho-1}
B^{\frac{1-\alpha-2\rho}{\rho}}>0,
\qquad
\frac{\partial^2F}{\partial a\,\partial k}>0,
\]
under $0<\rho<1-\alpha$, with $B=h^{\rho}+(\phi a)^{\rho}$. Hence a larger AI
stock raises the marginal products of labor and physical capital; by the
labor-market and no-arbitrage conditions (Proposition~\ref{prop5}), the
equilibrium factor prices satisfy
\[
\frac{\partial w_t}{\partial a_t}>0,
\qquad
\frac{\partial R_{t+1}}{\partial a_t}>0,
\]
the latter raising the gross return on saving and lowering the price of
retirement consumption $p_t=\gamma_{t+1}/(R_{t+1}\varphi_{t+1})$,
$\partial p_t/\partial a_t<0$.

\emph{Step 3: household block (closed form).}
Under logarithmic preferences the household solution is, with full income
$\Omega_t=w_t+d_t+b_t$,
\[
c_t^{w}=\frac{\Omega_t}{2+\beta\gamma_{t+1}},\quad
s_t^{w}=\frac{\beta\gamma_{t+1}\,\Omega_t}{2+\beta\gamma_{t+1}},\quad
n_t=\frac{\Omega_t}{(2+\beta\gamma_{t+1})\,\kappa\, w_t},\quad
c_{t+1}^{r}=\frac{\beta\gamma_{t+1}\,\Omega_t}{(2+\beta\gamma_{t+1})\,p_t}.
\]

\emph{Step 4: signs.}
Maintain $\partial\Omega_t/\partial a_t>0$ (full income rises with AI: the wage
gain of Step~2 and higher firm profits raise $\Omega_t$). Then, since $c_t^{w}$
and $s_t^{w}$ are proportional to $\Omega_t$,
\[
\frac{\partial c_t^{w}}{\partial a_t}
=\frac{1}{2+\beta\gamma_{t+1}}\frac{\partial\Omega_t}{\partial a_t}>0,
\qquad
\frac{\partial s_t^{w}}{\partial a_t}
=\frac{\beta\gamma_{t+1}}{2+\beta\gamma_{t+1}}\frac{\partial\Omega_t}{\partial a_t}>0.
\]
For retirement consumption, $c_{t+1}^{r}\propto\Omega_t/p_t$ with
$\partial\Omega_t/\partial a_t>0$ and $\partial p_t/\partial a_t<0$ (Step~2), so
\[
\frac{\partial c_{t+1}^{r}}{\partial a_t}>0.
\]
Finally, $n_t\propto\Omega_t/w_t$, hence
\[
\frac{\partial n_t}{\partial a_t}
=\frac{1}{(2+\beta\gamma_{t+1})\kappa\,}\,
\frac{\partial}{\partial a_t}\!\left(\frac{\Omega_t}{w_t}\right)
=\frac{n_t}{a_t}\!\left(\frac{\mathrm d\ln\Omega_t}{\mathrm d\ln a_t}
-\frac{\mathrm d\ln w_t}{\mathrm d\ln a_t}\right),
\]
so that
\[
\operatorname{sign}\frac{\partial n_t}{\partial a_t}
=\operatorname{sign}\!\left(\frac{\mathrm d\ln\Omega_t}{\mathrm d\ln a_t}
-\frac{\mathrm d\ln w_t}{\mathrm d\ln a_t}\right),
\]
which is negative if and only if full income rises less than proportionally to
the wage. Collecting the five signs gives the pattern
$(+,\,s_n,\,+,\,+,\,+)$ for $(w_t,n_t,c_t^{w},s_t^{w},c_{t+1}^{r})$, with
$s_n=\operatorname{sign}\big(\mathrm d\ln\Omega_t/\mathrm d\ln a_t-\mathrm d\ln w_t/\mathrm d\ln a_t\big)$.
\end{proof}

\subsection{Proof of Proposition~\ref{prop5}}
\label{appendprop5}
\begin{proof}
Throughout, maintain the complementarity assumption $0<\rho<1-\alpha$, fix the
physical capital stock $K_t>0$, and let $\bar w_t>0$ be the prevailing wage.
Aggregate labor demand equates the wage to the marginal product of labor,
\[
\mathcal F(H_t,a_t)\equiv
(1-\alpha)K_t^{\alpha}\,B^{\frac{1-\alpha-\rho}{\rho}}H_t^{\rho-1}-\bar w_t=0,
\qquad B\equiv H_t^{\rho}+(\phi a_t)^{\rho},
\]
with $\mathcal F\in C^{1}(\mathbb{R}_{++}^2)$. Writing $c\equiv(1-\alpha)K_t^{\alpha}>0$,
\[
\frac{\partial\mathcal F}{\partial H_t}
=c\,H_t^{\rho-2}B^{\frac{1-\alpha-2\rho}{\rho}}
\big[-\alpha H_t^{\rho}+(\rho-1)(\phi a_t)^{\rho}\big]<0
\quad\text{for all }\rho\in(0,1),
\]
both bracket terms being negative ($\alpha>0$, $\rho-1<0$): the marginal product
of labor is strictly decreasing in $H_t$ (strict concavity of output in labor).
Moreover,
\[
\frac{\partial\mathcal F}{\partial a_t}
=(1-\alpha)(1-\alpha-\rho)\phi^{\rho}K_t^{\alpha}H_t^{\rho-1}a_t^{\rho-1}
B^{\frac{1-\alpha-2\rho}{\rho}}>0
\]
under $\rho<1-\alpha$ (Edgeworth complementarity between labor and AI capital).

\emph{Conditional labor demand.} Because $\partial\mathcal F/\partial H_t<0$ on all
of $\mathbb{R}_{++}^2$, $H_t\mapsto\mathcal F(H_t,a_t)$ is strictly decreasing and,
together with $\lim_{H_t\to0^+}\mathcal F=+\infty$ and
$\lim_{H_t\to\infty}\mathcal F=-\bar w_t<0$, the equation $\mathcal F=0$ admits a
\emph{unique global} solution $H_t=H_t(a_t)$; the implicit function theorem
further makes it $C^{1}$, with
\[
\frac{\partial H_t}{\partial a_t}
=-\frac{\partial\mathcal F/\partial a_t}{\partial\mathcal F/\partial H_t}>0 ,
\]
so that, at a given wage, a larger AI stock shifts labor demand outward.

\emph{Equilibrium wage.} Let labor-market clearing be $H^{d}(w_t,a_t)=H^{s}(w_t)$,
with $\partial H^{d}/\partial w_t<0$, $\partial H^{d}/\partial a_t>0$ (the shift
just established), and $\mathrm{d}H^{s}/\mathrm{d}w_t\ge0$. Implicit
differentiation gives
\[
\frac{\partial w_t}{\partial a_t}
=\frac{\partial H^{d}/\partial a_t}
{\;\mathrm{d}H^{s}/\mathrm{d}w_t-\partial H^{d}/\partial w_t\;}>0,
\]
the denominator being strictly positive.

\emph{Rental of AI capital.} For
$R_t^{a}=\partial Y_t/\partial a_t
=(1-\alpha)\phi^{\rho}K_t^{\alpha}B^{\frac{1-\alpha-\rho}{\rho}}a_t^{\rho-1}$,
\[
\frac{\partial R_t^{a}}{\partial a_t}\Big|_{H_t}
=\frac{\partial^2 Y_t}{\partial a_t^{2}}
=(1-\alpha)\phi^{\rho}K_t^{\alpha}a_t^{\rho-2}B^{\frac{1-\alpha-2\rho}{\rho}}
\big[-\alpha\phi^{\rho}a_t^{\rho}+(\rho-1)H_t^{\rho}\big]<0
\]
(diminishing returns). Allowing labor to adjust, the total effect
\[
\frac{\mathrm{d}R_t^{a}}{\mathrm{d}a_t}
=\underbrace{\frac{\partial^2Y_t}{\partial a_t^{2}}}_{<0}
+\underbrace{\frac{\partial^2Y_t}{\partial a_t\partial H_t}}_{>0\ (\rho<1-\alpha)}
\,\underbrace{\frac{\mathrm{d}H_t}{\mathrm{d}a_t}}_{>0}
\]
is of ambiguous sign.

\emph{Consumption.} Output is strictly increasing in AI capital,
$\partial Y_t/\partial a_t=R_t^{a}>0$. By the resource constraint $C_t=Y_t-I_t$,
$\partial C_t/\partial a_t>0$ whenever the induced rise in investment does not
exceed that in output. 
\end{proof}

\subsection{Proof of Corollary~\ref{prop6}}
\label{appendprop6}
\begin{proof}
\emph{Limit technology.} As $\rho\to1^{-}$, the CES aggregator converges to a
linear one, $\big[H_t^{\rho}+(\phi a_t)^{\rho}\big]^{1/\rho}\to H_t+\phi a_t$, so
\[
Y_t=K_t^{\alpha}\big(H_t+\phi a_t\big)^{1-\alpha}.
\]
The marginal products are
$\mathrm{MPL}=\partial Y_t/\partial H_t=(1-\alpha)K_t^{\alpha}(H_t+\phi a_t)^{-\alpha}$
and $\mathrm{MPA}=\partial Y_t/\partial a_t=\phi\,\mathrm{MPL}$, so the marginal
rate of technical substitution $\mathrm{MPL}/\mathrm{MPA}=1/\phi$ is constant:
labor and AI capital are perfect substitutes.

\emph{Conditional labor demand.} At a given wage $\bar w_t>0$, the labor-demand
condition $\bar w_t=\mathrm{MPL}$ reads
$(H_t+\phi a_t)^{-\alpha}=\bar w_t/[(1-\alpha)K_t^{\alpha}]$, hence
$H_t+\phi a_t=\big[(1-\alpha)K_t^{\alpha}/\bar w_t\big]^{1/\alpha}$ is constant in
$a_t$. Differentiating gives $\mathrm{d}H_t+\phi\,\mathrm{d}a_t=0$, i.e.
\[
\frac{\partial H_t}{\partial a_t}=-\phi<0 :
\]
one efficiency unit of AI capital displaces $\phi$ units of labor.

\emph{Equilibrium wage.} Let labor-market clearing be $H^{d}(w_t,a_t)=H^{s}(w_t)$,
with $\partial H^{d}/\partial w_t<0$, $\partial H^{d}/\partial a_t=-\phi<0$, and
$\mathrm{d}H^{s}/\mathrm{d}w_t\ge0$. Implicit differentiation yields
\[
\frac{\partial w_t}{\partial a_t}
=\frac{\partial H^{d}/\partial a_t}
{\;\mathrm{d}H^{s}/\mathrm{d}w_t-\partial H^{d}/\partial w_t\;}\le 0,
\]
the denominator being strictly positive and the numerator nonpositive (with
equality only for perfectly elastic labor supply).

\emph{Rental of AI capital.} Since $\mathrm{MPA}=\phi\,\mathrm{MPL}$, competitive
pricing gives $R_t^{a}=\phi\,w_t$, whence
\[
\frac{\partial R_t^{a}}{\partial a_t}=\phi\,\frac{\partial w_t}{\partial a_t}\le 0 :
\]
being a perfect substitute for labor, AI capital depresses the marginal product
of the composite input and therefore its own rental, which falls together with
the wage.

\emph{Consumption.} Total output satisfies
$\mathrm{d}Y_t/\mathrm{d}a_t=\mathrm{MPL}\,(\,\mathrm{d}H_t/\mathrm{d}a_t+\phi\,)$,
with $\mathrm{d}H_t/\mathrm{d}a_t\in[-\phi,0]$ in equilibrium, so
$\mathrm{d}Y_t/\mathrm{d}a_t\ge0$. By the resource constraint $C_t=Y_t-I_t$, the
sign of $\partial C_t/\partial a_t$ is ambiguous, reflecting the opposing scale
(output) and substitution (labor-income) effects.

\emph{Remark on determinacy.} At $\rho=1$ exactly the nest is linear and interior
factor demands exist only on the price ray $R_t^{a}=\phi w_t$; the statics above
are therefore understood as the limit $\rho\to1^{-}$.
\end{proof}

\subsection{Proof of Proposition~\ref{prop_outer}}
\label{app:outer}
\begin{proof}
Write the labor--AI composite as $G_t=\big[H_t^{\rho}+(\phi_t a_t)^{\rho}\big]^{1/\rho}$
and $\Phi_t=\alpha K_t^{\eta}+(1-\alpha)G_t^{\eta}$, so that $Y_t=\Phi_t^{1/\eta}$.
The inner aggregator is homogeneous of degree one, with
$\partial G_t/\partial H_t=G_t^{1-\rho}H_t^{\rho-1}$ and
$\partial G_t/\partial a_t=\phi_t^{\rho}G_t^{1-\rho}a_t^{\rho-1}>0$; write
$\psi_t\equiv(\partial G_t/\partial a_t)/G_t>0$.

\emph{Marginal product of labor.} Differentiating $Y_t=\Phi_t^{1/\eta}$ and using
the chain rule,
\[
\mathrm{MPL}\equiv\frac{\partial Y_t}{\partial H_t}
=\Phi_t^{\frac{1-\eta}{\eta}}(1-\alpha)\,G_t^{\eta-1}\,\frac{\partial G_t}{\partial H_t}
=(1-\alpha)\,\Phi_t^{\frac{1-\eta}{\eta}}\,G_t^{\eta-\rho}\,H_t^{\rho-1}>0 .
\]

\emph{Cross partial.} Since $\mathrm{MPL}>0$, the sign of
$\partial^2 Y_t/\partial H_t\,\partial a_t=\partial\,\mathrm{MPL}/\partial a_t$
equals the sign of $\partial\ln\mathrm{MPL}/\partial a_t$. Taking logarithms,
\[
\ln\mathrm{MPL}=\ln(1-\alpha)+\tfrac{1-\eta}{\eta}\ln\Phi_t+(\eta-\rho)\ln G_t+(\rho-1)\ln H_t .
\]
With $\partial\Phi_t/\partial a_t=(1-\alpha)\,\eta\,G_t^{\eta}\,\psi_t$, so that
$(\partial\Phi_t/\partial a_t)/\Phi_t=\eta\,s_{G}\,\psi_t$ where
$s_{G}\equiv(1-\alpha)G_t^{\eta}/\Phi_t=1-s_{K}$,
\[
\frac{\partial\ln\mathrm{MPL}}{\partial a_t}
=\frac{1-\eta}{\eta}\,\frac{\partial\Phi_t/\partial a_t}{\Phi_t}
+(\eta-\rho)\,\psi_t
=\psi_t\big[(1-\eta)\,s_{G}+\eta-\rho\big] .
\]
Because $\psi_t>0$, and using $s_G=1-s_K$ together with
$(1-\eta)s_G+\eta-\rho=(1-\rho)-(1-\eta)s_K$,
\[
\operatorname{sign}\frac{\partial^2 Y_t}{\partial H_t\,\partial a_t}
=\operatorname{sign}\big[(1-\rho)-(1-\eta)\,s_{K}\big]
=\operatorname{sign}\!\left(\frac{1}{\sigma}-\frac{s_{K}}{\sigma_o}\right),
\]
the last equality substituting $1-\rho=1/\sigma$ and $1-\eta=1/\sigma_o$. Hence AI
capital is an Edgeworth complement to labor if and only if $1/\sigma>s_{K}/\sigma_o$,
i.e.\ $\sigma<\sigma_o/s_{K}$. In the Cobb--Douglas limit $\eta\to0$
($\sigma_o\to1$) the value share is constant at $s_{K}=\alpha$ and the condition
reduces to $\sigma<1/\alpha$, the threshold of Proposition~\ref{prop5} and
Corollary~\ref{prop6}.
\end{proof}

\subsection{Proof of Lemma~\ref{lem:labor}}
\label{app:labor}
\begin{proof}
Write $B_t=H_t^{\rho}+(\phi_t a_t)^{\rho}$, so that
$Y_t=K_t^{\alpha}B_t^{(1-\alpha)/\rho}$ and
$s_a=(\phi_t a_t)^{\rho}/B_t$, with $\partial\ln B_t/\partial\ln a_t=\rho\,s_a$.
Up to the common markup $\theta_t$, which cancels in every elasticity, the
competitive factor prices are
\begin{align*}
w_t   &=(1-\alpha)\,K_t^{\alpha}B_t^{\frac{1-\alpha-\rho}{\rho}}H_t^{\rho-1}, \\
R^{a}_t &=(1-\alpha)\,\phi_t^{\rho}K_t^{\alpha}B_t^{\frac{1-\alpha-\rho}{\rho}}a_t^{\rho-1}.
\end{align*}
Fix $K_t$ and $H_t$. Since $a_t$ enters $Y_t$, $w_t$, and $R^{a}_t$ only through
$B_t$, and additionally through the explicit factor $a_t^{\rho-1}$ in $R^{a}_t$,
logarithmic differentiation with $\partial\ln B_t/\partial\ln a_t=\rho\,s_a$ gives
\begin{align}
\frac{\partial\ln Y_t}{\partial\ln a_t}
   &=\frac{1-\alpha}{\rho}\,\frac{\partial\ln B_t}{\partial\ln a_t}=(1-\alpha)\,s_a,
   \label{eq:pf_Y}\\
\frac{\partial\ln w_t}{\partial\ln a_t}
   &=\frac{1-\alpha-\rho}{\rho}\,\frac{\partial\ln B_t}{\partial\ln a_t}
   =(1-\alpha-\rho)\,s_a,
   \label{eq:pf_w}\\
\frac{\partial\ln R^{a}_t}{\partial\ln a_t}
   &=(1-\alpha-\rho)\,s_a+(\rho-1)=(1-\alpha-\rho)\,s_a-(1-\rho).
   \label{eq:pf_Ra}
\end{align}
For the labor share, Euler's theorem together with the wage expression yields
$s_L=w_tH_t/(\theta_t Y_t)=(1-\alpha)H_t^{\rho}/B_t=(1-\alpha)(1-s_a)$. The AI
share $s_a=(\phi_t a_t)^{\rho}/B_t$ satisfies
$\partial\ln s_a/\partial\ln a_t=\rho-\rho\,s_a=\rho(1-s_a)$, so that
$\partial s_a/\partial\ln a_t=\rho\,s_a(1-s_a)$; hence
\begin{equation}
\frac{\partial\ln s_L}{\partial\ln a_t}
=\frac{\partial\ln(1-s_a)}{\partial\ln a_t}
=-\frac{1}{1-s_a}\,\frac{\partial s_a}{\partial\ln a_t}
=-\rho\,s_a.
\label{eq:pf_sL}
\end{equation}
Equations~\eqref{eq:pf_Y}--\eqref{eq:pf_sL} are the four claimed elasticities. The
sign statements of Propositions~\ref{prop_share}
and~\ref{prop_two} and Corollary~\ref{cor_decouple} follow because $s_a\in(0,1)$
and $\operatorname{sign}\rho=\operatorname{sign}(\sigma-1)$.

\emph{General outer nest (Remark~\ref{rem:hicks}).} For the CES nest
$Y_t=[\alpha K_t^{\eta}+(1-\alpha)G_t^{\eta}]^{1/\eta}$ of
Appendix~\ref{app:outer}, the aggregate labor share is
$s_{L}=w_tH_t/(\theta_t Y_t)=s_{G}\,(1-s_{a})=(1-s_{K})(1-s_{a})$, where
$s_{G}=(1-\alpha)G_t^{\eta}/\Phi_t=1-s_{K}$. From
$\partial\ln\Phi_t/\partial\ln a_t=\eta\,s_{G}s_{a}$ one has
$\partial\ln s_{K}/\partial\ln a_t=-\eta\,s_{G}s_{a}$, hence
$\partial\ln(1-s_{K})/\partial\ln a_t=\eta\,s_{K}s_{a}$; combined with
$\partial\ln(1-s_{a})/\partial\ln a_t=-\rho\,s_{a}$ this yields
\[
\frac{\partial\ln s_{L}}{\partial\ln a_t}
=\eta\,s_{K}s_{a}-\rho\,s_{a}=-\,s_{a}\,(\rho-\eta\,s_{K}),
\]
which reduces to $-\rho\,s_a$ at $\eta=0$ and confirms Remark~\ref{rem:hicks}.
\end{proof}

\subsection{Parenting time, fertility, and effective labor}
\label{appendprop7}
\begin{lemma}[Parenting time, fertility, and effective labor]
\label{prop7}
Let preferences be $U(c_t^{w},n_t)+\beta\gamma_{t+1}V(c_{t+1}^{r})$ with $U,V\in
C^{1}$, $U$ strictly increasing and quasi-concave, $V'>0$, and let
$\Omega_t\equiv w_t+d_t$ denote full income. The household's optimal fertility
$n_t$ is a $C^{1}$ function of the price of a child $\kappa\, w_t$ and of
$(\Omega_t,\gamma_{t+1},R_{t+1})$, and effective labor supply is
$h_t=1-\kappa\, n_t$. Define the own-price elasticity of fertility demand
\[
\varepsilon_t\equiv\frac{\partial\ln n_t}{\partial\ln\kappa\,}
=\frac{\kappa\,}{n_t}\frac{\partial n_t}{\partial\kappa\,}.
\]
Then parenting time and effective labor respond to the time cost according to
\[
\frac{\partial(\kappa\, n_t)}{\partial\kappa\,}=n_t\,(1+\varepsilon_t),
\qquad
\frac{\partial h_t}{\partial\kappa\,}=-\,n_t\,(1+\varepsilon_t),
\]
and, at given output and factor prices, the time cost of children affects factor
demands only through $h_t$,
\[
\frac{\partial k_t}{\partial\kappa\,}
=\frac{\partial k_t}{\partial h_t}\,\frac{\partial h_t}{\partial\kappa\,},
\qquad
\frac{\partial a_t}{\partial\kappa\,}
=\frac{\partial a_t}{\partial h_t}\,\frac{\partial h_t}{\partial\kappa\,}.
\]
In particular, the labor-scarcity effect vanishes
($\partial h_t/\partial\kappa\,=\partial k_t/\partial\kappa\,=\partial a_t/\partial\kappa\,=0$)
if and only if fertility demand is unit-elastic, $\varepsilon_t=-1$; it is
negative when demand is inelastic ($\varepsilon_t>-1$) and positive when demand is
elastic ($\varepsilon_t<-1$). The logarithmic specification is the knife-edge case
$\varepsilon_t=-1$.
\end{lemma}

\begin{proof}
\textbf{Standing assumptions.}
Let $U\in C^{2}(\mathbb{R}_{++}^{2})$ and $V\in C^{2}(\mathbb{R}_{++})$ satisfy
$U_c,U_n>0$, $V'>0$, the Hessian $D^{2}U$ negative definite, and $V''<0$, so that
the objective
$W(c^{w},n,c^{r})\equiv U(c^{w},n)+\beta\gamma_{t+1}V(c^{r})$ is strictly concave
on $\mathbb{R}_{++}^{3}$. Assume the Inada conditions
$\lim_{x\to0^{+}}U_c=\lim_{x\to0^{+}}U_n=\lim_{x\to0^{+}}V'=+\infty$, ruling out
boundary optima. Fix $(\Omega_t,w_t,p_t)\in\mathbb{R}_{++}^{3}$ with
$\Omega_t=w_t+d_t$ and $p_t=\gamma_{t+1}/R_{t+1}$.

\textbf{Step 1 (existence, uniqueness, interiority).}
For each $\kappa\,>0$ the budget set
$B(\kappa\,)=\{(c^{w},n,c^{r})\in\mathbb{R}_{+}^{3}:c^{w}+p_tc^{r}+\kappa\, w_t n\le\Omega_t\}$
is compact and convex with nonempty interior. As $W$ is continuous and strictly
concave, it attains a unique maximizer on $B(\kappa\,)$; by the Inada conditions the
maximizer lies in $\operatorname{int}B(\kappa\,)$ and the budget binds. Hence the
optimum is the unique interior point satisfying the first-order conditions, for
some multiplier $\lambda_t>0$:
\begin{equation}
U_c(c_t^{w},n_t)=\lambda_t,\qquad
U_n(c_t^{w},n_t)=\lambda_t\,\kappa\, w_t,\qquad
\beta\gamma_{t+1}V'(c_{t+1}^{r})=\lambda_t\,p_t,
\label{eq:foc}
\end{equation}
together with $c_t^{w}+p_tc_{t+1}^{r}+\kappa\, w_t n_t=\Omega_t$.

\textbf{Step 2 ($C^{1}$ dependence on $\kappa\,$).}
Define $\Psi:\mathbb{R}_{++}^{3}\times\mathbb{R}_{++}\times\mathbb{R}_{++}\to\mathbb{R}^{4}$,
with state $x=(c^{w},n,c^{r},\lambda)$ and parameter $\kappa\,$, by
\[
\Psi(x,\kappa\,)=
\begin{pmatrix}
U_c(c^{w},n)-\lambda\\
U_n(c^{w},n)-\lambda\kappa\, w_t\\
\beta\gamma_{t+1}V'(c^{r})-\lambda p_t\\
\Omega_t-c^{w}-p_tc^{r}-\kappa\, w_t n
\end{pmatrix}.
\]
At the optimum $\Psi(x_t,\kappa\,)=0$. Its Jacobian in $x$ is the bordered Hessian
\[
D_x\Psi(x_t,\kappa\,)=
\begin{pmatrix}
U_{cc} & U_{cn} & 0 & -1\\
U_{nc} & U_{nn} & 0 & -\kappa\, w_t\\
0 & 0 & \beta\gamma_{t+1}V'' & -p_t\\
-1 & -\kappa\, w_t & -p_t & 0
\end{pmatrix}.
\]
Strict concavity of $W$ (so that $D^{2}W$ is negative definite) together with the
nonzero constraint gradient $g\equiv(1,\kappa\, w_t,p_t)^{\top}\neq 0$ implies, by the
standard second-order theory of constrained optimization, that $D^{2}W$ is negative
definite on $\ker g^{\top}$; equivalently the bordered Hessian $D_x\Psi(x_t,\kappa\,)$
is nonsingular. Since $\Psi\in C^{1}$, the implicit function theorem yields a
neighborhood of $\kappa\,$ on which the solution
$x_t(\kappa\,)=(c_t^{w},n_t,c_{t+1}^{r},\lambda_t)(\kappa\,)$ is the unique zero of
$\Psi(\cdot,\kappa\,)$ and is continuously differentiable, with
\begin{equation}
\frac{\mathrm{d}x_t}{\mathrm{d}\kappa\,}
=-\big[D_x\Psi(x_t,\kappa\,)\big]^{-1}\,\partial_\kappa\,\Psi(x_t,\kappa\,),
\qquad
\partial_\kappa\,\Psi=\big(0,\,-\lambda_t w_t,\,0,\,-w_t n_t\big)^{\top}.
\label{eq:ift}
\end{equation}
In particular $\partial n_t/\partial\kappa\,$ exists and is finite, and the
own-price elasticity of fertility demand
$\varepsilon_t\equiv(\kappa\,/n_t)\,\partial n_t/\partial\kappa\,$ is well defined.

\textbf{Step 3 (parenting time and effective labor).}
Let $T_t\equiv\kappa\, n_t$ denote parenting time per worker. By Step~2,
$\kappa\,\mapsto T_t$ is $C^{1}$ and
\[
\frac{\partial T_t}{\partial\kappa\,}
=n_t+\kappa\,\,\frac{\partial n_t}{\partial\kappa\,}
=n_t\Big(1+\frac{\kappa\,}{n_t}\frac{\partial n_t}{\partial\kappa\,}\Big)
=n_t\,(1+\varepsilon_t).
\]
Effective labor supply per worker is $h_t=1-T_t$, whence
$\partial h_t/\partial\kappa\,=-n_t(1+\varepsilon_t)$, and aggregate effective labor
$L_t=h_tN_t^{w}$ ($N_t^{w}$ predetermined) satisfies
$\partial L_t/\partial\kappa\,=-N_t^{w}n_t(1+\varepsilon_t)$.

\textbf{Step 4 (transmission to factor demands).}
At given output $Y_t$ and factor prices $(w_t,R_t^{k},R_t^{a})$, the firm's
conditional factor demands solve
$\min_{k,h,a}\,\{w_th+R_t^{k}k+R_t^{a}a:\,F(k,h,a)=Y_t\}$, where
$F(k,h,a)=k^{\alpha}[h^{\rho}+(\phi a)^{\rho}]^{(1-\alpha)/\rho}$ is, for
$\rho<1$, $C^{2}$ and concave with $DF\neq 0$. The cost function is therefore
$C^{2}$ and, by Shephard's lemma, the conditional demands
$k_t=\hat k(h_t;\cdot)$ and $a_t=\hat a(h_t;\cdot)$ are $C^{1}$ in the labor input
$h_t$. Because $\kappa\,$ enters the firm's program only through $h_t$, the maps
$\kappa\,\mapsto k_t$ and $\kappa\,\mapsto a_t$ are compositions of $C^{1}$ functions,
and the chain rule gives
\[
\frac{\partial k_t}{\partial\kappa\,}
=\frac{\partial\hat k}{\partial h_t}\,\frac{\partial h_t}{\partial\kappa\,}
=-\frac{\partial\hat k}{\partial h_t}\,n_t(1+\varepsilon_t),
\qquad
\frac{\partial a_t}{\partial\kappa\,}
=\frac{\partial\hat a}{\partial h_t}\,\frac{\partial h_t}{\partial\kappa\,}
=-\frac{\partial\hat a}{\partial h_t}\,n_t(1+\varepsilon_t).
\]

\textbf{Step 5 (the unit-elastic knife edge).}
Since $n_t>0$, the common factor $\partial h_t/\partial\kappa\,=-n_t(1+\varepsilon_t)$
vanishes if and only if $\varepsilon_t=-1$, in which case
$\partial k_t/\partial\kappa\,=\partial a_t/\partial\kappa\,=0$ and the time cost of
children has no first-order labor-scarcity effect on factor demands. Otherwise the
sign of each response is that of $-(1+\varepsilon_t)\,\partial\hat x/\partial h_t$
($x\in\{k,a\}$): for inelastic fertility demand ($\varepsilon_t>-1$) the labor
input falls with $\kappa\,$, while for elastic demand ($\varepsilon_t<-1$) it rises.
The logarithmic specification delivers $\varepsilon_t\equiv-1$, the knife-edge case
in which $T_t$, $h_t$, and the factor demands are all invariant to $\kappa\,$. This
establishes the lemma. \qedhere
\end{proof}

\subsection{AI and economic growth: a two-capital decomposition}
\label{appendprop8}
\begin{proposition}[AI and economic growth: a two-capital decomposition]
\label{prop8}
Let $z_t=k_t+a_t$ denote the aggregate capital stock, with the two stocks evolving
as $k_{t+1}=(1-\delta_k)k_t+i_t^{k}$ and $a_{t+1}=(1-\delta_a)a_t+i_t^{a}$,
$\delta_k,\delta_a\in(0,1)$, and $i_t^{k}+i_t^{a}=Y_t-C_t$ (goods-market clearing).
Define the growth rate $g_t\equiv z_{t+1}/z_t-1$. Under competitive factor markets
and an interior equilibrium, the effect of AI capital on growth admits the exact
decomposition
\[
\frac{\partial g_t}{\partial a_t}
=\underbrace{\frac{1}{z_t}\frac{\partial Y_t}{\partial a_t}}_{\text{(i) productivity}}
-\underbrace{\frac{1}{z_t}\frac{\partial C_t}{\partial a_t}}_{\text{(ii) consumption crowding-out}}
-\underbrace{\frac{\delta_k+g_t}{z_t}\frac{\partial k_t}{\partial a_t}}_{\text{(iii) physical-capital dilution}}
-\underbrace{\frac{\delta_a+g_t}{z_t}}_{\text{(iv) AI-capital dilution}}.
\]
The marginal product of AI capital is strictly positive and, under monopolistic
competition, equals its rental scaled by the inverse markup,
$\partial Y_t/\partial a_t=R_t^{a}/\theta_t>0$, where $\theta_t=(\xi-1)/\xi$,
for every $\rho\in(0,1)$. By Proposition~\ref{prop5} (Edgeworth complementarity
$\rho<1-\alpha$ together with an interior investment share), the equilibrium
responses satisfy $\partial C_t/\partial a_t>0$ and $\partial k_t/\partial a_t>0$.
Terms (i)--(iv) thus have signs $(+),(-),(-),(-)$, and $\partial g_t/\partial a_t$
is of indeterminate sign in general. It is strictly positive if and only if the
productivity gain exceeds the combined crowding-out and dilution effects,
\[
\frac{\partial Y_t}{\partial a_t}-\frac{\partial C_t}{\partial a_t}
>(\delta_k+g_t)\,\frac{\partial k_t}{\partial a_t}+(\delta_a+g_t),
\]
or, equivalently,
$\dfrac{\mathrm{d}\ln z_{t+1}}{\mathrm{d}a_t}>\dfrac{\mathrm{d}\ln z_t}{\mathrm{d}a_t}$.
\end{proposition}

\begin{proof}
Throughout, the equilibrium objects $K_t,H_t,k_t,a_t,C_t,Y_t$ are treated as
$C^{1}$ functions of $a_t$ on an open neighborhood of the interior equilibrium,
with $z_t=k_t+a_t>0$, so that all partial derivatives below exist and the quotient
rule applies.

\smallskip
\textbf{Step 1 (marginal product and rental).}
With $\Phi_t\coloneqq H_t^{\rho}+(\phi a_t)^{\rho}>0$ and
$Y_t=K_t^{\alpha}\Phi_t^{(1-\alpha)/\rho}$,
\[
\frac{\partial Y_t}{\partial a_t}
=(1-\alpha)\,\phi^{\rho}K_t^{\alpha}\,\Phi_t^{\frac{1-\alpha-\rho}{\rho}}a_t^{\rho-1}>0
\qquad\text{for all }\rho\in(0,1).
\]
Under monopolistic competition the firm equates the rental to the marginal
revenue product, $R_t^{a}=\theta_t\,\partial Y_t/\partial a_t$ with
$\theta_t=(\xi-1)/\xi\in(0,1)$, so $\partial Y_t/\partial a_t=R_t^{a}/\theta_t>0$.

\smallskip
\textbf{Step 2 (variation of next-period capital).}
Summing the two laws of motion,
$z_{t+1}=(1-\delta_k)k_t+(1-\delta_a)a_t+(Y_t-C_t)$, and differentiating with
respect to $a_t$ gives, by linearity and using $\partial a_t/\partial a_t=1$,
\begin{equation}
\frac{\partial z_{t+1}}{\partial a_t}
=\frac{\partial Y_t}{\partial a_t}-\frac{\partial C_t}{\partial a_t}
+(1-\delta_k)\frac{\partial k_t}{\partial a_t}
+(1-\delta_a).
\label{eq:dz}
\end{equation}
By Proposition~\ref{prop5}, $\partial C_t/\partial a_t>0$ and
$\partial k_t/\partial a_t>0$.

\smallskip
\textbf{Step 3 (growth decomposition).}
By the quotient rule applied to $g_t=z_{t+1}/z_t-1$,
\begin{equation}
\frac{\partial g_t}{\partial a_t}
=\frac{1}{z_t}\frac{\partial z_{t+1}}{\partial a_t}
-\frac{z_{t+1}}{z_t^{2}}\frac{\partial z_t}{\partial a_t},
\qquad
\frac{\partial z_t}{\partial a_t}=\frac{\partial k_t}{\partial a_t}+1.
\label{eq:dg}
\end{equation}
Substituting \eqref{eq:dz} into \eqref{eq:dg} and grouping, with
$z_{t+1}=(1+g_t)z_t$,
\[
\frac{1-\delta_x}{z_t}-\frac{z_{t+1}}{z_t^{2}}
=\frac{(1-\delta_x)z_t-z_{t+1}}{z_t^{2}}
=-\frac{\delta_x+g_t}{z_t},
\qquad x\in\{k,a\},
\]
yields the decomposition
\begin{equation}
\frac{\partial g_t}{\partial a_t}
=\frac{1}{z_t}\frac{\partial Y_t}{\partial a_t}
-\frac{1}{z_t}\frac{\partial C_t}{\partial a_t}
-\frac{\delta_k+g_t}{z_t}\frac{\partial k_t}{\partial a_t}
-\frac{\delta_a+g_t}{z_t}.
\label{eq:decomp}
\end{equation}
Under the sign restrictions of Step~2 and $\delta_x+g_t>0$, the four terms of
\eqref{eq:decomp} have signs $(+),(-),(-),(-)$, so $\partial g_t/\partial a_t$ is
not sign-definite. Multiplying \eqref{eq:decomp} by $z_t>0$ gives
\[
\frac{\partial g_t}{\partial a_t}>0
\;\Longleftrightarrow\;
\frac{\partial Y_t}{\partial a_t}-\frac{\partial C_t}{\partial a_t}
>(\delta_k+g_t)\frac{\partial k_t}{\partial a_t}+(\delta_a+g_t).
\]

\smallskip
\textbf{Step 4 (logarithmic form).}
Dividing \eqref{eq:dg} by $z_{t+1}/z_t=1+g_t>0$ preserves the sign and gives
\[
\frac{1}{1+g_t}\frac{\partial g_t}{\partial a_t}
=\frac{1}{z_{t+1}}\frac{\partial z_{t+1}}{\partial a_t}
-\frac{1}{z_t}\frac{\partial z_t}{\partial a_t}
=\frac{\mathrm{d}\ln z_{t+1}}{\mathrm{d}a_t}-\frac{\mathrm{d}\ln z_t}{\mathrm{d}a_t},
\]
so $\partial g_t/\partial a_t>0\iff
\mathrm{d}\ln z_{t+1}/\mathrm{d}a_t>\mathrm{d}\ln z_t/\mathrm{d}a_t$, which
completes the proof.
\end{proof}

\subsection{Proof of Proposition~\ref{prop_robpref}}
\label{app:robpref}
\begin{proof}
At an interior equilibrium the capital market clears,
\[
s_t^{w}(R_t;\Theta)=q_t^{k}k_{t+1}(R_t;\phi_t)+q_t^{a}a_{t+1}(R_t;\phi_t)
\equiv\mathcal{D}_t(R_t;\phi_t),
\]
where $\mathcal{D}_t$ is aggregate capital demand and $\Theta$ collects the
preference parameters. By the firms' first-order conditions and the strict
concavity of the production nest, capital demand is strictly decreasing in the
return, $\partial\mathcal{D}_t/\partial R_t<0$, while desired saving has a
nonnegative slope, $\partial s_t^{w}/\partial R_t\ge0$ (for logarithmic preferences
$s_t^{w}$ is independent of $R_t$). Implicit differentiation of the clearing
condition gives
\[
\frac{\partial R_t}{\partial\phi_t}
=\frac{\partial\mathcal{D}_t/\partial\phi_t}
{\partial s_t^{w}/\partial R_t-\partial\mathcal{D}_t/\partial R_t},
\qquad
\frac{\partial R_t}{\partial\gamma_{t+1}}
=\frac{-\,\partial s_t^{w}/\partial\gamma_{t+1}}
{\partial s_t^{w}/\partial R_t-\partial\mathcal{D}_t/\partial R_t},
\]
with strictly positive denominator. A positive AI shock raises the marginal
products of both stocks (Proposition~\ref{prop5} and the Remark following
Proposition~\ref{prop3}), so $\partial\mathcal{D}_t/\partial\phi_t>0$ and
$\partial R_t/\partial\phi_t>0$. A positive longevity shock raises desired saving,
$\partial s_t^{w}/\partial\gamma_{t+1}>0$ (Proposition~\ref{prop1}), so the second
numerator is negative and $\partial R_t/\partial\gamma_{t+1}<0$. The parameters
$(\theta,\eta,\chi)$ enter only $\partial s_t^{w}/\partial R_t$ and the level of
$s_t^{w}$, hence the magnitude of the responses but not their sign.
\end{proof}

\subsection{Derivation of the consumption-equivalent measure
(Definition~\ref{def:welfare})}
\label{app:cev}
\begin{proof}
Consider the cohort that is young at date $t$, whose equilibrium welfare is
\[
W_t=\log c_t^{w}+\log n_t+\beta\,\gamma_{t+1}\,\log c_{t+1}^{r}.
\]
The consumption-equivalent variation $\omega_t$ is defined as the permanent
proportional change in lifetime consumption that, applied uniformly to the two
consumption goods entering felicity, $c_t^{w}$ and $c_{t+1}^{r}$, reproduces a
given welfare change $\mathrm{d}W_t$. Fertility $n_t$ is a time-allocation choice
rather than a consumption good, and is therefore held fixed in the experiment.
Replacing $c_t^{w}$ by $(1+\omega_t)\,c_t^{w}$ and $c_{t+1}^{r}$ by
$(1+\omega_t)\,c_{t+1}^{r}$ gives the scaled welfare level
\[
W_t(\omega_t)
=\log\!\big[(1+\omega_t)c_t^{w}\big]+\log n_t
+\beta\gamma_{t+1}\,\log\!\big[(1+\omega_t)c_{t+1}^{r}\big].
\]
Applying $\log(ab)=\log a+\log b$ to the first and third terms and collecting the
$\log(1+\omega_t)$ contributions,
\begin{align*}
W_t(\omega_t)
&=\underbrace{\big[\log c_t^{w}+\log n_t+\beta\gamma_{t+1}\log c_{t+1}^{r}\big]}_{=\,W_t}
+\log(1+\omega_t)+\beta\gamma_{t+1}\log(1+\omega_t)\\
&=W_t+\big(1+\beta\gamma_{t+1}\big)\log(1+\omega_t),
\end{align*}
since the scalar $\log(1+\omega_t)$ factors out of both the worker term and the
survival-weighted retiree term, while the fertility term $\log n_t$ is unaffected.
By definition $\omega_t$ equates the welfare gain from this scaling to the
perturbation $\mathrm{d}W_t$:
\[
W_t(\omega_t)-W_t=\big(1+\beta\gamma_{t+1}\big)\log(1+\omega_t)=\mathrm{d}W_t .
\]
Solving for $\omega_t$,
\[
\log(1+\omega_t)=\frac{\mathrm{d}W_t}{1+\beta\gamma_{t+1}}
\;\Longrightarrow\;
1+\omega_t=\exp\!\left(\frac{\mathrm{d}W_t}{1+\beta\gamma_{t+1}}\right)
\;\Longrightarrow\;
\omega_t=\exp\!\left(\frac{\mathrm{d}W_t}{1+\beta\gamma_{t+1}}\right)-1 ,
\]
which is the expression stated in Definition~\ref{def:welfare}. Two observations
clarify its content. First, the coefficient $1+\beta\gamma_{t+1}$ is exactly the
sum of the weights attached to the log-consumption terms in $W_t$, namely unity on
the worker term and $\beta\gamma_{t+1}$ on the retiree term; a permanent uniform
rise in both consumptions therefore raises welfare in proportion to this total
weight. Second, a first-order expansion around $\omega_t=0$ gives
$\omega_t\approx \mathrm{d}W_t/(1+\beta\gamma_{t+1})$, so $\omega_t$ inherits the
sign of $\mathrm{d}W_t$ and the denominator merely converts the utils into
consumption-equivalent units.
\end{proof}

\subsection{Envelope representation of marginal welfare}
\label{app:envelope}
\begin{lemma}[Envelope representation of marginal welfare]
\label{lem:envelope}
Let $\lambda_t>0$ be the multiplier on the household's consolidated budget
constraint $c_t^{w}+(\gamma_{t+1}/R_{t+1})\,c_{t+1}^{r}=w_t(1-\kappa n_t)+d_t$,
and write $h_t=1-\kappa n_t$. Evaluating $W_t$ at the optimal policy, the
partial effects of the household's environment on welfare are
\begin{align}
\frac{\partial W_t}{\partial w_t}        &=\lambda_t h_t,
   \label{eq:env_w}\\[2pt]
\frac{\partial W_t}{\partial d_t}        &=\lambda_t,
   \label{eq:env_d}\\[2pt]
\frac{\partial W_t}{\partial R_{t+1}}    &=\lambda_t\,\frac{\gamma_{t+1}\,c_{t+1}^{r}}{R_{t+1}^{2}}>0,
   \label{eq:env_R}\\[2pt]
\frac{\partial W_t}{\partial \gamma_{t+1}}
   &=\underbrace{\beta\,\log c_{t+1}^{r}}_{\text{value of longer life}}
    -\underbrace{\lambda_t\,\frac{c_{t+1}^{r}}{R_{t+1}}}_{\text{cost of financing retirement}} .
   \label{eq:env_gamma}
\end{align}
\end{lemma}

\begin{proof}
The household chooses $(c_t^{w},n_t,c_{t+1}^{r})\gg0$ to maximize
$U=\log c_t^{w}+\log n_t+\beta\gamma_{t+1}\log c_{t+1}^{r}$ subject to
$M_t-c_t^{w}-(\gamma_{t+1}/R_{t+1})c_{t+1}^{r}\ge0$, with
$M_t=w_t(1-\kappa n_t)+d_t$. With multiplier $\lambda_t$ and Lagrangian
$\mathcal{L}=U+\lambda_t[\,M_t-c_t^{w}-(\gamma_{t+1}/R_{t+1})c_{t+1}^{r}\,]$, the
objective is strictly concave and the constraint linear, so the interior
first-order conditions
\begin{align*}
\frac{1}{c_t^{w}} &=\lambda_t, \\
\frac{1}{n_t} &=\lambda_t\kappa w_t, \\
\frac{\beta\gamma_{t+1}}{c_{t+1}^{r}} &=\lambda_t\frac{\gamma_{t+1}}{R_{t+1}},
\end{align*}
characterize the unique optimum, with $\lambda_t=1/c_t^{w}>0$. Let
$W_t=\mathcal{L}^{\star}$ denote the optimized value. The envelope theorem yields,
for each $\theta\in\{w_t,d_t,R_{t+1},\gamma_{t+1}\}$,
$\partial W_t/\partial\theta=\partial\mathcal{L}/\partial\theta$ at the optimum:
\begin{align*}
\frac{\partial W_t}{\partial w_t}
   &=\lambda_t\frac{\partial M_t}{\partial w_t}=\lambda_t(1-\kappa n_t)=\lambda_t h_t, \\
\frac{\partial W_t}{\partial d_t} &=\lambda_t, \\
\frac{\partial W_t}{\partial R_{t+1}}
   &=-\lambda_t c_{t+1}^{r}\,\frac{\partial}{\partial R_{t+1}}\!\Big(\frac{\gamma_{t+1}}{R_{t+1}}\Big)
   =\lambda_t\frac{\gamma_{t+1}c_{t+1}^{r}}{R_{t+1}^{2}}>0, \\
\frac{\partial W_t}{\partial\gamma_{t+1}}
   &=\beta\log c_{t+1}^{r}-\lambda_t\frac{c_{t+1}^{r}}{R_{t+1}} .
\end{align*}
Positivity of $\partial W_t/\partial R_{t+1}$ follows from
$\lambda_t,\gamma_{t+1},c_{t+1}^{r}>0$.
\end{proof}

\subsection{Proof of Proposition~\ref{prop_welfdecomp}}
\label{app:welfdecomp}
\begin{proof}
Write the optimized value as
$W_t=\mathcal{V}(w_t,d_t,R_{t+1},\gamma_{t+1})$ (Lemma~\ref{lem:envelope}). For a
disturbance $\xi$ the chain rule gives
\[
\frac{\mathrm{d}W_t}{\mathrm{d}\xi}
=\frac{\partial W_t}{\partial w_t}\frac{\mathrm{d}w_t}{\mathrm{d}\xi}
+\frac{\partial W_t}{\partial d_t}\frac{\mathrm{d}d_t}{\mathrm{d}\xi}
+\frac{\partial W_t}{\partial R_{t+1}}\frac{\mathrm{d}R_{t+1}}{\mathrm{d}\xi}
+\frac{\partial W_t}{\partial\gamma_{t+1}}\frac{\mathrm{d}\gamma_{t+1}}{\mathrm{d}\xi}.
\]
The optimal policies $(c_t^{w},n_t,c_{t+1}^{r})$ also vary with $\xi$, but their
contribution is
$\nabla_{(c^{w},n,c^{r})}\mathcal{L}\cdot(\mathrm{d}c_t^{w},\mathrm{d}n_t,\mathrm{d}c_{t+1}^{r})/\mathrm{d}\xi=0$
because the first-order conditions hold. Substituting the partials of
Lemma~\ref{lem:envelope} delivers the stated decomposition.
\end{proof}

\subsection{Proof of Proposition~\ref{prop_welfsign}}
\label{app:welfsign}
\begin{proof}
\emph{(a)} For $\xi=\phi_t$, $\mathrm{d}\gamma_{t+1}/\mathrm{d}\phi_t=0$, so the
longevity channel vanishes. By Propositions~\ref{prop4} and~\ref{prop5}
(complementarity) $\mathrm{d}w_t/\mathrm{d}\phi_t>0$ and
$\mathrm{d}d_t/\mathrm{d}\phi_t>0$, and by Proposition~\ref{prop_robpref}
$\mathrm{d}R_{t+1}/\mathrm{d}\phi_t>0$. The coefficients
$\lambda_t h_t,\ \lambda_t,\ \lambda_t\gamma_{t+1}c_{t+1}^{r}/R_{t+1}^{2}$ are all
strictly positive, so every surviving term in Proposition~\ref{prop_welfdecomp} is
positive and $\mathrm{d}W_t/\mathrm{d}\phi_t>0$.
\emph{(b)} For the longevity shock $\mathrm{d}\gamma_{t+1}=1$, and the longevity
channel contributes $\beta\log c_{t+1}^{r}-\lambda_t c_{t+1}^{r}/R_{t+1}$. Capital
deepening raises the wage, $\mathrm{d}w_t/\mathrm{d}\gamma_{t+1}>0$, while
$\mathrm{d}R_{t+1}/\mathrm{d}\gamma_{t+1}<0$ (Proposition~\ref{prop_robpref}).
Placing the positive terms (value of a longer life and the wage gain) and the
negative terms (financing cost, return compression, net of any dividend change) on
opposite sides of Proposition~\ref{prop_welfdecomp} gives the stated
necessary-and-sufficient condition; the sign is not determined by theory alone.
\end{proof}

\subsection{Proof of Proposition~\ref{prop_dyneff}}
\label{app:dyneff}
\begin{proof}
On a balanced path aggregate capital and output grow at the gross rate of the
working-age population, $n_t$; the modified golden-rule (Cass--Diamond) comparison
is therefore between $R_t$ and $n_t$, and the equilibrium is dynamically efficient
iff $R_t\ge n_t$ \citep{diamond1965national}. Linearity of the margin gives
$\mathrm{d}(R_t-n_t)/\mathrm{d}\xi=\mathrm{d}R_t/\mathrm{d}\xi-\mathrm{d}n_t/\mathrm{d}\xi$.
For $\xi=\phi_t$, $\mathrm{d}R_t/\mathrm{d}\phi_t>0$
(Proposition~\ref{prop_robpref}) and $\mathrm{d}n_t/\mathrm{d}\phi_t\ge0$
(Proposition~\ref{prop4}); the difference is positive iff
$\mathrm{d}R_t/\mathrm{d}\phi_t>\mathrm{d}n_t/\mathrm{d}\phi_t$. For
$\xi=\gamma_{t+1}$, $\mathrm{d}R_t/\mathrm{d}\gamma_{t+1}<0$
(Proposition~\ref{prop_robpref}) and $\mathrm{d}n_t/\mathrm{d}\gamma_{t+1}<0$
(Proposition~\ref{proplonge}); the difference of two negative numbers is of
ambiguous sign. In both cases the saving-supply force lowers $R_t$ toward $n_t$,
which establishes the claim.
\end{proof}

\section{Additional robustness exercises}
\label{app:robust}

This appendix reports two further sensitivity exercises complementing the
$\sigma$ analysis of Section~\ref{sec:robustness}: variation of the
physical-capital share $\alpha$ (Appendix~\ref{app:robust_alpha}) and of
the persistence of the AI shock $\rho_{\phi}$
(Appendix~\ref{app:robust_rhophi}). Both figures follow the conventions of
Figure~\ref{fig:robust_sigma}.

\subsection{The physical-capital share}
\label{app:robust_alpha}

The capital share $\alpha$ jointly governs the steady-state factor income
distribution and, together with $\sigma$, the location of the
complementarity threshold $\sigma=1/\alpha$. We sweep
$\alpha \in \{0.25,\, 0.30,\, 0.33,\, 0.36,\, 0.40\}$, covering the
conventional range of estimates for advanced economies.

\begin{figure}[!htbp]
\centering
\includegraphics[width=\textwidth,height=0.85\textheight,keepaspectratio]%
                {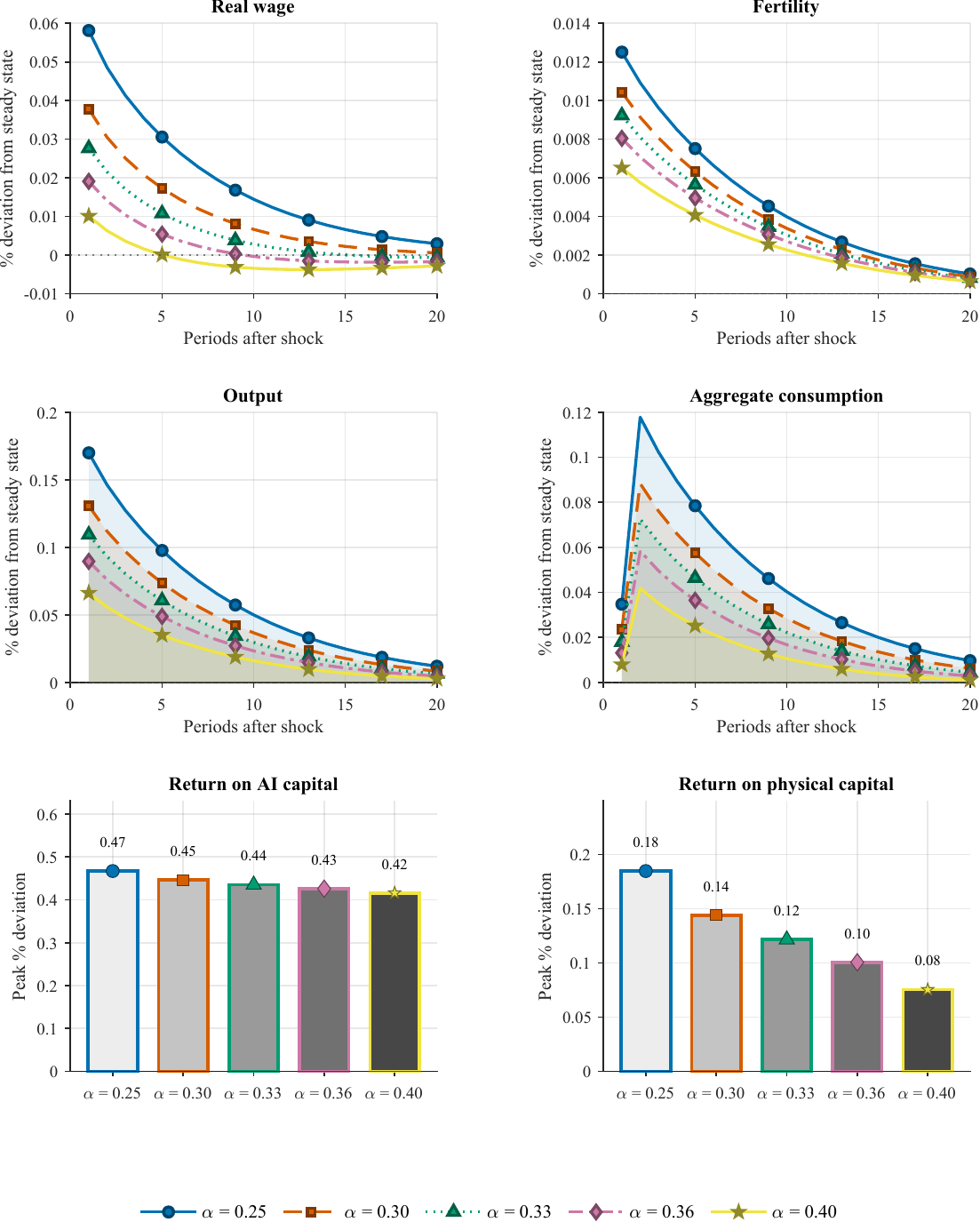}
\caption{Impulse responses to a positive AI shock for alternative values
of the physical-capital share $\alpha$.}
\label{fig:robust_alpha}
\end{figure}

The wage and fertility responses (Figure~\ref{fig:robust_alpha}) remain
positive at impact across the entire grid: with $\sigma=2$ held at its
baseline and the threshold $1/\alpha\in[2.5,\,4.0]$ exceeding 2 everywhere
on the grid, the economy remains in the complementarity regime. A lower
$\alpha$, corresponding to a larger labor share, amplifies the wage and
fertility responses, because the AI-induced productivity gain passes
through more strongly to labor when labor weighs more heavily in
production; a higher $\alpha$ dampens them, and at the upper end of the
grid the wage response turns marginally negative at longer horizons.
Output is likewise decreasing in $\alpha$ at impact, reflecting the
smaller labor base through which the AI-induced productivity gain is
amplified; the rental of AI capital is mildly decreasing in $\alpha$
because physical capital captures a larger share of that gain. None of
these comparative-static patterns alters the qualitative conclusions: the
results are insensitive to plausible variation in the capital share.

\subsection{The persistence of the AI shock}
\label{app:robust_rhophi}

The persistence parameter $\rho_{\phi}$ controls the half-life of the
productivity impulse. We sweep
$\rho_{\phi} \in \{0.50,\, 0.70,\, 0.85,\, 0.95,\, 0.99\}$, spanning
short-lived to highly persistent disturbances.

\begin{figure}[!htbp]
\centering
\includegraphics[width=\textwidth,height=0.85\textheight,keepaspectratio]%
                {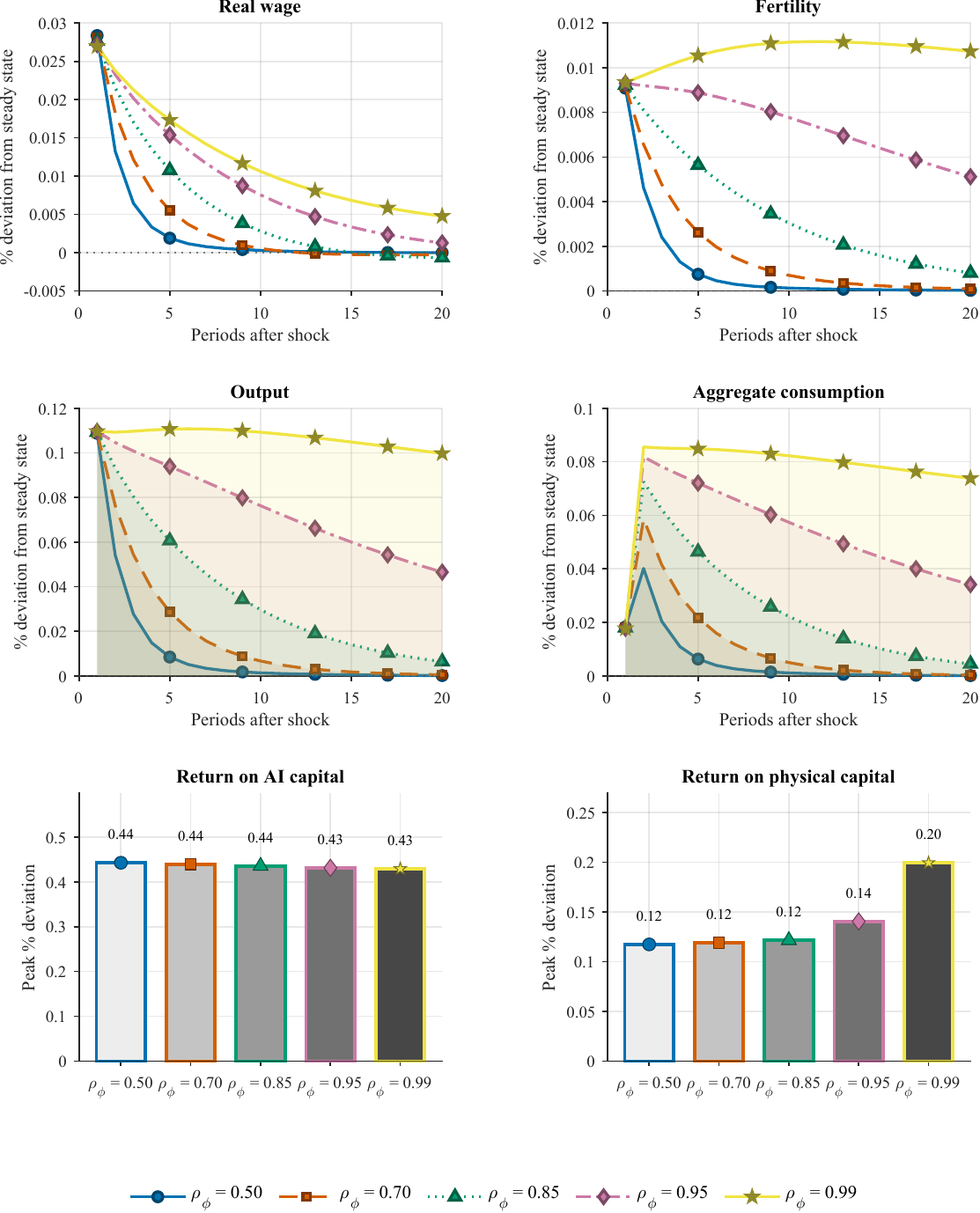}
\caption{Impulse responses to a positive AI shock for alternative values
of the shock persistence $\rho_{\phi}$.}
\label{fig:robust_rhophi}
\end{figure}

Persistence governs magnitudes and the speed of decay, not signs
(Figure~\ref{fig:robust_rhophi}). More persistent shocks generate larger
and longer-lived responses of wages, fertility, output, and consumption,
because a more durable productivity gain elicits a larger permanent-income
response and a stronger investment reaction. The impact jump in the AI
rental is essentially invariant to $\rho_{\phi}$, since it is governed by
the size of the shock rather than its persistence, whereas the peak
response of the physical-capital rental rises with $\rho_{\phi}$ as the
more durable impulse sustains a larger investment boom. As
$\rho_{\phi}\to 1$, the responses approach those of a permanent
technology shift, with wages and fertility settling at elevated levels
for an extended horizon. The qualitative profile, hump-shaped output and
wage paths, a sharp impact response of $R_t^{a}$, and a mild positive
fertility response, is robust across the entire grid.

\section{Tables}
\label{app:tables}
This appendix collects the second-order statistics. 
\begin{table}[H]
\centering
\small
\setlength{\tabcolsep}{5pt}
\renewcommand{\arraystretch}{1.0}
\begin{minipage}[t]{0.56\textwidth}
\centering
\captionof{table}{Selected theoretical moments: mean, standard deviation,
first-order autocorrelation $\rho(1)$, and contemporaneous correlation with
output.}
\label{tab:moments}
\begin{tabular}{lrrrr}
\toprule
Variable & Mean & Std.\ dev. & $\rho(1)$ & $\mathrm{corr}(\cdot,y)$ \\
\midrule
$y$ & 0.4313 & 0.0045 & 0.96 & 1.00 \\
$c$ & 0.2705 & 0.0028 & 0.91 & 0.94 \\
$c^r$ & 0.1385 & 0.0029 & 0.93 & 0.96 \\
$c^w$ & 0.1319 & 0.0004 & 0.65 & -0.32 \\
$s^w$ & 0.1178 & 0.0024 & 0.93 & 0.97 \\
$k$ & 0.0837 & 0.0014 & 0.98 & 0.97 \\
$a$ & 0.0341 & 0.0007 & 0.98 & 0.98 \\
$i^k$ & 0.0904 & 0.0006 & 0.99 & 0.93 \\
$i^a$ & 0.0382 & 0.0005 & 0.99 & 0.92 \\
$R$ & 2.4498 & 0.0151 & 0.88 & -0.78 \\
$q^k$ & 1.0000 & 0.0038 & 0.73 & 0.50 \\
$q^a$ & 1.0000 & 0.0067 & 0.80 & 0.63 \\
$\psi$ & 0.4650 & 0.0141 & 0.93 & 0.95 \\
$h$ & 0.6103 & 0.0024 & 0.90 & 0.92 \\
$n$ & 2.0000 & 0.0125 & 0.90 & -0.92 \\
$w$ & 0.3386 & 0.0016 & 0.98 & 0.98 \\
$R^k$ & 1.5298 & 0.0104 & 0.97 & -0.79 \\
$R^a$ & 1.5698 & 0.0133 & 0.85 & -0.27 \\
\bottomrule
\end{tabular}
\end{minipage}\hfill
\begin{minipage}[t]{0.40\textwidth}
\centering
\captionof{table}{Unconditional variance decomposition (percent of variance
attributable to each shock).}
\label{tab:vardec}
\begin{tabular}{lrr}
\toprule
Variable & $\varepsilon_\phi$ (AI) & $\varepsilon_\gamma$ \\
\midrule
$y$ & 4.21 & 95.79 \\
$c$ & 1.88 & 98.12 \\
$c^r$ & 1.21 & 98.79 \\
$c^w$ & 4.28 & 95.72 \\
$s^w$ & 0.10 & 99.90 \\
$k$ & 0.32 & 99.68 \\
$a$ & 2.06 & 97.94 \\
$i^k$ & 1.06 & 98.94 \\
$i^a$ & 7.53 & 92.47 \\
$R$ & 9.31 & 90.69 \\
$q^k$ & 0.84 & 99.16 \\
$q^a$ & 18.49 & 81.51 \\
$\psi$ & 0.00 & 100.00 \\
$h$ & 0.10 & 99.90 \\
$n$ & 0.10 & 99.90 \\
$w$ & 0.85 & 99.15 \\
$R^k$ & 18.76 & 81.24 \\
$R^a$ & 75.48 & 24.52 \\
\bottomrule
\end{tabular}
\end{minipage}
\end{table}

\begin{landscape}
\begin{table}[!ht]
\centering
\caption{Contemporaneous cross-correlations (lower triangle; the matrix is symmetric, so the upper triangle is omitted).}
\label{tab:corr}
\resizebox{\linewidth}{!}{%
\begin{tabular}{lrrrrrrrrrrrrrrrrrr}
\toprule
 & $y$ & $c$ & $c^r$ & $c^w$ & $s^w$ & $k$ & $a$ & $i^k$ & $i^a$ & $R$ & $q^k$ & $q^a$ & $\psi$ & $h$ & $n$ & $w$ & $R^k$ & $R^a$ \\
\midrule
$y$ & 1.00 &  &  &  &  &  &  &  &  &  &  &  &  &  &  &  &  &  \\
$c$ & 0.94 & 1.00 &  &  &  &  &  &  &  &  &  &  &  &  &  &  &  &  \\
$c^r$ & 0.96 & 0.99 & 1.00 &  &  &  &  &  &  &  &  &  &  &  &  &  &  &  \\
$c^w$ & -0.32 & -0.10 & -0.24 & 1.00 &  &  &  &  &  &  &  &  &  &  &  &  &  &  \\
$s^w$ & 0.97 & 0.87 & 0.93 & -0.53 & 1.00 &  &  &  &  &  &  &  &  &  &  &  &  &  \\
$k$ & 0.97 & 0.94 & 0.96 & -0.34 & 0.97 & 1.00 &  &  &  &  &  &  &  &  &  &  &  &  \\
$a$ & 0.98 & 0.95 & 0.96 & -0.21 & 0.93 & 0.97 & 1.00 &  &  &  &  &  &  &  &  &  &  &  \\
$i^k$ & 0.93 & 0.93 & 0.93 & -0.16 & 0.90 & 0.97 & 0.96 & 1.00 &  &  &  &  &  &  &  &  &  &  \\
$i^a$ & 0.92 & 0.91 & 0.89 & -0.01 & 0.82 & 0.88 & 0.97 & 0.92 & 1.00 &  &  &  &  &  &  &  &  &  \\
$R$ & -0.78 & -0.68 & -0.77 & 0.74 & -0.92 & -0.86 & -0.75 & -0.76 & -0.56 & 1.00 &  &  &  &  &  &  &  &  \\
$q^k$ & 0.50 & 0.30 & 0.43 & -0.97 & 0.68 & 0.50 & 0.38 & 0.31 & 0.18 & -0.82 & 1.00 &  &  &  &  &  &  &  \\
$q^a$ & 0.63 & 0.44 & 0.53 & -0.74 & 0.70 & 0.54 & 0.45 & 0.36 & 0.30 & -0.67 & 0.84 & 1.00 &  &  &  &  &  &  \\
$\psi$ & 0.95 & 0.86 & 0.92 & -0.58 & 1.00 & 0.96 & 0.90 & 0.87 & 0.77 & -0.94 & 0.72 & 0.72 & 1.00 &  &  &  &  &  \\
$h$ & 0.92 & 0.81 & 0.88 & -0.65 & 0.99 & 0.93 & 0.87 & 0.84 & 0.73 & -0.95 & 0.77 & 0.74 & 0.99 & 1.00 &  &  &  &  \\
$n$ & -0.92 & -0.81 & -0.88 & 0.65 & -0.99 & -0.93 & -0.87 & -0.84 & -0.73 & 0.95 & -0.77 & -0.74 & -0.99 & -1.00 & 1.00 &  &  &  \\
$w$ & 0.98 & 0.97 & 0.97 & -0.20 & 0.93 & 0.98 & 0.99 & 0.98 & 0.94 & -0.76 & 0.38 & 0.49 & 0.91 & 0.88 & -0.88 & 1.00 &  &  \\
$R^k$ & -0.79 & -0.81 & -0.84 & 0.32 & -0.84 & -0.92 & -0.83 & -0.92 & -0.70 & 0.87 & -0.43 & -0.32 & -0.85 & -0.83 & 0.83 & -0.86 & 1.00 &  \\
$R^a$ & -0.27 & -0.34 & -0.36 & 0.19 & -0.38 & -0.48 & -0.41 & -0.53 & -0.32 & 0.57 & -0.16 & 0.23 & -0.39 & -0.40 & 0.40 & -0.39 & 0.75 & 1.00 \\
\bottomrule
\end{tabular}}
\end{table}
\end{landscape}
\end{document}